\documentclass[11pt]{article}
\usepackage{graphicx} %
\usepackage{float}

\usepackage{xspace}
\usepackage{xcolor}
\usepackage{amscd}
\usepackage{amsmath}
\usepackage{amssymb}
\usepackage{amstext}
\usepackage{amsthm}
\usepackage{bbold}
\usepackage{bm}
\usepackage{colonequals}
\usepackage{tikz}

\usepackage{hyperref}

\usepackage[dvips,letterpaper,margin=1in]{geometry}

\usepackage{algpseudocode}
\usepackage[ruled,vlined]{algorithm2e}
\newcommand{\sA}{\mathcal{A}}

\newcommand{\sD}{\mathcal{D}}

\newcommand{\sG}{\mathcal{G}}
\newcommand{\sH}{\mathcal{H}}
\newcommand{\sI}{\mathcal{I}}

\newcommand{\sR}{\mathcal{R}}

\newcommand{\sT}{\mathcal{T}}

\newcommand{\DD}{\mathbb{D}}

\newcommand{\PP}{\mathbb{P}}
\newcommand{\QQ}{\mathbb{Q}}
\newcommand{\RR}{\mathbb{R}}

\newcommand{\ZZ}{\mathbb{Z}}
\DeclareSymbolFont{bbold}{U}{bbold}{m}{n}
\DeclareSymbolFontAlphabet{\mathbbold}{bbold}

\newcommand{\Erdos}{Erd\H{o}s}
\newcommand{\Renyi}{R\'{e}nyi}

\newcommand{\Col}{\mathrm{Col}}

\newcommand{\Var}{\mathrm{Var}}
\newcommand{\poly}{\mathrm{poly}}
\newcommand{\tw}{\mathrm{tw}}

\newcommand{\Inj}{\mathrm{Inj}}
\newcommand{\Aut}{\mathrm{Aut}}
\newcommand{\Iso}{\mathrm{Iso}}
\newcommand{\conn}{\mathrm{conn}}
\newcommand{\width}{\mathrm{width}}

\newcommand{\DP}{\mathrm{DP}}
\DeclareMathOperator*{\E}{\mathbb{E}}
\newcommand{\1}{\mathbf{1}}

\newcommand{\Cov}{\mathrm{Cov}}

\newcommand{\Bin}{\mathrm{Bin}}
\newtheorem{claim}{Claim}
\newtheorem{theorem}{Theorem}[section]
\newtheorem{remark}[theorem]{Remark}
\newtheorem{assumption}[theorem]{Assumption}
\newtheorem{lemma}[theorem]{Lemma}
\newtheorem{question}[theorem]{Question}

\newtheorem{definition}[theorem]{Definition}

\newtheorem{proposition}[theorem]{Proposition}
\newtheorem{corollary}[theorem]{Corollary}
\newtheorem{conjecture}[theorem]{Conjecture}

\renewcommand{\emptyset}{\varnothing}

\usepackage{pgfplots}
\pgfplotsset{compat=1.16}

\title{Improved polynomial-time algorithms for detecting and recovering planted $\Theta(\sqrt{n})$-cliques}
\usepackage{authblk}
\author[1]{Dmitriy Kunisky\thanks{Email: \texttt{kunisky@jhu.edu}}}
\author[2]{Songtao Mao\thanks{Email: \texttt{smao13@jhu.edu}.}}
\affil[1]{Department of Applied Mathematics \& Statistics, Johns Hopkins University}
\affil[2]{Department of Computer Science, Johns Hopkins University}

\date{September 21, 2026}

\begin{document}
\maketitle
\thispagestyle{empty}

\begin{abstract}
In the planted clique problem, one observes either an \Erdos--\Renyi\ graph on $n$ vertices or such a graph with a clique added to $k = k(n)$ vertices, and seeks to detect or recover the clique.
It is widely believed that $k = \Theta(\sqrt{n})$ is the smallest clique size for which polynomial-time algorithms exist for these tasks.
We develop new algorithms in this regime using color-coding to estimate signed subgraph counts, further accelerated with fast matrix multiplication.

We first show that, for each $t \geq 1$, for $c(t)$ a constant associated to the order of growth of the number of connected graphs of treewidth at most $t$, cliques of size $k = \lambda\sqrt{n}$ planted in a random location with $\lambda > 1 / \sqrt{c(t)}$ can be detected and recovered in time $n^{t + 1 + o(1)}$.
For instance, since $c(1) = e$, this recovers by counting signed trees the performance of the $\widetilde{O}(n^2)$-time message-passing algorithm of Deshpande--Montanari~(2015) that succeeds when $\lambda > 1 / \sqrt{e} \approx 0.6066$. For $t \geq 3$, the exact value of $c(t)$ is not known, but lower bounds on it give a hierarchy of slower polynomial-time algorithms that succeed for smaller $\lambda$.

We further show that the above algorithm for $t = 2$ can be implemented in time $n^{\omega + o(1)}$ for $\omega$ the constant of square matrix multiplication and succeeds when $\lambda > 0.3320$; under the folklore conjecture that $\omega = 2$, this runs in the nearly-linear time of the algorithm of Deshpande--Montanari while finding smaller cliques.
Second, we show that the above algorithm for $t = 1$ can be combined with the boosting scheme of Alon--Krivelevich--Sudakov (1998) using rectangular matrix multiplication, giving improved runtimes for smaller $\lambda$.
Taken together, our results achieve the best known tradeoff between runtime and signal strength $\lambda$.
\end{abstract}
\newpage

\thispagestyle{empty}
\tableofcontents
\thispagestyle{empty}

\newpage

\section{Introduction}
\pagenumbering{arabic}

\subsection{Planted clique problem}

The planted clique problem is one of the most influential problems in average-case complexity theory, random graph algorithms, and high-dimensional statistics, and it has become a central benchmark for understanding statistical--computational tradeoffs. We write $\sG(n, 1/2, C)$ for the distribution involved in this problem, the law of a random graph $G$ formed by first drawing an Erd\H{o}s--R\'enyi graph $G \sim \sG(n,1/2)$, and then adding a clique on a subset $C \subseteq [n]$ of vertices.
We further write $\sG(n, 1/2, k)$ for the case where $C$ is drawn uniformly at random among subsets of size $k$.
We focus on the latter case for now and will return to enumerating over specific deterministic cliques $C$ later.
The task associated to this model is then either to \emph{detect} the presence of the planted clique, i.e., to test whether a graph $G$ was drawn from $\sG(n, 1/2, k)$ or from $\sG(n, 1/2)$, or to \emph{recover} the clique vertices when $G \sim \sG(n, 1/2, k)$. We think of $k = k(n)$ and consider a sequence of such problems as $n \to \infty$.

Since the size of the largest clique in $G \sim \sG(n, 1/2)$ is with high probability in the interval $[(2 - \varepsilon)\log_2 n, (2 + \varepsilon)\log_2 n]$ for any $\varepsilon > 0$ (see \cite[Section 7.2]{FK-2016-RandomGraphs}), both tasks are possible by brute force search once $k \geq (2 + \varepsilon)\log_2 n$. In contrast, no polynomial-time algorithm is known to succeed until $k = \Omega(\sqrt{n})$. This leaves a broad regime,
\[
(2+\varepsilon)\log_2 n \leq k \ll \sqrt n,
\]
in which detection and recovery are possible but widely conjectured to be computationally intractable. Various forms of the \emph{planted clique conjecture} give precise forms of this hardness; a slightly weaker and widely believed conjecture is that, for instance, if $k \leq n^{1/2 - \varepsilon}$ for some $\varepsilon > 0$ then no polynomial-time algorithm can detect or recover the planted clique (which we state below in Conjecture~\ref{conj:clique}).

When $k=\Omega(\sqrt n)$, various polynomial-time algorithms are known for detecting or recovering planted cliques:
\begin{itemize}
    \item \textbf{Degree thresholding:} In \cite{kuvcera1995expected}, Kučera observes that when $k\ge C\sqrt{n\log n}$ for a constant $C > 0$, then the $k$ vertices of highest degree in $G \sim \sG(n, 1/2, k)$ are outliers of the degree distribution and are with high probability the planted clique vertices, giving $O(n^2)$-time algorithms for detection and recovery.
    \item \textbf{Spectral methods:} In \cite{alon1998finding}, Alon, Krivelevich, and Sudakov use the second eigenvector of the adjacency matrix to recover a planted clique of size $k \ge 10\sqrt{n}$ in time $O(n^3)$. Standard random matrix theory implies that detection can be achieved by examining the largest eigenvalue of the adjacency matrix for the broader range $k \geq \sqrt{n}$ (as presented in, for example, \cite[Section 1]{lugosi2017lectures}), and a simple modified spectral algorithm of \cite{ma2025nonlinear} is conjectured to achieve detection when $k \geq 0.76\sqrt{n}$.

    \item \textbf{Boosting:} In \cite{alon1998finding} it is also pointed out that one may run the above algorithm on all induced subgraphs of common neighbors of cliques of some small size $s$ (some of which will be subsets of the larger planted clique). If $t = 2\lceil \log_2(10 / \lambda) \rceil + 2$, then this method can detect cliques of size $k \geq \lambda\sqrt{n}$ in time $O(n^{t + 3})$. The same method may also be used to ``boost'' the performance of other algorithms by combining either with brute force search over $s$-subsets or with slightly subtler choices of subsets that give faster algorithms, as we describe in Section~\ref{sec:prior}.
    \item \textbf{Convex optimization:} Appropriate semidefinite programs can also find cliques of size $k = \Omega(\sqrt{n})$, as shown by \cite{feige2000finding}, though without an explicit constant $k \geq C\sqrt{n}$. Other convex optimization methods were studied by \cite{ames2011nuclear,ames2015guaranteed}. Compared to other algorithms, the current analyses of these algorithms seem to require $k \geq C\sqrt{n}$ for $C$ large and sometimes to have polynomial runtime with a large exponent in order to use general-purpose convex optimization algorithms.
    \item \textbf{Markov chain Monte Carlo and gradient descent:} While, as we mention below, the natural Metropolis dynamics Markov chain does \emph{not} recover planted cliques of size $k = \Theta(\sqrt{n})$, \cite{GJX-2023-PlantedCliqueMCMC} show that a variant thereof as well as gradient descent on a related objective function does find cliques of size $k \geq C\sqrt{n}$ in polynomial time for some $C$, but again this constant is not given explicitly.
    \item \textbf{Combinatorial algorithms:} In \cite{feige2010finding}, the authors propose a purely combinatorial ``pruning'' algorithm that runs in time $O(n^2)$ and, for $k \geq C\sqrt n$ with $C$ sufficiently large, finds the hidden clique with probability at least $2/3$. A similar $O(n^2)$-time algorithm of \cite{dekel2014finding} succeeds with an explicit constant $k \geq 1.26\sqrt n$ and with high probability.
    \item \textbf{Message-passing:} In \cite{deshpande2015finding}, Deshpande and Montanari develop a belief propagation algorithm and analysis that achieves quite strong constants: it recovers cliques of size $k \geq \sqrt{n / e} \approx 0.61\sqrt{n}$ and runs in $O(n^2\log n)$ time. This is (to our knowledge) the smallest clique size known to be recoverable by nearly-linear time algorithms.\footnote{``Nearly linear'' in this context is relative to the $\Theta(n^2)$ size of the input.}
\end{itemize}

Despite decades of effort, no polynomial-time algorithm is known for planted clique in the regime $k=o(\sqrt n)$, and multiple lines of evidence support the conjecture that none exists.
One important line of evidence comes from restricted computational models: Metropolis dynamics~\cite{Jerrum-1992-LargeCliques,CMZ-2023-LinearPlantedCliquesMetropolis}, algorithms in the statistical query~(SQ) computational model~\cite{feldman2017statistical}, low-degree polynomial algorithms~\cite[Section~2.4]{hopkins2018statistical}, and sum-of-squares semidefinite programming relaxations~\cite{MPW-2015-PlantedClique,DM-2015-SOSPlantedClique,HKPRS-2018-PlantedCliqueSOS4,barak2019nearly}.\footnote{The evidence from sum-of-squares relaxations is slightly weaker; the best known lower bounds for the degree $d$ relaxation that runs in time $n^{O(d)}$ show that it cannot find cliques of size $k\le n^{1/2-c\sqrt{d/\log n}}$, which does not quite rule out the entire regime $k = o(\sqrt{n})$.} The planted clique conjecture has thus become a cornerstone assumption in average-case complexity; assuming its truth, researchers have reduced many other problems to planted clique or variants thereof, including sparse principal component analysis~\cite{berthet2013computational}, submatrix detection~\cite{ma2015computational,brennan2019universality}, detection of planted dense subgraphs and other notions of hidden community structure in graphs~\cite{brennan2020reducibility,bresler2023detection}, and other clustering-type problems~\cite{chen2016statistical,abbe2018community}.

\subsection{Fine-grained polynomial-time complexity}

Many of the above works have, implicitly or explicitly, studied the tradeoff between signal strength and computational complexity in the $k = \Omega(\sqrt{n})$ regime: if we restrict ourselves to algorithms running in time $O(n^{\tau})$, then what size $k \geq \lambda\sqrt{n}$ of planted cliques can we detect? As for instance the boosting scheme of \cite{alon1998finding} demonstrates, there is in fact a tradeoff: paying more runtime, that is, allowing larger $\tau$, lets us handle smaller constants $\lambda$. While this has been studied in special cases, it seems that the general question of determining this tradeoff has not been posed explicitly before, which we do below.

We will see that there are some subtle issues surrounding the presence of randomness in two places: the location of the clique and the design of an algorithm.
Note that, for randomized algorithms, the distinction between a uniformly random clique location $G \sim \sG(n, 1/2, k)$ and an arbitrary fixed clique location $G \sim \sG(n, 1/2, C)$ is immaterial.
Indeed, a randomized algorithm may first apply a uniformly random permutation $\pi$ to the input vertices, reducing the latter to the former.

For deterministic algorithms, however, this randomization is not allowed and so the choice of model of clique location matters more. Thus we distinguish three tradeoff functions: (1) where the algorithm is randomized and the clique location is uniformly random (to which deterministic clique location can be reduced), (2) where the algorithm is deterministic and the clique location is uniformly random, and (3) where the algorithm is deterministic and the clique location is deterministic and worst-case (i.e., the algorithm must succeed with high probability for any location of planted clique).
We will give an important example below related to the boosting scheme of~\cite{alon1998finding} that illustrates how these models can allow for algorithms with different performance.

\begin{definition}\label{def:tau-variants}
For $\lambda>0$, define $\tau_{\mathrm{rand}}(\lambda)$, $\tau_{\det,\mathrm{avg}}(\lambda)$, and $\tau_{\det}(\lambda)$ as follows. In each case, the value is the infimum over all $\alpha>0$ such that, for every $\varepsilon>0$, there are algorithms running in time $O(n^{\alpha+\varepsilon})$ that achieve strong recovery and strong detection (in the sense of Definition~\ref{def:detection-recovery}) in a planted clique problem whenever $k\ge(\lambda+\varepsilon)\sqrt n$.
The three variants differ only in the allowed algorithm and in the model of clique location:
\begin{itemize}
\item $\tau_{\mathrm{rand}}$: the algorithm may be randomized and succeeds for $G \sim \sG(n, 1/2, k)$.
\item $\tau_{\det, \mathrm{avg}}$: the algorithm is deterministic and succeeds for $G \sim \sG(n,1/2,k)$.
\item $\tau_{\det}$: the algorithm is deterministic and, for every fixed sequence of clique locations $C_n\in\binom{[n]}{k}$, succeeds for $G \sim \sG(n, 1/2, C_n)$.
\end{itemize}
Lastly, when no ambiguity is possible, we write $\tau=\tau_{\mathrm{rand}}$. Since strong recovery also yields strong detection, it suffices throughout to establish strong recovery.
\end{definition}

\noindent
Clearly,
\[
\tau_{\mathrm{rand}}(\lambda)
\le
\tau_{\det,\mathrm{avg}}(\lambda)
\le
\tau_{\det}(\lambda).
\]
We allow for an ``$\varepsilon$ of room'' in the signal strength and runtime to exclude subtleties about factors of sub-leading order in $n$ in either quantity. This gives a precise formulation of the smooth tradeoff between required runtime and signal strength in the planted clique problem. We therefore believe the following is a fundamental question about the planted clique problem.

\begin{question}
\label{question}
Determine $\tau_{\mathrm{rand}}(\lambda)$, $\tau_{\det,\mathrm{avg}}(\lambda)$, and $\tau_{\det}(\lambda)$ for every $\lambda>0$.
\end{question}
\noindent
Figure~\ref{fig:four-envelope} illustrates our best bounds on these functions compared to the best (to our knowledge) implied by prior work.

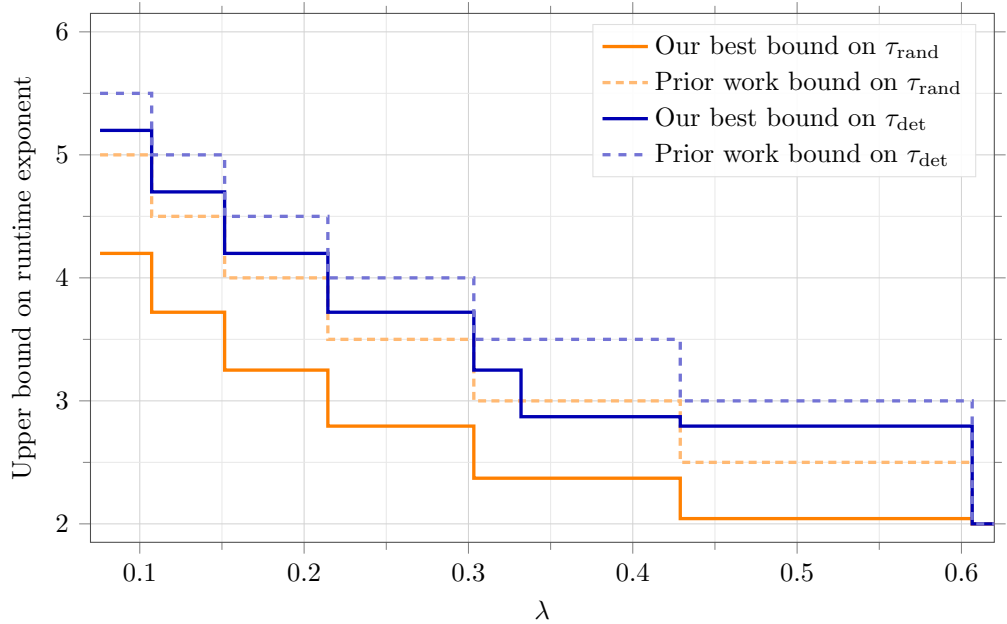
\begin{figure}[t]

\centering
\begin{tikzpicture}
\begin{axis}[
    width=0.82\textwidth,
    height=0.52\textwidth,
    xmin=0.07, xmax=0.62,
    ymin=1.85, ymax=6.15,
    xlabel={$\lambda$},
    ylabel={Upper bound on runtime exponent},
    xtick={0.1,0.2,0.3,0.4,0.5,0.6},
    ytick={2,3,4,5,6},
    minor x tick num=1,
    minor y tick num=1,
    grid=both,
    major grid style={draw=gray!35, line width=0.25pt},
    minor grid style={draw=gray!18, line width=0.15pt},
    tick align=outside,
    axis line style={black!70},
    ticklabel style={font=\small},
    label style={font=\small},
    legend style={
        draw=gray!25,
        fill=white,
        fill opacity=0.92,
        text opacity=1,
        font=\small,
        at={(0.98,0.98)},
        anchor=north east,
        row sep=1pt
    },
    legend cell align={left},
    every axis plot/.append style={line width=1.25pt}
]

\addplot+[
    orange,
    const plot,
    mark=none
]
coordinates {
    ({1/sqrt(2^6*exp(1))},4.1989)
    ({1/sqrt(2^5*exp(1))},3.7205)
    ({1/sqrt(2^4*exp(1))},3.2501)
    ({1/sqrt(2^3*exp(1))},2.7947)
    ({1/sqrt(2^2*exp(1))},2.3712)
    (0.3320,2.3712)
    ({1/sqrt(2^1*exp(1))},2.0428)
    ({1/sqrt(exp(1))},2.0000)
    (0.6200,2.0000)
};

\addplot+[
    orange!55,
    densely dashed,
    const plot,
    mark=none
]
coordinates {
    ({1/sqrt(2^6*exp(1))},5.0000)
    ({1/sqrt(2^5*exp(1))},4.5000)
    ({1/sqrt(2^4*exp(1))},4.0000)
    ({1/sqrt(2^3*exp(1))},3.5000)
    ({1/sqrt(2^2*exp(1))},3.0000)
    (0.3320,3.0000)
    ({1/sqrt(2^1*exp(1))},2.5000)
    ({1/sqrt(exp(1))},2.0000)
    (0.6200,2.0000)
};

\addplot+[
    blue!70!black,
    const plot,
    mark=none
]
coordinates {
    ({1/sqrt(2^6*exp(1))},5.1989)
    ({1/sqrt(2^5*exp(1))},4.6989)
    ({1/sqrt(2^4*exp(1))},4.1989)
    ({1/sqrt(2^3*exp(1))},3.7205)
    ({1/sqrt(2^2*exp(1))},3.2501)
    (0.3320,2.8712)
    ({1/sqrt(2^1*exp(1))},2.7947)
    ({1/sqrt(exp(1))},2.0000)
    (0.6200,2.0000)
};

\addplot+[
    blue!70!black!55!white,
    dashed,
    const plot,
    mark=none
]
coordinates {
    ({1/sqrt(2^6*exp(1))},5.5000)
    ({1/sqrt(2^5*exp(1))},5.0000)
    ({1/sqrt(2^4*exp(1))},4.5000)
    ({1/sqrt(2^3*exp(1))},4.0000)
    ({1/sqrt(2^2*exp(1))},3.5000)
    (0.3320,3.5000)
    ({1/sqrt(2^1*exp(1))},3.0000)
    ({1/sqrt(exp(1))},2.0000)
    (0.6200,2.0000)
};

\legend{
    Our best bound on $\tau_{\mathrm{rand}}$,
    Prior work bound on $\tau_{\mathrm{rand}}$,
    Our best bound on $\tau_{\det}$,
    Prior work bound on $\tau_{\det}$
}

\end{axis}
\end{tikzpicture}
\vspace{-1em}
\caption{Comparison of the best bounds on $\tau_{\mathrm{rand}}$ and $\tau_{\det}$ obtained from our results (the minimum of all applicable bounds for each $\lambda$) and those of prior work as given in \eqref{eq:tau-rand-aks} and \eqref{eq:tau-det-aks}.}
\label{fig:four-envelope}
\end{figure}

\begin{table}[t]
\centering
\small
\vspace{1em}
\begin{tabular}{c||c|c||c|c}
\hline
$\lambda$ &
Our bound: $\tau_{\mathrm{rand}}$ &
Prior work: $\tau_{\mathrm{rand}}$ &
Our bound: $\tau_{\det}$ &
Prior work: $\tau_{\det}$ \\
\hline
0.6066 & 2.0000 & 2.0000 & 2.0000 & 2.0000 \\
0.4289 & 2.0428 & 2.5000 & 2.7947 & 3.0000 \\
0.3320 & 2.3712 & 3.0000 & 2.8712 & 3.5000 \\
0.3033 & 2.3712 & 3.0000 & 3.2501 & 3.5000 \\
0.2145 & 2.7947 & 3.5000 & 3.7205 & 4.0000 \\
0.1517 & 3.2501 & 4.0000 & 4.1989 & 4.5000 \\
0.1073 & 3.7205 & 4.5000 & 4.6989 & 5.0000 \\
0.0759 & 4.1989 & 5.0000 & 5.1989 & 5.5000 \\
\hline
\end{tabular}
\caption{We give the numerical values of point bounds on the various $\tau$ functions that produce the step functions plotted in Figure~\ref{fig:four-envelope}.}
\label{tab:figure-values}
\end{table}

\subsection{Summary of consequences of prior results}
\label{sec:prior}

Several prior algorithmic ideas give basic upper bounds on these tradeoff functions.
We review below the best bounds we have been able to deduce based on these results, sometimes augmented with small additional arguments.

\paragraph{Message passing and nearly-linear time algorithms}
The algorithm of~\cite{deshpande2015finding} gives
\[
\tau_{\det}(e^{-1/2})\le 2,
\]
and hence also the same bound for $\tau_{\det,\mathrm{avg}}$ and $\tau_{\mathrm{rand}}$.
This is both the best point bound on these functions whose right-hand side is 2 (corresponding to a nearly-linear time algorithm) and the best to use in the boosting arguments we discuss below.
In fact, \cite{deshpande2015finding} speculate whether this might be the smallest size of clique detectable in nearly-linear time.
Translated to our notation, they write:
\begin{quote}
    ``...it is natural to ask whether the threshold $\sqrt{n / e}$ has a fundamental computational meaning or is instead only relevant for our specific algorithm.''
\end{quote}
We will show below that, under the conjecture that the constant $\omega$ of square matrix multiplication equals 2, it is in fact the latter that holds and that smaller cliques can still be detected in nearly-linear time.

\paragraph{Boosting}
The boosting scheme of \cite{alon1998finding} gives more general relationships between these values for different choices of $\lambda$.
Actually, the original formulation admits a simple improvement, and this leads to different costs depending on the choice of clique location model.
The original premise is to do a brute force search over small subsets of vertices of some size $s$, and repeatedly apply a planted clique algorithm on the subgraphs induced on simultaneous neighbors of all members of the set.
For some choice of such a subset in $\binom{[n]}{s}$, all vertices will lie in the planted clique, in which case the induced subgraph will be a smaller planted clique instance with a relatively larger clique.
We refer to these subsets as \emph{seed} sets of vertices whenever such methods are being applied.

While \emph{a priori} this multiplies the runtime by $\binom{n}{s} \leq n^s$, this factor can be improved.
If the algorithm can be randomized, one can choose random sets.  For every $k\ge\lambda\sqrt n$, a random $s$-subset falls in the planted clique with probability at least $\Omega_{s,\lambda}(n^{-s/2})$, so $O(n^{s/2 + o(1)})$ random draws suffice to, with high probability, include a subset contained in the clique.
This gives:
\[
\tau_{\mathrm{rand}}(2^{-s/2} \lambda)
\le
\max\{\tau_{\mathrm{rand}}(\lambda),2\}+\frac{s}{2}
\]
for every fixed integer $s\ge0$ and all $\lambda > 0$.
One may apply this to any ``base'' planted clique algorithm bounding a single $\tau_{\mathrm{rand}}(\lambda)$ to get a sequence of bounds for smaller $\lambda$.
To the best of our knowledge, the best possible application of this combination of ideas based on prior work is to apply this to the message-passing algorithm of \cite{deshpande2015finding}, which gives, for all $\lambda > 0$,
\begin{equation}
\tau_{\mathrm{rand}}(\lambda)
\le
2+\frac12\left\lceil \log_2\frac{1}{e\lambda^2}\right\rceil_+ \sim \log_2 \frac{1}{\lambda}, \label{eq:tau-rand-aks}
\end{equation}
the latter asymptotic applying to the limit $\lambda \to 0$ of small planted cliques.
The same argument and bound apply to $\tau_{\det, \mathrm{avg}}$ as well: one may draw and fix a list of $n^{s/2 + o(1)}$ random $s$-subsets in advance, and a uniformly random planted clique will cover some subset in this list with high probability by the same argument.

For $\tau_{\det}$, however, to apply this strategy we must fix a deterministic list of $s$-subsets that must ``hit'' every $k$-subset (i.e., every possible clique location).
Using all $s$-subsets as in the naive application of the strategy above gives
\begin{align*}
\tau_{\det}(2^{-s/2}\lambda)
&\le
\max\{\tau_{\det}(\lambda),2\}+s, \\
\tau_{\det}(\lambda) &\leq 2+\left\lceil \log_2\frac{1}{e\lambda^2}\right\rceil_+ \sim 2 \log_2 \frac{1}{\lambda},
\end{align*}
asymptotically larger than the bound on $\tau_{\mathrm{rand}}$ by a factor of 2.
This can be partly but not entirely repaired by a standard Tur\'an-type block construction.  For a fixed integer $s\ge2$, let
$b=\lfloor (k-1)/(s-1)\rfloor$, partition $[n]$ into $b$ nearly equal parts, and list every $s$-subset contained in one part (for $s=1$, list all singletons).  If a $k$-set contained no listed seed subset, then each part would contain at most $s-1$ of its vertices, for a total of at most $b(s-1)\le k-1$, a contradiction.  The list has size
$O_s(b(n/b)^s)=O_s(n^s/k^{s-1})\le n^{(s+1)/2+o(1)}$ whenever $k\ge\lambda\sqrt n$.
This gives
\begin{align}
\tau_{\det}(2^{-s/2}\lambda)
&\le
\max\{\tau_{\det}(\lambda),2\}+ \frac{s + 1}{2}, \nonumber \\
\tau_{\det}(\lambda) &\leq \frac{5}{2} + \frac12\left\lceil \log_2\frac{1}{e\lambda^2}\right\rceil_+ \sim \log_2 \frac{1}{\lambda}, \label{eq:tau-det-aks}
\end{align}
which recovers the asymptotic behavior of $\tau_{\mathrm{rand}}$ but with a worse next-order term (of size $\Theta(1)$ above).
The lines giving results of prior work in Figure~\ref{fig:four-envelope} correspond to applying these boosting bounds starting with the algorithm of \cite{deshpande2015finding}.

\subsection{Main results}

Before stating our main results, we recall the standard notions of detection and recovery for the planted clique problem.
We switch notation to work with more standard general notation for such settings from the literature.

\begin{definition}[Planted clique model]
    We write $\QQ = \QQ_n = \sG(n, 1/2)$ for the law of an \Erdos-\Renyi\ random graph on $n$ vertices, $\PP_{n, C}$ for the planted clique model with a specific clique $C \subseteq [n]$, and $\PP_{n, k} = \sG(n, 1/2, k)$ for the planted clique model with a uniformly random clique $C \sim \mathrm{Unif}(\binom{[n]}{k})$.
    We denote the clique by $C$ in both of the latter cases, view it as a random variable defined implicitly whenever we speak of the $\PP$ laws, and abbreviate both models as $\PP = \PP_n$ to allow the cases of a deterministic clique or a random one to be covered.
\end{definition}
\noindent We note that both of the above distributions can be viewed, by constructing suitable adjacency matrices, as distributions on $\{0, 1\}^{\binom{n}{2}}$ or $\{\pm 1\}^{\binom{n}{2}}$.

\begin{definition}[Strong detection and recovery]
\label{def:detection-recovery}
We call functions $\sA = \sA_n: \{0, 1\}^{\binom{n}{2}} \to \{0, 1\}$ a \emph{test}. A test \emph{(strongly) distinguishes} or achieves \emph{(strong) detection} between $\PP_{n}$ and $\QQ_n$ if
\[
\Pr_{G\sim \QQ_n}[\sA(G)=1]
+
\Pr_{G\sim \PP_{n}}[\sA(G)=0]
=
o(1).
\]
We call a sequence of functions $\sR = \sR_n: \{0,1\}^{\binom{n}{2}} \to 2^{[n]}$ a \emph{recovery algorithm} or \emph{estimator}. A recovery algorithm $\sR:\{0,1\}^{\binom{n}{2}} \to 2^{[n]}$ \emph{(strongly) recovers} the planted clique if
\[
\Pr_{G\sim \PP_{n}}\bigl[\sR(G)=C\bigr]
=
1-o(1).
\]
\end{definition}

\begin{remark}
    \label{rem:recovery-detection}
    We note that an algorithm for strong recovery also achieves strong detection provided $k \geq (2 + \varepsilon)\log_2 n$: we can simply check whether the output subset is indeed a large clique in the input graph, which with high probability will successfully test whether $G \sim \PP_{n}$ or $G \sim \QQ_n$.
\end{remark}

We first establish a connection between the fine-grained polynomial-time complexity of detecting and recovering planted cliques on the one hand, and the enumeration of bounded-treewidth graphs on the other by testing using a statistic calculating signed counts of such graphs.
See Definition~\ref{def:treewidth} for the definition of treewidth of a graph $G$, which we denote $\tw(G)$. For reference, the connected graphs of treewidth 1 are the trees, while the connected graphs of treewidth at most 2 can be viewed as trees assembled from series-parallel graphs. For integers $v\ge 1$ and $t\ge 1$, let
\[
\sG_{v,\le t}^{\mathrm{conn}} \colonequals \{\text{simple graphs } G = ([v], E): G \text{ connected}, \tw(G) \leq t\}.
\]
Here we emphasize that we view $G$ as having its vertices labelled by $[v]$, and we consider $G$ distinct up to labelled isomorphism.

For any constant $t$, the leading order behavior of $|\sG_{v,\leq t}^{\mathrm{conn}}|$ is of order $v!$ (or equivalently $v^v$). The following describes the second-order exponential growth, whose base depends on $t$:
\begin{definition}
For each integer $t\ge 1$, define
\[
\log c(t) \colonequals
\liminf_{v \to \infty} \frac{1}{v} \log\left(\frac{|\sG_{v,\le t}^{\mathrm{conn}}|}{v!}\right).
\]
\end{definition}
\noindent In words, $c(t)$ is the largest constant such that, for all $\varepsilon > 0$, for all sufficiently large $v$, we have $|\sG^{\mathrm{conn}}_{v, \leq t}| \geq (c(t) - \varepsilon)^v v!$. It follows from the results of \cite{baste2017number} that $0 < c(t) < \infty$ for all $t \geq 1$ (the upper bound also follows from the abstract results of \cite{norine2006proper}).

By considering signed subgraph counts of these graphs, we then obtain the following; the random clique location part is proved in Sections~\ref{sec:approx-detection} and~\ref{sec:derandomization}, and the worst-case deterministic location part is proved in Theorem~\ref{thm:universal-deterministic-exact-recovery}.

\begin{theorem}
\label{thm:intro-main-general}
\label{thm:main-bounded-treewidth}
Let $t\ge 1$ be an integer.
Then, we have
\begin{align*}
\tau_{\mathrm{rand}}\!\left(c(t)^{-1/2}\right)
&\le
\tau_{\det,\mathrm{avg}}\!\left(c(t)^{-1/2}\right)
\le t+1,\\
\tau_{\det}\!\left(c(t)^{-1/2}\right)
&\le t+\frac32.
\end{align*}
Moreover, in the worst-case deterministic clique location model, strong detection alone can be achieved in time $\widetilde O(n^{t+1})$.
\end{theorem}

\begin{remark}
In the worst-case deterministic model, the detection algorithm has exponent $t+1$. The additional $1/2$ in the displayed bound on $\tau_{\det}$ comes only from the deterministic splitting of a holdout set used in the exact recovery algorithm. For brevity, we omit the corresponding bounds in the next several results; they follow the same pattern.
\end{remark}
\noindent
Thus, understanding the quantity $c(t)$ translates directly into upper bounds on the three tradeoff functions.
By plugging into this result what is known about $c(t)$ from the combinatorics literature, we obtain the following consequences.

By Cayley's theorem~\cite{Cayley-1889-CountingTrees} counting labeled trees, we have $c(1)=e\approx 2.7183$, and so we find the following (ignoring here for the sake of simplicity the above outlier case of deterministic algorithms and worst-case clique locations).

\begin{corollary}
$\tau_{\mathrm{rand}}(e^{-1/2}) \leq \tau_{\mathrm{det},\mathrm{avg}}(e^{-1/2}) \le 2$.
\end{corollary}
\noindent
This matches, up to lower-order factors, the threshold and runtime achieved by the deterministic message-passing algorithm of Deshpande and Montanari~\cite{deshpande2015finding}.
Indeed, there is a folklore understanding prior to our work that, after unfolding the message-passing iterations, the algorithm of~\cite{deshpande2015finding} amounts to ``counting signed trees'' in a suitable sense.
Similarly, the value $c(2) \approx 9.073$ is described, albeit without a closed-form expression, by \cite{bodirsky2007enumeration}, and leads to an analogous bound.
\begin{corollary}
$\tau_{\mathrm{rand}}(0.3320) \leq \tau_{\mathrm{det},\mathrm{avg}}(0.3320) \le 3$.
\end{corollary}

While no other value of $c(t)$ for $t \geq 3$ is known exactly (to the best of our knowledge), the asymptotics of $c(t)$ as $t\to\infty$ are again well understood up to a logarithmic factor. In particular, \cite{baste2017number} proves that
\[
\frac{1}{128}\cdot \frac{1}{\log t}\cdot t2^t
\leq c(t)
\leq e\cdot t2^t
\qquad \text{for all } t\ge2,
\]
and also suggests that the upper bound is likely to be tight and the logarithmic factor superfluous. Applying the lower bound gives the following.

\begin{corollary}\label{cor:intro-explicit}
For all sufficiently small $\lambda>0$,
\[
\tau_{\mathrm{rand}}(\lambda) \leq \tau_{\mathrm{det},\mathrm{avg}}(\lambda)
\leq
\left\lceil
2\log_2 \frac{1}{\lambda}
-\log_2 \log_2 \frac{1}{\lambda}
+\log_2 \log_2 \log_2 \frac{1}{\lambda}
\right\rceil
+11.
\]
\end{corollary}

Further improvements are possible and give even faster algorithms, taking advantage of opportunities to use fast matrix multiplication.
Below, $\omega=\omega(1,1,1)$ is the square matrix multiplication runtime exponent, while $\omega(a,b,c)$ is the rectangular runtime exponent for multiplying an $n^a\times n^b$ matrix by an $n^b\times n^c$ matrix.
In Section~\ref{sec:mm-prelim} we describe the best known bounds on these exponents that we use in our claims of numerical exponent values below.
We obtain the following from two different uses of matrix multiplication; see the brief discussion below and Theorems~\ref{thm:tw2-mm}, \ref{thm:seeded-tree-rectangular-mm}, and~\ref{thm:universal-deterministic-exact-recovery} for more details.
\begin{theorem}
    \label{thm:tau-fast-tw2}
    We have:
\begin{align*}
\tau_{\mathrm{rand}}(c(2)^{-1/2})
&\leq
\tau_{\det, \mathrm{avg}}(c(2)^{-1/2})
\leq
\omega,\\
\tau_{\det}(c(2)^{-1/2})
&\le
\omega+\frac12,
\end{align*}
where $c(2)^{-1/2} < 0.3320$.
Moreover, in the worst-case deterministic clique location model, strong detection alone for every fixed $\lambda>c(2)^{-1/2}$ can be achieved in time $n^{\omega+o(1)}$.
\end{theorem}

\begin{theorem}
\label{thm:intro-main-bounds}
\label{thm:intro-revised-tau}
\label{thm:tau-summary}
\label{thm:main-upper-bounds}
For every $\lambda>0$,
\begin{align*}
\tau_{\mathrm{rand}}(\lambda) \leq \tau_{\det,\mathrm{avg}}(\lambda) &\leq \min_{s\in\ZZ_{\geq 0}}\left\{\omega(1,1,s/2)
+
\frac12
\left\lceil
\log_2\frac{1}{\lambda^2\,2^s e}
\right\rceil_+\right\}, \\
\tau_{\det}(\lambda)
&\leq
\min_{\substack{s\in\ZZ_{\geq0}\\ \lambda\geq(2^s e)^{-1/2}}}
\min\left\{
\omega(1,1,\kappa_s+1/2),
\frac52+\kappa_s
\right\},
\end{align*}
where, in the worst-case deterministic bound,
\[
\kappa_s=
\begin{cases}
0, & s=0,\\
(s+1)/2, & s\ge1.
\end{cases}
\]
\end{theorem}

Figure~\ref{fig:four-envelope} plots the best bound among these and the others given above as compared to the corresponding bounds \eqref{eq:tau-rand-aks} and \eqref{eq:tau-det-aks} that come from combining the results of~\cite{alon1998finding} and~\cite{deshpande2015finding} along with the improvements to the boosting scheme of the former described above.

In the next result, we also record the resulting leading-order behavior as $\lambda\to0$.

\begin{corollary}[Asymptotic upper bound]
\label{cor:intro-asymptotic-envelopes}
For sufficiently small $\lambda > 0$,
\begin{align*}
\tau_{\mathrm{rand}}(\lambda) \leq \tau_{\det,\mathrm{avg}}(\lambda)
&\le
\log_2\frac1\lambda+0.978, \\
\tau_{\det}(\lambda)
&\le
\log_2\frac1\lambda+1.978
\end{align*}
\end{corollary}
\noindent
These have the same leading order $\Theta(\log_2 \lambda^{-1})$ behavior as \eqref{eq:tau-rand-aks} and \eqref{eq:tau-det-aks} respectively, but improve on the next order $\Theta(1)$ constants, by roughly $0.801$ and $0.301$, respectively.

Lastly, as an ancillary result, we note that our analysis of signed subgraph counts also turns out to describe some polynomials that, if they could be calculated efficiently, would achieve strong detection even if $k \leq n^{1/2 - \varepsilon}$.
The following standard conjecture predicts that this should be impossible.
\begin{conjecture}[Planted clique hardness]\label{conj:clique}
For every fixed $\varepsilon\in(0,1/2)$, no polynomial-time algorithm achieves strong detection, and hence no polynomial-time algorithm achieves strong recovery, in the planted clique problem when $k\le n^{1/2-\varepsilon}$.
\end{conjecture}
\noindent
Thus, if we believe the conjecture, then our results also show the hardness of calculating certain polynomials. Indeed, as a simple example, the function
\begin{equation}
f(G) = \#\{\text{cliques of size } (2 + \varepsilon)\log_2 n \text{ in G}\} \label{eq:hard-f}
\end{equation}
can be written as a polynomial of degree $O(\log^2 n)$ (the number of edges in a clique of the size being counted), and certainly this argument applies to that $f(G)$ since with high probability there exist no cliques of this size in $G \sim \QQ_n$. But, the same actually holds for a large and robust family of polynomials based on signed subgraph counts of sufficiently rich families of graphs.
The details of this result may be found in Theorem~\ref{thm:separation-hardness}.

\subsection{Proof techniques}

\subsubsection{Signed graph count polynomials}
\label{sec:signed-count}

The starting point of our algorithms follows the approach of a long line of recent work \cite{hopkins2017bayesian,DHS-2020-SpikedMatrixHeavyTailed,mao2024testing,mao2023exact,li2025algorithmic,li2025smooth,chen2025detecting,chen2026computational} in studying polynomials that sum certain symmetric classes of Boolean Fourier coefficients of a graph's adjacency matrix, or equivalently count a certain notion of \emph{signed subgraph} of an input graph $G$ in order to achieve detection and recovery. We define these notions precisely below.

We identify a graph $G$ on labelled vertex set $[n]$ with its $\{\pm 1\}$-valued adjacency matrix $X \in \{\pm 1\}^{\binom{[n]}{2}}$, having off-diagonal entries
\[
X_{ij} \colonequals 2\cdot \mathbf{1}\{\{i,j\}\in E(G)\}-1\in\{\pm1\} \text{ for } 1\le i<j\le n.
\]
We consider polynomials of $X$ associated to subsets of its indices, $H \subseteq \binom{[n]}{2}$. We view such a subset as itself being a graph, with the following standard but important details of this identification:
\begin{definition}
    \label{def:H}
    For $H \subseteq \binom{[n]}{2}$, we view $H$ as a set of edges on vertex set $[n]$. We write $V(H) \colonequals \{i \in [n]: i \in e \text{ for some } e \in H\}$, the set of vertices that are not isolated in this graph.
\end{definition}
\noindent The polynomials we will study are then as follows.
\begin{definition}[Graph characters and symmetrizations]
\label{def:char}
Given a graph $H$ on some number $|V(H)| = v \leq n$ of labelled vertices and an injective embedding $\phi:V(H) \hookrightarrow [n]$, we define the associated \emph{character} by
\[
\chi_{H,\phi}(X) \colonequals \prod_{\{a,b\}\in E(H)} X_{\phi(a)\phi(b)} \in \{\pm1\}.
\]
All Boolean Fourier characters over $\{\pm 1\}^{\binom{n}{2}}$ can be viewed as realized in this way for some $H$ and an injective $\phi$. Summing over all injective embeddings of a given $H$, we define the polynomials
\[
T_H(X) \colonequals \sum_{\phi:V(H)\hookrightarrow[n]} \chi_{H,\phi}(X).
\]
More generally, for a set of such graphs $\sH$, we define the polynomials
\[
F_{\sH}(X) \colonequals \sum_{H\in \sH}T_H(X).
\]
\end{definition}

\begin{remark}
Note that if $H$ is isomorphic to $H'$, i.e., $H$ and $H'$ only differ by a relabeling of their vertices, then $T_H=T_{H'}$. We will end up taking $\sH$ a family of labeled graphs closed under relabeling or isomorphism (see Definition~\ref{def:iso-closed}), which in particular contains several isomorphic copies of various graphs. Thus, effectively $F_{\sH}$ for such a family is a weighted sum over isomorphism classes of graphs, where the isomorphism class of each $H$ is counted with multiplicity $v! / |\Aut(H)|$. This particular weighting scheme will turn out to be convenient in our calculations to come.
\end{remark}

We omit it from this overview for the sake of brevity, but in Section~\ref{sec:recovery} we also define a vector-valued analog of $F_{\sH}(X)$ that we use to design recovery algorithms.

\subsubsection{Separation by polynomials}
In order for thresholding $F_{\sH}(X)$ to succeed at detection, it will suffice to consider the first two moments of $F(X)$ by relying on the following stronger notion of detection that is standard in the literature on low-degree polynomial algorithms.
\begin{definition}[Strong separation]
    Let $F = F_n: \{0, 1\}^{\binom{n}{2}} \to \RR$. We say that these functions \emph{strongly separate} $\QQ = \QQ_n$ from $\PP = \PP_{n, k}$ if
    \begin{equation}
\left(\E_{\PP}[F(G)]-\E_{\QQ}[F(G)]\right)^2 \gg
\Var_{\PP}[F(G)] + \Var_{\QQ}[F(G)], \label{eq:separation}
\end{equation}
     as $n \to \infty$.
\end{definition}
\noindent The following is then a simple consequence of Chebyshev's inequality.
\begin{proposition}
    If $F$ strongly separates $\PP$ from $\QQ$, then, for a suitable sequence $\gamma = \gamma_n \in \RR$, the test $\mathcal{A}(G) = \1\{F(G) \geq \gamma\}$ achieves strong detection when $\E_{\PP}F>\E_{\QQ}F$, while the test with the reverse inequality achieves strong detection when $\E_{\PP}F<\E_{\QQ}F$.
\end{proposition}

Thus, it suffices for us to produce strongly separating statistics $F(X)$ (we go back and forth freely between thinking of the graph input $G$ and the adjacency matrix input $X$ as needed). But, we must balance two considerations: on the one hand we would like \eqref{eq:separation} to hold, but on the other we need to avoid making $\sH$ so complicated that $F_{\sH}(X)$ becomes intractable to compute.

In order for \eqref{eq:separation} to hold, it essentially just suffices for $\sH$ to be sufficiently large. We state our criterion for strong separation as a standalone result since it will be useful for both our upper and our lower bounds and may be of independent interest.

\begin{assumption}[Graph family assumptions]
    \label{asm:H}
    We suppose $\sH = \sH^{(v)} \subseteq 2^{\binom{[v]}{2}}$ is a sequence of sets of graphs satisfying the following:
    \begin{enumerate}
        \item As the notation indicates, each $H \in \sH^{(v)}$ is a graph on $v$ labeled vertices.
        \item $\sH^{(v)}$ is closed under relabeling of vertices (equivalently, $\sH^{(v)}$ is \emph{isomorphism-closed} in the terminology of Definition~\ref{def:iso-closed}).
        \item Each $H \in \sH^{(v)}$ is connected.
        \item We have
        \[
        \log c(\sH) \colonequals \liminf_{v \to \infty} \frac{1}{v} \log\left(\frac{|\sH^{(v)}|}{v!}\right) > -\infty.
        \]
    \end{enumerate}
\end{assumption}
\noindent Note that the last notation is compatible with our previous notation in the sense that $c(t)$ discussed earlier is $c(t) = c(\sG^{\conn}_{v, \leq t})$.

\begin{theorem}[General separation criterion]
    \label{thm:general-separation}
    Let $\sH = \sH^{(v)}$ be a sequence of sets of graphs satisfying Assumption~\ref{asm:H}, and let $v(n)$ satisfy
    \[
    \omega(1) \leq v(n) \leq \frac{\log n}{5\log \log n}.
    \]
    If $k = k(n) \geq \lambda\sqrt{n}$ with $\lambda > 1 / \sqrt{c(\sH)}$, then the statistic $F_{\sH^{(v(n))}}(X)$ strongly separates $\PP_{n, k(n)}$ from $\QQ_n$.
\end{theorem}
\noindent The connectedness condition is a mild but important one that turns out to be sufficient to control the variances in the separation condition \eqref{eq:separation}.\footnote{Our lower bound in Theorem~\ref{thm:separation-hardness} follows from the same idea, but requires us to handle the case of $|\sH^{(n)}|$ growing much more rapidly with $n$ in order to achieve \eqref{eq:separation} when $k = n^{1/2 - \varepsilon}$. That in turn requires a more delicate analysis of the variances in \eqref{eq:separation}, leading to more stringent connectivity conditions than merely that all $H \in \sH^{(n)}$ are connected.}

We emphasize that, by Theorem~\ref{thm:general-separation}, \emph{any} sufficiently large family $\sH$ of connected graphs for which we can efficiently compute $F_{\sH}$ will yield a distinguishing statistic for the planted clique problem. Thus the ``bottleneck'' in designing such algorithms is not in ensuring separation but in computing $F_{\sH}$, and this is where the bounded treewidth assumption plays an important role in our specific choice of $\sH$.

\subsubsection{Efficient approximation by color-coding}
\label{sec:intro-approx}

To make our upper bounds computationally effective, we must find $\sH$ that satisfies the above conditions but such that $F_{\sH}(X)$ can be computed efficiently. The difficulty with this is that to evaluate $F_{\sH}(X)$ naively takes time roughly $|\sH| \cdot (n)_v$, since for each $H \in \sH$ there are $(n)_v$ injective maps $V(H) \hookrightarrow [n]$. Recall that we take $|\sH| \approx c^v v!$ and $\omega(1) \leq v \leq o(\log n / \log\log n)$. Thus, $|\sH| = n^{o(1)}$, so calculating $F_{\sH}(X) = \sum_{H \in \sH}T_H(X)$ by enumerating over $H \in \sH$ does not affect the polynomial runtime. On the other hand, $(n)_v = n^{\omega(1)}$, so computing each $T_H(X)$ in this brute force way is prohibitively expensive.
Like many of the works cited below in Section~\ref{sec:related}, we address this issue using the \emph{color-coding} technique \cite{alon1995color}, whose idea is to fix a coloring of the vertices and consider only summing over ``colorful'' $\phi$ that assign each vertex a different color, which turns out to be tractable via dynamic programming provided the graphs being counted have bounded treewidth.
See Section~\ref{sec:approx} for further details.

\subsubsection{Fast matrix multiplication and seeding}

Furthermore, similar to observations for instance of the recent work \cite{bringmann2021current}, certain operations in color coding can be optimized by recognizing that they perform implicit matrix multiplications.
In particular, this occurs in our setting for signed counts of graphs of treewidth at most $t = 2$, leading to the improvement in Theorem~\ref{thm:tau-fast-tw2} and the appearance of $\omega$.

For Theorem~\ref{thm:intro-main-bounds}, we instead consider using an algorithm computing signed counts of trees ($t = 1$) together with the boosting scheme of \cite{alon1998finding}.
As explained above, this works by running the algorithm on the induced subgraph on mutual neighbors of small seed subsets of vertices.
It turns out that summing over all trees for each such seed choice can be written in terms of a sequence of rectangular matrix multiplications, one of whose axes enumerates all chosen seed subsets; thus, all seed choices are treated at once rather than serially as in the original proposal of \cite{alon1998finding}.
Applying fast matrix multiplication to these operations leads to the exponents $\omega(1, 1, s/2)$ appearing in Theorem~\ref{thm:intro-main-bounds}.
Naive matrix multiplication would be equivalent to iterating over the seed subsets serially; at an intuitive level, fast matrix multiplication can be viewed here as in effect giving a ``partial parallelization'' of that computation.

\subsubsection{Low-degree recovery algorithms}

For recovery, we use a rooted version of the same polynomials $F_{\sH}(X)$, which we discuss in Section~\ref{sec:recovery}. For each $H \in \sH$, we sum over choices of a root vertex $r \in V(H)$, and sum over embeddings $\phi$ with $\phi(r) = i$ to produce a polynomial $F_{\sH, i}(X)$. We show a suitable rescaling of the vector-valued polynomial $(F_{\sH, 1}(X), \dots, F_{\sH, n}(X))$ gives a good estimator of the indicator vector of the planted clique in an $\ell^2$ sense; while $F_{\sH}(X)$ for us plays the role of an approximation to the low-degree likelihood ratio, this vector plays the role of an approximation to the low-degree MMSE estimator, as studied by works like \cite{schramm2022computational,sohn2025sharp}.

However, this only gives a ``soft'' estimator of clique membership. We then perform a somewhat involved rounding procedure to actually exactly recover the clique. First, thresholding this vector yields an approximate candidate set that is not too large and contains most of the clique vertices. However, since the distribution of the induced subgraph on this set of vertices is no longer a tractable distribution like a smaller planted clique instance, we instead perform a ``cleanup'' step by only running the above strategy on a large induced subgraph of the input graph, and retain the small induced subgraph on the complement as a ``holdout'' set. We then identify the clique vertices in the holdout set by thresholding their degrees to the initial clique estimate, which with high probability outputs exactly the clique vertices in the holdout set. Finally, the clique vertices in the first subgraph are with high probability those adjacent to all clique vertices in the holdout set. This ``zig-zag'' strategy going back and forth between two subgraphs allows for most of the cleanup phase of the algorithm to use edges independent from those used in the main low-degree polynomial estimator.

\subsection{Related work}
\label{sec:related}

\paragraph{Low-degree polynomials as statistical algorithms}
The idea of using explicit low-degree polynomials to construct statistical algorithms has been used in several notable recent works. We note that much work on low-degree polynomial algorithms in statistics focuses on \emph{lower bounds} against such algorithms; see \cite{kunisky2019notes,wein2025computational} for surveys mostly focused on this direction. For \emph{upper bounds} this technique, especially relying on color-coding, seems to originate with the work of \cite{hopkins2017bayesian} on the stochastic block model, who proposed using low-degree polynomials to approximate the Bayesian posterior mean estimator. Ideas very close in spirit to ours have appeared in similar applications to matching correlated graphs \cite{mao2024testing,mao2023random} and detecting communities in correlated stochastic block models \cite{chen2025detecting,chen2026computational}. The performance of these algorithms ends up related to the asymptotics of certain graph counts, in their case Otter's constant associated to the growth of the number of unlabeled trees \cite{otter1948number}. This is precisely the same role that graphs of bounded treewidth play in our setting. Similar ideas have also been applied to matrix and tensor PCA \cite{DHS-2020-SpikedMatrixHeavyTailed,li2025algorithmic,li2025smooth}. The latter reference \cite{li2025smooth} is quite close in spirit to ours as well, as it both relies on the asymptotics of counts of certain hypergraphs (since it works with a tensor problem) and uses these to establish a smooth tradeoff between algorithm runtime and required signal strength. The work of \cite{yu2025counting} is also close to ours in examining using polynomials computing signed subgraph counts for general planted subgraph problems, but they give a coarser analysis and their algorithms are based on simple families of graphical polynomials associated to \emph{stars}. Their results apply to the planted clique problem, but only in the sub- and super-critical regimes $k \ll \sqrt{n}$ or $k \gg \sqrt{n}$, not in the dependence on $\lambda$ in the critical regime $k = \lambda\sqrt{n}$.

\paragraph{Message-passing, belief propagation, and tree-shaped polynomials}
Message-passing algorithms such as belief propagation (BP) and approximate message passing (AMP) are another common framework for designing statistical algorithms, also closely related to analytical tools from statistical physics; see \cite{yedidia2001understanding,zdeborova2015statphys,FVRS-2021-TutorialAMP} for some general surveys. Various work has studied the connection between these methods and graphical polynomials associated specifically to trees. For example, this connection is made quite explicitly in \cite{montanari2022equivalence,jones2024fourier}. For specifically the planted clique problem, as we have mentioned, it is understood as folklore in the low-degree polynomials literature that the BP algorithm of \cite{deshpande2015finding} is essentially equivalent to computing a low-degree polynomial whose terms are associated to trees. Our approach may be viewed as extending such a polynomial to graphs of higher treewidth; it is an interesting question whether there is a natural message-passing algorithm that corresponds to that extension.

\paragraph{Algorithms and treewidth}
Graphs of bounded treewidth are well-known to be algorithmically tractable for various problems beyond our applications; see, e.g., \cite{fomin2012faster,bringmann2021current} for fast algorithms for detecting and counting graphs of bounded treewidth. Our results will rely on the enumeration of labeled graphs of treewidth at most $t$, which has been studied by numerous works in the combinatorics literature. The case $t = 1$ admits a closed-form solution by Cayley's theorem \cite{Cayley-1889-CountingTrees} (the current historical consensus seems to be that this was rather originally discovered by Borchardt \cite{Borchardt-1860-CountingTrees}). Precise asymptotics for the case $t = 2$ were obtained by \cite{bodirsky2007enumeration}. For $t \geq 3$, similarly precise asymptotics are not known, as discussed by \cite{castellvi2024chordal}. Looser upper and lower bounds for all $t$ were established by \cite{baste2017number}.
\section{Preliminaries}

\subsection{Notation}\label{sec:notation}

\paragraph{Combinatorics}
For an integer $n \ge 1$, let $[n] := \{1,2,\dots,n\}$. For $0 \le v \le n$, let $\binom{[n]}{v}$ denote the collection of all $v$-subsets of $[n]$. We use the falling factorial $(n)_v := n(n-1)\cdots(n-v+1)$, with $(n)_0 := 1$, so that the number of injections $\phi:[v]\hookrightarrow [n]$ is $(n)_v$. We write $S_v$ for the symmetric group of permutations of $[v]$.

\paragraph{Asymptotics}
We write $\poly(\cdot)$ for an unspecified polynomial, and use the standard asymptotic notation $O(\cdot)$, $\Omega(\cdot)$, $\Theta(\cdot)$, and $o(\cdot)$ as $n\to\infty$. Throughout, asymptotic statements are with respect to $n\to\infty$; auxiliary parameters may depend on $n$ unless explicitly declared to be fixed constants. We write ``w.h.p.'' for ``with probability $1-o(1)$ as $n \to \infty$.'' For two nonnegative sequences $(a_n)$ and $(b_n)$, we write $a_n\gg b_n$ to mean $a_n/b_n\to\infty$ as $n\to\infty$; equivalently, $b_n=o(a_n)$. Unless stated otherwise, all logarithms are base $2$. For $x\in\RR$, write $\lceil x\rceil_+\colonequals \max\{0,\lceil x\rceil\}$.
\paragraph{Graphs}
Given a graph $G = (V, E)$, if $|V| = n$ we often identify $V$ with $[n]$. We also write $v(G) \colonequals |V(G)|$ and $e(G) \colonequals |E(G)|$. For $v \in V(G)$, we write $N_G(v) \colonequals \{w \in V(G): \{v, w\} \in E(G)\}$, the set of neighbors of $v$ in $G$. The adjacency matrix of $G$ is $A\in\{0,1\}^{n\times n}$, where for $1\le i<j\le n$, $A_{ij}=A_{ji}=\mathbf{1}\{\{i,j\}\in E\}$. We also sometimes encode edges by the centered variables $X_{ij}:=2A_{ij}-1\in\{\pm1\}$ for $1\le i<j\le n$, and extend symmetrically by setting $X_{ji}=X_{ij}$ and $X_{ii}=0$.

\paragraph{Graph embeddings}
For another graph $H = (V(H), E(H))$, we call an \emph{injective embedding} of $H$ into $V(G) = [n]$ an injection $\phi:V(H)\hookrightarrow [n]$. We write $\Inj(V(H),[n])$ for the set of all such injections. Note that this only exists when $|V(H)| \leq n$. We reserve the variable $v = v(H) = |V(H)|$ for such situations, where $v$ denotes the number of vertices in the smaller graph and $n = |V(G)|$ the number of vertices in the larger graph into which the smaller one is being embedded. As for $G$, we often identify $V(H)$ with $[v]$ when we are following this notation. In such situations, we call $G$ the \emph{host graph} and $H$ the \emph{pattern graph}.

For such a graph $H$ with labeled vertex set $[v]$, for each permutation $\pi \in S_v$ we write $\pi(H)$ for $H$ with its vertices relabeled according to $\pi$. We denote its automorphism group by
\[
\Aut(H)
\colonequals
\bigl\{\pi\in S_v:\ \{a,b\}\in E(H) \text{ if and only if } \{\pi(a),\pi(b)\}\in E(H)\bigr\}.
\]
For two such graphs $H,H'$, we denote the set of isomorphisms between them as
\[
\Iso(H,H')
\colonequals
\bigl\{\pi\in S_v:\ \{a,b\}\in E(H) \text{ if and only if } \{\pi(a),\pi(b)\}\in E(H')\bigr\}.
\]
In particular, $\Iso(H,H)=\Aut(H)$. We write $H \cong H'$ if $\Iso(H, H') \neq \emptyset$, and say that $H$ and $H'$ are isomorphic in this case.

\subsection{Computational model}

Since we are discussing granular polynomial runtimes of algorithms, we must be somewhat careful about the meaning of these runtimes. Our results apply to any reasonable model that supports exact computation with real numbers; the real RAM model \cite{Shamos-1978-ComputationalGeometry} is a common such choice. We view the input into the algorithms we discuss as the $n \times n$ adjacency matrix of the $n$-vertex graph $G$, and an $O(n^{\tau})$ runtime means an algorithm performing $O(n^{\tau})$ arithmetic operations. Our algorithms will truly only involve actual arithmetic operations since they are based on computing polynomials; earlier when we compared to more complicated computations such as spectral algorithms, we appealed implicitly to standard results about their runtime.

\subsection{Graph families}

As mentioned above, we will work extensively with sets of graphs $\sH$ on the same vertex set $[v]$. All such sets we work with will be of the following kind:

\begin{definition}[Isomorphism-closed family]\label{def:iso-closed}
    We call a set of graphs $\sH$ on $[v]$ \emph{isomorphism-closed} if whenever $H \in \sH$ and $H' \cong H$, then $H' \in \sH$. Equivalently, every relabeling of the vertices of any $H \in \sH$ is also in $\sH$, and again equivalently membership in $\sH$ depends only on the isomorphism type of a graph.
\end{definition}

\begin{lemma}\label{lem:iso-sum-vfact}
Let $\sH$ be an isomorphism-closed family of labeled graphs on $[v]$. Then, for every $H\in\sH$,
\[
\sum_{H'\in\sH} |\Iso(H,H')|
=
v!.\]
\end{lemma}

\begin{proof}
    Since $\sH$ is isomorphism-closed, $\pi(H) \in \sH$ for every $\pi \in S_v$.
    The result follows by grouping the $v!$ permutations of $S_v$ according to the value of $\pi(H)$.
\end{proof}

\subsection{Treewidth}
\label{sec:treewidth}

\label{subsec:bounded-tw-family}

Treewidth is a fundamental graph parameter, introduced by Robertson and Seymour in their theory of graph minors \cite{robertson1986graph}, which in a certain sense measures how much a graph behaves like a tree.

\begin{definition}[Tree decomposition]
A \emph{tree decomposition} of a graph $H=(V(H),E(H))$ is a pair
\[
\Bigl(\mathsf T,\{B_t\subseteq V(H)\}_{t\in V(\mathsf T)}\Bigr),
\]
where $\mathsf T$ is a tree and each $B_t$ is a \emph{bag}, such that:
\begin{enumerate}
  \item $\bigcup_{t\in V(\mathsf T)} B_t = V(H)$.
  \item For every edge $\{u,v\}\in E(H)$, there exists $t\in V(\mathsf T)$ such that $\{u,v\}\subseteq B_t$.
  \item For every vertex $u\in V(H)$, the set $\{t\in V(\mathsf T): u\in B_t\}$ induces a connected subtree of $\mathsf T$.
\end{enumerate}
\end{definition}

Informally, a tree decomposition represents $H$ by a tree of overlapping vertex sets, which are the bags. The first two conditions say that every vertex of $H$ and every edge of $H$ are contained in some bag. The final condition says that the tree structure on the bags is such that there are paths between any two bags that overlap. Thus, $H$ may be viewed as assembled from the bags arranged in a tree-like fashion. Note that, because every edge must be contained in some bag, all tree decompositions of a graph with at least one edge will have some bags of size at least 2.

\begin{definition}[Treewidth]\label{def:treewidth}
The \emph{width} of a tree decomposition $(\mathsf T,\{B_t\})$ is defined by
\[
\width(\mathsf T,\{B_t\})
:=
\max_{t\in V(\mathsf T)} \bigl(|B_t|-1\bigr).
\]
The \emph{treewidth} of a graph $H$ is the smallest width of any tree decomposition of $H$.
\end{definition}

The graphs of treewidth at most 1 are the forests, and the connected ones with at least two vertices are trees; the graphs of treewidth at most 2 are formed by series-parallel graphs glued together at pairs of vertices into a tree structure (confusingly, sometimes graphs of treewidth at most 2 are themselves called series-parallel graphs) \cite{bodlaender1997treewidth}. Alternatively, the graphs of treewidth at most some $t \geq 1$ are characterized by a collection of forbidden minors; for $t = 1$ the only forbidden minor is $K_3$, and for $t = 2$ the only forbidden minor is $K_4$.

Fix integers $v\ge 2$ and $t\ge 1$. Recall that we have defined
\[
\mathcal{G}_{v,\le t}^{\conn}
\colonequals
\bigl\{H \text{ a labeled graph on vertex set }[v] : H \text{ connected}, \tw(H)\le t\bigr\}.
\]
Our results will rely on several pieces of prior work that compute or bound the sizes of these sets, which are as follows.

\begin{theorem}[Cayley's formula \cite{Borchardt-1860-CountingTrees,Cayley-1889-CountingTrees}]\label{thm:cayley}
    For all $v \geq 1$,
    \[
    |\sG^{\conn}_{v, \leq 1}| = v^{v - 2}.
    \]
    The set $\sG^{\conn}_{v, \leq 1}$ is also the set of labelled trees on $v$ vertices.
\end{theorem}

\begin{theorem}[Theorems 3.7 and 6.1 of \cite{bodirsky2007enumeration}]\label{thm:count-tw2}
As $v \to \infty$, we have the asymptotic
\[
|\mathcal G^{\conn}_{v,\le 2}|
\sim
C\, v^{-5/2}\rho^{-v} v!,
\]
where $\rho \approx 0.11021$ and $C \approx 0.0067912$.
\end{theorem}

\begin{remark}
In Theorem 3.7 of \cite{bodirsky2007enumeration}, the authors enumerate all graphs of treewidth at most~2, not only connected ones. But, in Theorem 6.1 they also deduce the asymptotic fraction of such graphs that are connected, and combining these results yields the statement of our Theorem~\ref{thm:count-tw2}.
\end{remark}
\noindent For larger treewidth, we use the following general bounds.

\begin{theorem}[Lemma 1 and Theorem 2 of \cite{baste2017number}]\label{lem:count-tw}
For integers $v$ and $t$ with $2\le t\le v$,
\begin{equation}
2^{-\frac{t(t+3)}{2}}
t^{-2t-2} \cdot \left(\frac{1}{128e}\cdot \frac{t2^t}{\log t}\,v\right)^{\!v}
\le
|\mathcal G^{\conn}_{v,\le t}|
\le
2^{-\frac{t(t+1)}{2}}
t^{-t} \cdot \left(t2^tv\right)^{\!v} \label{eq:bns}
\end{equation}
\end{theorem}

\begin{remark}
The bounds in \cite{baste2017number} are stated for all graphs of treewidth at most $t$. But, their upper bound of course bounds all connected such graphs also, while the graphs constructed to prove their lower bound are all connected and thus also give a lower bound on connected such graphs.
\end{remark}

Unfortunately, the lower bound of \eqref{eq:bns} is very far from tight for small $t$. So, we instead rely on the work of \cite{castellvi2024chordal} giving precise asymptotics for the subset of connected \emph{chordal} graphs of treewidth at most $t$, which then gives a lower bound on $|\sG_{v, \leq t}^{\conn}|$. A graph is \emph{chordal} if every cycle of length at least four has a chord.
\begin{theorem}[Theorem 1.1 in \cite{castellvi2024chordal}]
\label{thm:chordal}
For every $t\ge 2$, there exist constants $c_t>0$ and $\gamma_t>1$ such that
\[
|\mathcal G^{\conn}_{v,\le t}|
\ge
c_t\,v^{-5/2}\gamma_t^{\,v}v!
\]
for all sufficiently large $v$. Moreover, the constants $c_t$ and $\gamma_t$ can be approximated numerically; for several small values of $t$, the corresponding values of $1/\gamma_t$ are listed in Table~\ref{table:gamma}.
\end{theorem}

\begin{table}
\begin{center}
\begin{tabular}{c||ccccccc}
    \hline
    $t$
& $1$
& $2$
& $3$
& $4$
& $5$
& $6$
& $7$
\\
\hline
$1/\gamma_{t}$
& $0.36788$
& $0.14665$
& $0.07703$
& $0.04444$
& $0.02657$
& $0.01608$
& $0.00974$
\\ \hline
\end{tabular}
\caption{Approximate numerical values of $1/\gamma_{t}$ for $1\le t\le 7$, as appear in Theorem~\ref{thm:chordal}. The value for $t=1$ follows from Cayley's formula, and the values for $2\le t\le7$ are as given in Table~1 of \cite{castellvi2024chordal}.}\label{table:gamma}
\end{center}
\end{table}

\subsection{Matrix multiplication exponents}
\label{sec:mm-prelim}

We will use matrix multiplication exponents to state several running-time bounds. Let $\omega$ denote the usual square matrix multiplication exponent. In our numerical claims, we use the bound $\omega < 2.371177$ from \cite{dupont2026improving}.

More generally, for $a,b,c\ge 0$, let $\omega(a,b,c)$ denote the rectangular matrix multiplication exponent. That is, $\omega(a,b,c)$ is the infimum over all $\gamma$ such that an $n^a\times n^b$ matrix can be multiplied by an $n^b\times n^c$ matrix in $n^{\gamma+o(1)}$ arithmetic operations. Thus $\omega=\omega(1,1,1)$. We use the convention $\omega(a,b,0)=a+b$ corresponding to the input-size bound for multiplying an $n^a\times n^b$ matrix by an $n^b\times 1$ vector. In particular, $\omega(1,1,0)=2$. We also use the standard cyclic and transpose symmetries, which imply that $\omega(a,b,c)$ is invariant under all permutations of $(a,b,c)$. Thus a product of shape $(n^\kappa\times n)\cdot(n\times n)$ has exponent $\omega(\kappa,1,1)=\omega(1,1,\kappa)$. We will also use the trivial output-size lower bound
\[
\omega(a,b,c)\ge a+c,
\]
since the product has $n^{a+c}$ output entries. In particular, $\omega(1,1,\kappa)\ge 1+\kappa$.

For later numerical comparisons, we will use the bounds on rectangular matrix multiplication exponents given in Table~\ref{tab:rectangular-exponents}. For $\kappa=1/2,3/2,2$, we use the bounds of \cite{alman2024asymmetry}; for $\kappa=1$, we use \cite{dupont2026improving}; and for $\kappa=5/2,3$, we use the bounds of \cite{vassilevska2024new}.
\begin{table}[t]
\centering
\vspace{1em}
\begin{tabular}{c||c|c|c|c|c|c|c}
\hline
$\kappa$ & $0$ & $1/2$ & $1$ & $3/2$ & $2$ & $5/2$ & $3$ \\
\hline
Upper bound on $\omega(1,1,\kappa)$
& 2.0000 & 2.0428 & 2.3712 & 2.7947 & 3.2501 & 3.7205 & 4.1989 \\
\hline
\end{tabular}
\caption{Upper bounds on rectangular matrix multiplication exponents. The numerical bounds use \cite{alman2024asymmetry} for $\kappa=1/2,3/2,2$, \cite{dupont2026improving} for $\kappa=1$, and \cite{vassilevska2024new} for $\kappa=5/2,3$.}
\label{tab:rectangular-exponents}
\end{table}
Here the entry $\kappa=0$ follows from the convention $\omega(1,1,0)=2$, and the entry $\kappa=1$ is the square matrix multiplication exponent.
For $\kappa\ge3$, partitioning the long dimension into blocks gives
\[
\omega(1,1,\kappa)
\le
\omega(1,1,3)+\kappa-3.
\]
To see this, partition the long output dimension into $n^{\kappa-3}$ blocks of size $n^3$ and apply the exponent $\omega(1,1,3)$ to each block. We use this bound for the small values of $\lambda$ on the left of Figure~\ref{fig:four-envelope} and in Corollary~\ref{cor:intro-asymptotic-envelopes}.

We note that all statements involving quantities of the form $n^{\omega(a,b,c)+o(1)}$ are interpreted with $a,b,c$ fixed constants and $n\to\infty$.
\section{Analysis of low-degree statistics}
\label{sec:stats}
In this section, we introduce the statistics associated with a graph family and develop the framework needed for both detection and recovery in the planted clique problem. We focus on the properties of the low-degree polynomials we propose to compute, and leave the matter of showing that these can be approximated efficiently to the next section.

\subsection{Separation by low-degree polynomials}

Recall from Section~\ref{sec:signed-count} that our plan is to compute the functions, for $\sH$ a family of labelled pattern graphs $H$ on some number of vertices $2 \leq v = v(n) < n$,
\[
F_{\sH}(X) = \sum_{H \in \sH} T_H(X) = \sum_{H \in \sH} \sum_{\phi: V(H) \hookrightarrow [n]} \chi_{H, \phi}(X),
\]
where $\chi_{H, \phi}$ is the Boolean Fourier character associated to the graph $H$ embedded by $\phi$. Recall also that Theorem~\ref{thm:general-separation} gives a general statement about when such $F_{\sH}$ achieves strong separation in the planted clique problem, which in turn implies achieving strong detection. We always assume below that $\sH$ obeys Assumption~\ref{asm:H}.

To show that $F_{\sH}$ achieves strong separation, we must control the four quantities $\E_{\QQ}[F_{\sH}(X)]$, $\Var_{\QQ}[F_{\sH}(X)]$, $\E_{\PP}[F_{\sH}(X)]$, and $\Var_{\PP}[F_{\sH}(X)]$. We do this below, first treating the simpler null model $\QQ$ and then the planted model $\PP$.

\subsubsection{Null model analysis}

For the null model, we show the following.
\begin{lemma}
    \label{lem:detection-null}
    \label{lem:null-moments}
    For any isomorphism-closed family $\sH$ of graphs on $[v]$ having no isolated vertices, we have
    \begin{align*}
        \E_{\QQ} F_{\sH}(X) &= 0, \\
        \Var_{\QQ} F_{\sH}(X) &= |\sH| \cdot (n)_v \cdot v!\,.
    \end{align*}
\end{lemma}
\begin{proof}
    The first claim is immediate since we have
    \[
    \E_{\QQ} \chi_{H, \phi}(X) = \prod_{\{a,b\}\in E(H)} \E_{\QQ}[X_{\phi(a)\phi(b)}] = 0
    \]
    as under $X \sim \QQ$ the $X_{ij}$ are i.i.d.\ with law $\mathrm{Unif}(\{\pm 1\})$.

    Thus also $\Var_{\QQ} F_{\sH}(X) = \E_{\QQ} F_{\sH}(X)^2$. For this second moment, for $H, H' \in \sH$, we have that for any $\phi \in\Inj(V(H),[n])$ and $\psi \in \Inj(V(H'), [n])$,
\begin{align*}
\E_{\QQ} \chi_{H,\phi}(X)\chi_{H',\psi}(X)
&=
\E_{\QQ} \prod_{e\in E(H)} X_{\phi(e)} \prod_{f \in E(H')}X_{\psi(f)} \\
&=
\prod_{\{a, b\} \in \binom{[n]}{2}} \E_{\QQ} X_{a, b}^{\1\{\{a, b\} \in \phi(E(H))\} + \1\{\{a, b\} \in \psi(E(H'))\}},
\end{align*}
    where we group according to how many times each edge appears. Again by independence and since $X_{a, b}^2 = 1$, this expression is 1 if every edge appears either in both $\phi(E(H))$ and $\psi(E(H'))$ or in neither, and is 0 otherwise. But this happens precisely if $H$ and $H'$ are isomorphic, so that $\psi^{-1} \circ \phi$ is an isomorphism between them. Thus, we have
    \[
    \E_{\QQ}[T_H(X) T_{H'}(X)] = (n)_v |\Iso(H, H')|.
    \]
    Expanding the second moment,
    \[
\Var_{\QQ}(F_{\mathcal H})
=
\E_{\QQ}[F_{\mathcal H}^2]
=
\sum_{H,H'\in\mathcal H} (n)_v\,|\Iso(H,H')|
=
|\mathcal H|\,(n)_v\,v!
\]
    by Lemma~\ref{lem:iso-sum-vfact}.
\end{proof}

\subsubsection{Planted model analysis}
\label{subsec:planted-moments-embed}

For the planted model, we show the following more complicated estimates. We first give an exact calculation of the first two moments, but which does not give a practical bound on the variance.

For every integer $s$ with $0\le s\le n$, introduce the important quantities
\[
\alpha_s \colonequals \frac{(k)_s}{(n)_s}.
\]
In particular,
\[
\alpha_v = \frac{k}{n} \cdot \frac{k - 1}{n - 1} \cdots \frac{k - v + 1}{n - v + 1}.
\]
Also, given $H, H' \in \sH$ and $\phi \in \Inj(V(H), [n])$ and $\psi \in \Inj(V(H'), [n])$, recall that $\phi(H)$ and $\psi(H')$ may be viewed as labeled graphs on $[n]$ (whose edge sets are given by applying $\phi$ and $\psi$ to the edge sets of $H$ and $H'$, respectively). We write $\phi(H) \Delta \psi(H')$ for the graph formed by the symmetric difference of these edge sets. We recall our notation that $V(\phi(H) \Delta \psi(H'))$ is the set of vertices of such a graph, not including the isolated vertices. We then define
\[
\sigma(H, H'; \phi, \psi) \colonequals |V(\phi(H) \Delta \psi(H'))|.
\]
\begin{lemma}
    \label{lem:detection-planted}
    \label{lem:planted-moments}
    For any isomorphism-closed family $\sH$ of graphs on $[v]$ having no isolated vertices, we have
    \begin{align}
        \E_{\PP} F_{\sH}(X) &= |\sH| \cdot (n)_v \cdot \alpha_v = |\sH| \cdot (k)_v, \\
        \Var_{\PP} F_{\sH}(X) &=
\sum_{H,H'\in\mathcal H}
\sum_{\substack{\phi\in\Inj(V(H),[n]) \\ \psi\in\Inj(V(H'),[n])}}
\bigl(\alpha_{\sigma(H,H';\phi,\psi)}-\alpha_v^2\bigr). \label{eq:VarPF}
    \end{align}
\end{lemma}
\begin{proof}
    Let $C \sim \mathrm{Unif}(\binom{[n]}{k})$ denote the planted clique under $G \sim \PP$. For any $H \in \sH$ and $\phi: V(H) \to [n]$, conditional on $C$, we have
    \[
\E_{\PP}[\chi_{H,\phi}(X) \mid C]
=
\prod_{\{a,b\}\in E(H)} \E[X_{\phi(a)\phi(b)}\mid C]
=
\mathbf 1\{\phi(V(H))\subseteq C\},
\]
    since conditional on $C$ the $X_{ij}$ are either equal to 1 if $i, j \in C$ or distributed as $\mathrm{Unif}(\{\pm 1\})$ otherwise. Thus we have
    \[
    \E_{\PP}[\chi_{H,\phi}(X)] = \PP[\phi(V(H))\subseteq C] = \alpha_v,
    \]
    and summing over all injections gives $\E_{\PP}[T_H(X)] = (n)_v \cdot \alpha_v = (k)_v$, and then summing over all $H \in \sH$ gives the first result (recall that by assumption $|V(H)| = v$ for all $H \in \sH$).

    For the second moment, we similarly compute
    \begin{align*}
\E_{C}\E_{\PP}\bigl[\chi_{H,\phi}(X)\chi_{H',\psi}(X)\mid C\bigr]
&=
\E_{C}\E_{\PP}\left[\prod_{\{a, b\} \in \phi(E(H)) \triangle \psi(E(H'))} X_{a,b} \,\,\bigg|\,\, C\right] \\ &=
\PP[V(\phi(H) \Delta \psi(H')) \subseteq C] \\
&= \alpha_{\sigma(H, H'; \phi, \psi)}.
\end{align*}
    Summing over pairs of graphs and embeddings then gives that the second moment is
    \[
\E_{\PP}[F_{\mathcal H}(X)^2]
=
\sum_{H,H'\in\mathcal H}
\sum_{\substack{\phi\in\Inj(V(H),[n]) \\ \psi\in\Inj(V(H'),[n])}}
\alpha_{\sigma(H,H';\phi,\psi)}.
\]
    and the formula for the variance follows.
\end{proof}

Next, we produce a convenient bound on the ``low-degree advantage'' that governs the separation condition \eqref{eq:separation} we are trying to establish. To write down our bound, we introduce some additional parameters associated to the family of graphs $\sH$. Consider choosing $H,H'$ independently and uniformly at random from $\mathcal H$, and $\phi,\psi$ independently and uniformly at random from $\Inj([v],[n])$. We will use the following definition only in regimes where $2v\le n$. Then, for integers $u\in\{0,1,\dots,v\}$ and $\ell\in\{0,1,\dots,u\}$, we define
\[
p_{u,\ell}(\mathcal H)
\colonequals
\Pr\!\left[
\sigma(H,H';\phi,\psi)=2v-u-\ell
\,\,\bigg|\,\,
|\phi([v])\cup \psi([v])|=2v-u
\right].
\]
Equivalently, $p_{u,\ell}(\mathcal H)$ is the conditional probability that exactly $\ell$ vertices in the overlap $\phi([v])\cap \psi([v])$ do not belong to the vertex set of the symmetric difference, $V(\phi(H) \Delta \psi(H'))$.
\begin{lemma} \label{lem:sufficient-conditions-overlap}\label{lem:variance-master}
    Let $\sH$ be an isomorphism-closed set of graphs on $v$ labeled vertices, each having no isolated vertices. Suppose $v = o(\sqrt{k})$. Then, we have
    \begin{align*}
    &\frac{\Var_{\PP}[F_{\sH}(X)] + \Var_{\QQ}[F_{\sH}(X)]}{ (\E_{\PP} F_{\sH}(X) - \E_{\QQ} F_{\sH}(X))^2} \\
    &\hspace{1cm} \leq O\left(\frac{(n / k^2)^v(v + 1)!}{|\mathcal H|}
+ \max_{\substack{0\le \ell \le u \leq v - 1\\ (u,\ell)\ne(0,0)}}
  \frac{v^{2u}\,n^{\ell}}{k^{u + \ell}}\,p_{u,\ell}(\mathcal H) + \max_{0 \leq \ell \leq v - 1}
  \frac{(v+1)!\,n^{\ell}}{k^{v + \ell}}\,p_{v,\ell}(\mathcal H)
\right).
\end{align*}
\end{lemma}
\noindent We defer the proof to Appendix~\ref{apx:sufficient-condition}, but the idea is simply that the only complicated quantity above is $\Var_{\PP}[F_{\sH}(X)]$, and we may bound this by establishing suitable asymptotics for the $\alpha_t$ and using the $p_{u, \ell}(\sH)$ to expand according to the value of $\sigma(H, H'; \phi, \psi)$.

We now move towards applying the more specific assumptions on $\sH$ that appear in Theorem~\ref{thm:general-separation}. Connectedness of the graphs in $\sH$ gets us the following.
\begin{lemma}\label{lem:puu-zero-connected}\label{lem:connected-overlap}
If each $H\in\mathcal H$ is connected, then, for every $1\le u\le v-1$, $p_{u,u}(\mathcal H) = 0$.
\end{lemma}
\begin{proof}
Fix $H,H'\in\mathcal H$ and injections $\phi,\psi\in\Inj([v],[n])$ such that
\[
|\phi([v])\cup \psi([v])|=2v-u.
\]
Then the vertex intersection of the two embeddings, namely $\phi([v])\cap \psi([v])$, has size $u$. If we have $\sigma(H,H';\phi,\psi)=2v-2u$, or equivalently if $\ell=u$, then every vertex in this intersection is absent from the vertex set of the symmetric difference. This means that no vertex in $\phi([v])\cap \psi([v])$ is incident to any edge of $\phi(E(H)) \triangle \psi(E(H'))$.

However, since $H$ and $H'$ are connected and $u<v$, this is impossible: we have $\phi([v]) \setminus \psi([v]) \neq \emptyset$, and since $H$ is connected this set must be connected to $\phi([v]) \cap \psi([v])$. In particular, there exists some $x \in \phi([v]) \setminus \psi([v])$ and $y \in \phi([v]) \cap \psi([v])$ such that $\{x, y\} \in \phi(E(H))$. Further, since $x \notin \psi([v])$, we have $\{x, y\} \in \phi(E(H)) \triangle \psi(E(H'))$ as well, and so the vertex $y$ gives a contradiction.

Hence there are no $\phi, \psi$ such that $\ell = u$, and so $p_{u,u}(\mathcal H)=0$.
\end{proof}

\subsubsection{Sufficient condition on \texorpdfstring{$\sH$}{H}: Proof of Theorem~\ref{thm:general-separation}}
\label{sec:general-separation}

We now apply the general criterion above to the special kinds of families of graphs that are covered by Theorem~\ref{thm:general-separation}.

\begin{proof}[Proof of Theorem~\ref{thm:general-separation}]
    By the assumptions of the Theorem we may apply Lemma~\ref{lem:sufficient-conditions-overlap}, and it suffices to show that the bound provided there is $o(1)$:
    \[
    \frac{(n / k^2)^v\, (v + 1)!}{|\mathcal H|} + \max_{\substack{0\le \ell \le u < v\\(u,\ell)\ne(0,0)}} \frac{v^{2u}\,n^{\ell}}{k^{u + \ell}}\,p_{u,\ell}(\mathcal H) + \max_{0\le \ell < v} \frac{(v+1)!\,n^{\ell}}{k^{v + \ell}}\,p_{v,\ell}(\mathcal H) = o(1).
    \]
    We bound the three terms individually.

    For the first term, choose $c_0<c(\mathcal H)$ such that $c_0\lambda^2>1$. By the definition of $c(\mathcal H)$, for all sufficiently large $v$ we have $|\sH| \geq c_0^v v!$. Combining this with $k \geq \lambda\sqrt n$ and $v=o(\sqrt{k})$, we have
    \[
    \frac{(n / k^2)^v\, (v + 1)!}{|\mathcal H|} \leq (c_0\lambda^2)^{-v}(v + 1) = o(1)
    \]
    since $c_0\lambda^2 > 1$ and $v = \omega(1)$ by assumption.

    For the second term, since every graph in $\mathcal H$ is connected, Lemma~\ref{lem:puu-zero-connected} gives $p_{u,u}(\mathcal H)=0$ for all $1\le u\le v-1$. Thus it suffices to consider $1\le u<v$ and $0\le \ell\le u-1$. Using $p_{u,\ell}(\mathcal H)\le 1$ and $k \geq \lambda\sqrt n$, we obtain
\[
v^{2u}n^\ell k^{-(u+\ell)}p_{u,\ell}(\mathcal H)
\le
v^{2u}n^\ell (\lambda\sqrt n)^{-(u+\ell)}
=
\lambda^{-(u+\ell)}v^{2u}n^{(\ell-u)/2}\le (1 \wedge \lambda)^{-2v}v^{2v}n^{-1/2}
\]
since $\ell\le u-1$. Moreover, from $v< \frac{\log n}{5\log\log n}$, we get $(1 \wedge \lambda)^{-2v}=n^{o(1)}$ and
\[
v^{2v}=2^{2v\log v}\le n^{2/5}
\]
for all sufficiently large $n$. Hence
\[
\max_{\substack{1\le u<v\\0\le \ell\le u-1}}
v^{2u}n^\ell k^{-(u+\ell)}p_{u,\ell}(\mathcal H)
=o(1).
\]

Finally, we bound the third error term. The missing diagonal endpoint $\ell=0$ is harmless:
\[
(v+1)!k^{-v}p_{v,0}(\mathcal H)
\le
(v+1)!(\lambda\sqrt n)^{-v}
=o(1).
\]
For the remaining values of $\ell$, and in fact uniformly for $0\le \ell\le v-1$, bounding by $p_{v,\ell}(\mathcal H)\le 1$ gives
\begin{align*}
(v+1)! \, n^\ell k^{-(v+\ell)}p_{v,\ell}(\mathcal H)
&\le
(v+1)! \, n^\ell (\lambda\sqrt n)^{-(v+\ell)}
 \\ &=
(v+1)!\,\lambda^{-(v+\ell)}n^{(\ell-v)/2} \\ &\le (v+1)!\,(1 \wedge \lambda)^{-2v}n^{-1/2}
\end{align*}
because $\ell\le v-1$. By Stirling's formula and the assumption $v<\frac{\log n}{5\log\log n}$, we have $(v+1)!\le n^{1/5+o(1)}$, while also $(1 \wedge \lambda)^{-2v}=n^{o(1)}$. Together with the factor $n^{-1/2}$ above, this gives
\[
\max_{0\le\ell< v}
(v+1)! \, n^\ell k^{-(v+\ell)}p_{v,\ell}(\mathcal H)
=o(1).
\]

All three error terms in Lemma~\ref{lem:sufficient-conditions-overlap} are therefore $o(1)$, and it follows that $F_{\sH}(X)$ achieves strong separation, as claimed.
\end{proof}

By the observations in Section~\ref{sec:treewidth}, the family $\sH = \sG_{v, \leq t}^{\conn}$ with $v = \lfloor\sqrt{\log n}\rfloor$ satisfies these assumptions for every fixed $t$. Thus, for $\mathcal H=\mathcal G^{\mathrm{conn}}_{v,\le t}$, Theorem~\ref{thm:general-separation} gives the required statistical separation, while Lemma~\ref{lem:enumerate-bounded-treewidth-patterns} makes the pattern-enumeration cost $n^{o(1)}$. Section~\ref{sec:approx-detection} gives the randomized algorithm based on color coding, and Section~\ref{sec:derandomization} replaces its random choices to prove Theorem~\ref{thm:intro-main-general}.

\subsubsection{Hardness of computing graph-family statistics}
\label{sec:hardness-rich-families}

We now formulate our statement that $F_{\sH}(X)$ for sufficiently rich families of graphs $\sH$ could moreover achieve strong detection for $k = n^{1/2 - \varepsilon}$, violating Conjecture~\ref{conj:clique} on the hardness of this problem. This gives evidence via reduction that the polynomials $F_{\sH}(X)$ should be hard to compute.

To state our result precisely, we introduce the following structural condition on $H \in \sH$. For $H$ a graph and any $S\subseteq V(H)$, recall that $N_H(v)$ denotes the neighborhood of $v$ in $H$, and define
\[
B(S) \colonequals \{v\in V(H)\setminus S: N_H(v) \cap S=\emptyset\},
\]
the set of vertices in $H$ with no neighbor in $S$ and not in $S$ themselves.
For disjoint sets $S,T\subseteq V(H)$, write
\[
E_H(S,T)
\colonequals
\bigl\{\{x,y\}\in E(H):x\in S,\ y\in T\bigr\}.
\]

\begin{definition}\label{def:VC}
We say that a graph $H$ on $v$ vertices is \emph{$\beta$-vertex-connected}, denoted $\mathrm{VC}(\beta)$, for some $\beta>0$ if, for every $S\subseteq V(H)$,
\begin{equation}\label{eq:vcb}
|B(S)|\le \max\left\{2^{-\beta|S|}(v-|S|),\frac{2}{\beta}\log v\right\}.
\end{equation}
\end{definition}
\noindent This is a one-sided pseudorandomness condition with a
logarithmic cutoff. It implies the following bound for random
disjoint vertex subsets.

\begin{lemma}\label{lem:VC-implies}
Assume that $H$ satisfies $\mathrm{VC}(\beta)$ and that
$v=|V(H)|$ is sufficiently large depending only on $\beta$. Fix integers $s,t$ with $0\le s\le |V(H)|$ and $0\le t\le |V(H)|-s$. Sample $S$ uniformly from $\binom{V(H)}{s}$ and then sample $T$ uniformly from $\binom{V(H)\setminus S}{t}$. Then,
\[
\Pr\bigl[E_H(S,T)=\emptyset\bigr]\le 2^{-(\beta/8)st}.
\]
\end{lemma}

\begin{proof}
The cases $s=0$ or $t=0$ are immediate. By symmetry, assume
$1\le s\le t$. Set $L=2\log v/\beta$ and
$h=\lfloor\log v/(2\beta)\rfloor$.
For sufficiently large $v$, we have $h\ge L/8$,
$v-h\ge v/2$, and $\sqrt v/2\ge L$.

If $s>L$, then $2^{-\beta s}(v-s)<1/v\le L$, so
$|B(S)|\le L<t$, and the event in question is impossible.

Otherwise, let $r=\min\{s,h\}$ and choose a uniformly random
$r$-subset $R$ of $S$. Then $r\ge s/8$, and
\[
2^{-\beta r}(v-r)\ge\sqrt v/2\ge L,
\]
so \eqref{eq:vcb} gives $|B(R)|\le2^{-\beta r}(v-r)$.
The pair $(R,T)$ is uniform over disjoint sets of sizes $r,t$.
Consequently,
\[
\Pr[E_H(S,T)=\emptyset]
\le\Pr[E_H(R,T)=\emptyset]
=\E_R\frac{\binom{|B(R)|}{t}}{\binom{v-r}{t}}
\le2^{-\beta rt}
\le2^{-(\beta/8)st}.
\]
\end{proof}
We now give the precise statement of the hardness result discussed in the introduction.

\begin{theorem}\label{thm:separation-hardness}
Suppose $\mathcal H$ is an isomorphism-closed family of graphs on $[v]$ such that, for some $\beta > 0$ not depending on $n$, every $H\in\mathcal H$ satisfies $\mathrm{VC}(\beta)$ as in \eqref{eq:vcb}. Assume further that $|\mathcal H|\ge 2^{\theta v^2}$ for some $\theta>0$ also not depending on $n$. Then there exist constants $C_1=C_1(\beta,\theta)>0$, $C_2=C_2(\beta,\theta)>0$, and $\varepsilon=\varepsilon(\beta,\theta)>0$ such that, for all sufficiently large $n$, all $v = v(n)$ satisfying
\[
C_1\log n< v(n) <n^{C_2},
\]
and $k= k(n) = n^{1/2-\varepsilon}$, the statistic $F_{\mathcal H}$ achieves strong separation and thus strong detection between $\PP_{n,k}$ and $\QQ_n$.

Consequently, if Conjecture~\ref{conj:clique} holds, then there is no polynomial-time algorithm computing the polynomial $F_{\sH}(X)$.
\end{theorem}
\noindent The proof is deferred to Appendix~\ref{apx:separation}.

\subsection{Recovery by low-degree polynomials}
\label{sec:recovery}

We now show that variants of the $F_{\sH}(X)$ discussed above can also be used to design an algorithm recovering the planted clique in the same regime and under the same assumptions as in Theorem~\ref{thm:general-separation}. Recall that in the recovery problem we are only interested in the planted model $X \sim \PP$, so all expectations, variances, and probabilities in this section are with respect to $\PP$. Our strategy is to first construct a vector correlated with the indicator vector of the clique, and then to round it to identify the actual clique. We give the full procedure in Algorithm~\ref{alg:cleanup}.

\subsubsection{Polynomial estimator of clique indicator}

For a connected graph $H$, let
\[
\mathcal R(H)=\{r\in V(H):H-r\text{ is connected}\}.
\]
Every connected graph on at least two vertices has at least two vertices in $\mathcal R(H)$ (for instance, take the leaves of a spanning tree). For the first part, we use vector-valued analogs of the polynomials $F_{\sH}(X)$ from Definition~\ref{def:char}. Put
\[
\nu(\mathcal H)=\sum_{H\in\mathcal H}|\mathcal R(H)|.
\]
Then, for $\sH$ as before and $i \in [n]$, we define
\begin{equation}
\label{eq:rooted-sum}
F_{\sH,i}(X)
:=
\sum_{H\in\sH}\ \sum_{r\in\mathcal R(H)}\ \sum_{\substack{\phi:V(H)\hookrightarrow[n]\\ \phi(r)=i}}
\chi_{H,\phi}(X).
\end{equation}
We call $r$ above the \emph{root} of each inner summation above, and in forming $F_{\sH,i}(X)$ we sum over all roots in $\mathcal R(H)$, for each $H \in \sH$.

The next lemma gives the first two moments of these rooted statistics under the planted model.
\begin{lemma}\label{lem:rooted-moments}
Suppose $\mathcal H\subseteq\mathcal G^{\mathrm{conn}}_{v,\le t}$ is an isomorphism-closed family for some fixed $t$. Let $G\sim \PP_{n,k}$ with planted clique $C\subset[n]$ have $\{\pm 1\}$-valued adjacency matrix $X$. For each $i\in[n]$, the following hold.
\begin{enumerate}
\item For every $i\in[n]$,
\begin{align*}
\E\bigl[F_{\sH,i}(X)\mid i\in C\bigr]
&=
\nu(\sH) \cdot (k-1)_{v-1} \equalscolon \mu, \\
\E\bigl[F_{\sH,i}(X)\mid i\notin C\bigr] &= 0.
\end{align*}
\item Fix a constant $\lambda>0$, and assume $k=\lambda\sqrt n$ and $v=O\!\left(\frac{\log n}{\log\log n}\right)$. Then, uniformly over all $i\in[n]$ as $n \to \infty$,
\begin{align*}
\Var\bigl[F_{\sH,i}(X)\mid i\notin C\bigr]
&=
(1+o(1))\,\nu(\sH)\cdot (n-k-1)_{v-1}\cdot v!, \\
\Var\bigl[F_{\sH,i}(X)\mid i\in C\bigr]
&\le
(1+o(1))\,\nu(\sH)\cdot (n-k)_{v-1}\cdot v!
\end{align*}
\end{enumerate}
\end{lemma}
\noindent The proof is deferred to Appendix~\ref{sec:proof_root}.

Defining the normalization
\begin{equation}
p_i(G) = p_i(X) \colonequals \frac{1}{\mu} F_{\sH,i}(X), \label{eq:pi}
\end{equation}
we see that the vector $p(X) = (p_1(X), \dots, p_n(X))$ is an unbiased estimator of the indicator vector $\1^{C}$ having entries $\1^C_i = \1\{i \in C\}$. We next approximate the variance of this estimator.
\begin{lemma}
    \label{lem:p-var}
    \label{lem:rooted-mse}
    Suppose that $\sH = \sH^{(v)}$ satisfy Assumption~\ref{asm:H}, that $\sH^{(v)}\subseteq\mathcal G^{\mathrm{conn}}_{v,\le t}$ for some fixed $t$, and that $v(n)\to\infty$ with $v(n)=o(\log n/\log\log n)$. Define $p_i(G)$ as above from $F_{\sH^{(v(n))}}$, and suppose $k \geq \lambda\sqrt{n}$ with $\lambda > 1 / \sqrt{c(\sH)}$. Fix $c_0<c(\sH)$ such that $c_0\lambda^2>1$. Then, for each $i \in [n]$,
    \begin{align*}
    \E\bigl[(p_i(G)-1)^2\mid i\in C\bigr]
    &\leq (1+o(1))(c_0 \lambda^2)^{-(v - 1)}, \\
    \E\bigl[p_i(G)^2\mid i\notin C\bigr]
    &\leq (1+o(1))(c_0 \lambda^2)^{-(v - 1)},
    \end{align*}
    whenever the conditioning event has positive probability, where the $o(1)$ term is uniform over $i\in[n]$. Consequently, the same upper bound holds for $\E (p_i(G) - \1^C_i)^2$.
\end{lemma}
\noindent The proof is deferred to Appendix~\ref{apx:rooted-mse}.

\subsubsection{Rounded estimator}
We now take such an estimator as a black-box input, and analyze a preliminary thresholded rounding of $p_i(G)$, which outputs a set of indices that is correlated with, though not yet with high probability the same as, the planted clique $C$.

\begin{lemma}
    \label{lem:rounding}
    Suppose that $p_i: \{0, 1\}^{\binom{[n]}{2}} \to \RR$ satisfy, uniformly in $i$,
    \begin{align*}
    \E[(p_i(G)-1)^2\mid i\in C]&\leq \delta_n, \\
    \E[p_i(G)^2\mid i\notin C]&\leq \delta_n,
    \end{align*}
    whenever the conditioning event has positive probability, where $\delta_n=o(1)$.
    Define the rounded estimator
\begin{align*}
\widehat{\1}_i(G) &\colonequals \mathbf 1\{p_i(G)\ge 1/2\}, \\
C'(G) &\colonequals \{i\in[n]:\widehat{\1}_i(G)=1\}.
\end{align*}
    Then, with probability $1 - O(\delta_n^{1/2}) = 1 - o(1)$, we have
    \begin{align*}
        |C^{\prime}| &\leq k+O(n\delta_n^{1/2}), \\
        |C^{\prime} \cap C| &\geq (1 - \delta_n^{1/2})k.
    \end{align*}
\end{lemma}
\begin{proof}
    If $\delta_n=0$, the two conditional mean-square bounds imply $p_i(G)=\1_i^C$ almost surely for every $i$, and hence $C'(G)=C$ almost surely. We may therefore assume that $\delta_n>0$.
    We first control the error probability for each individual vertex. By Markov's inequality and the two conditional mean-square bounds,
\begin{align*}
\Pr\bigl[\widehat{\1}_i(G)\neq 1\mid i\in C\bigr]
&\le 4\,\E[(p_i(G)-1)^2\mid i\in C]\le4\delta_n,\\
\Pr\bigl[\widehat{\1}_i(G)\neq 0\mid i\notin C\bigr]
&\le 4\,\E[p_i(G)^2\mid i\notin C]\le4\delta_n.
\end{align*}
Let $X:=|C\setminus C'(G)|$ be the number of missed clique vertices. Then
\[
\E[X]
=
\sum_{i=1}^n\Pr(i\in C)\Pr(\widehat{\1}_i=0\mid i\in C)
\le4k\delta_n.
\]
Therefore,
\[
\Pr\bigl[|C'(G)\cap C|\le (1-\delta_n^{1/2})k\bigr]
=
\Pr[X\ge \delta_n^{1/2}\cdot k]
\le 4\delta_n^{1/2}.
\]
Now, let
\[
Y:=|C'(G)\triangle C|
\]
be the size of the symmetric difference. By the same argument, $\E[Y]\le 4n\delta_n$, and so again by Markov's inequality
\[
    \Pr\bigl[Y\ge n\delta_n^{1/2}\bigr]\le 4\delta_n^{1/2}.
    \]
    The results then follow from these two bounds by union bound and since $\delta_n=o(1)$.
\end{proof}

\subsubsection{Cleanup with holdout subgraph}

Finally, we perform a last ``cleanup'' step that refines the estimate $C^{\prime}$ to the exact planted clique $C$ with high probability. We formulate this by instead running the above algorithm on a ``main'' subgraph of the input graph $G$, and retaining a ``holdout'' subgraph for the cleanup step; while our method is similar in spirit to that of \cite{alon1998finding}, it is different in this technical aspect. The reason for doing this is that, conditional on $C$, the presence of edges in the main subgraph, the holdout subgraph, and between the two are all independent, simplifying the probabilistic reasoning.

\begin{lemma}
    \label{lem:recovery}
    \label{lem:holdout-cleanup-generic}
    Fix an integer $t\ge1$, constants $0<\lambda_-<\lambda$, and a known integer sequence $v=v(N)$ satisfying $v(N)\ge2$ for all sufficiently large $N$ and $v=o(\log N/\log\log N)$. Put $\mathcal H_N=\mathcal G^{\mathrm{conn}}_{v(N),\le t}$ and
    \[
    \mu_{N,K}\colonequals \nu(\mathcal H_N)(K-1)_{v(N)-1}.
    \]
    Suppose there is an algorithm which, on input $G[U]\sim\PP_{N,K}$ with $|U|=N$ and $K\ge\lambda_-\sqrt N$, outputs raw scores $\widehat F_i$, $i\in U$, such that the normalized estimates $\widehat p_i^{(K)}\colonequals\widehat F_i/\mu_{N,K}$ satisfy
    \begin{align*}
    \E[(\widehat p_i^{(K)}-1)^2\mid i\in C\cap U]&\le\delta_N,\\
    \E[(\widehat p_i^{(K)})^2\mid i\notin C\cap U]&\le\delta_N,
    \end{align*}
    whenever the conditioning event has positive probability, uniformly over $i\in U$, where
    \[
    \delta_N^{1/2}=o((\log N)^{-3}).
    \]
    Suppose the raw scores are computed in time $T(N)$ throughout this regime, where $T$ may be taken nondecreasing. Then there is an algorithm running in time $T(n)+O(n^2)$ that achieves strong recovery whenever $k\ge \lambda\sqrt n$.
\end{lemma}

We build up to the proof of this in a few steps. The idea is to partition the vertex set $V$ of the input $G = ([n], E)$ into a large part $V_1$ that will be the main subgraph above and a small part $V_2$ that will be the holdout set above. By symmetry, we may choose these deterministically or uniformly at random; this will not make a difference but we fix the sizes
\begin{align*}
    n_2 = |V_2| &\colonequals \lfloor n^{3/5} \rfloor, \\
    n_1 = |V_1| &\colonequals n - n_2.
\end{align*}
Our plan is to apply the algorithm suggested by Lemma~\ref{lem:rounding} to the induced subgraph $G[V_1]$ to obtain $V_3 \subseteq V_1$ that contains most vertices of $C$ and has sufficiently large signal-to-noise ratio for the holdout threshold. We then threshold numbers of edges between $V_2$ and $V_3$ to exactly recover $C_2 \colonequals C \cap V_2$. Finally, having this we can recover $C_1 \colonequals C \cap V_1$, and then $C = C_1 \sqcup C_2$.

\begin{algorithm}[t!]
\LinesNumbered
\DontPrintSemicolon
\caption{\textsc{RecoveryAlgorithm}}
\label{alg:cleanup}
\KwIn{A graph $G=([n],E)$, a target clique size $k$, a treewidth bound $t$, and the prescribed integer sequence $v=v(N)$}
\KwOut{An estimated clique set $\widehat C\subseteq [n]$}
Set $n_2 \colonequals \lfloor n^{3/5}\rfloor$, $n_1 \colonequals n-n_2$\;
Choose a uniformly random partition $[n]=V_1\sqcup V_2$, where $|V_1|=n_1$ and $|V_2|=n_2$\;
Set $v_1\colonequals v(n_1)$, $\widetilde k_1 \colonequals \lfloor(1-n^{-2/5})k\rfloor$, and $\widetilde\mu_1\colonequals\nu(\sH)(\widetilde k_1-1)_{v_1-1}$, where $\sH = \sG^{\conn}_{v_1, \leq t}$\;
Run the raw-score algorithm assumed in Lemma~\ref{lem:recovery} with $\sH$ on the induced subgraph $G[V_1]$ to obtain $(\widehat F_{\sH,i}(G[V_1]))_{i \in V_1}$, and set $\widehat p_i\colonequals \widehat F_{\sH,i}(G[V_1])/\widetilde\mu_1$\;
Apply the rounding rule in Lemma~\ref{lem:rounding} to $(\widehat p_i)_{i \in V_1}$ to obtain a vertex set $V_3\subseteq V_1$\;
Let $T \colonequals |V_3| / 2 + 0.2499\,\widetilde{k}_1$ (see Remark~\ref{rem:threshold-T}) \;
Initialize $\widehat{C}_2 \colonequals \varnothing$\;
\ForEach{$x\in V_2$}{
    \If{$|E(\{x\}, V_3)| \ge T$}{
        add $x$ to $\widehat{C}_2$\;
    }
}
Initialize
$\widehat C_1 \colonequals \varnothing$\;
\ForEach{$u\in V_1$}{
    \If{$u$ is adjacent to every vertex in $\widehat{C}_2$}{
        add $u$ to $\widehat C_1$\;
    }
}
Set $\widehat C:=\widehat C_1 \sqcup \widehat{C}_2$\;
\Return{$\widehat C$}\;
\end{algorithm}

\begin{lemma}\label{lem:partition-clique-size}\label{lem:partition-clique-sizes}\label{lem:holdout-partition}
    Suppose the assumptions of Lemma~\ref{lem:recovery} hold, and $[n] = V_1 \sqcup V_2$ with $|V_i| = n_i$ are as above. Let $C_i \colonequals C \cap V_i$ and $k_i \colonequals |C_i|$. Then, with probability $1 - o(1)$,
    \begin{align}
        k_1 &=(1-o(1))k, \label{eq:k1-big} \\
        k_2 &\geq \frac{\lambda}{4}n^{1/10}, \label{eq:k2-big}\\
        |k_1-\widetilde k_1|&\le \sqrt{k}\log n+O(k/n+1), \\
        \widetilde k_1 &\colonequals \lfloor(1-n^{-2/5})k\rfloor.
        \label{eq:k1-proxy-close}
    \end{align}
\end{lemma}
\noindent The proof is deferred to Appendix~\ref{apx:proof-recover1}.

Fix $\lambda_-$ as in Lemma~\ref{lem:recovery}, and write $v_1=v(n_1)$. On the event in Lemma~\ref{lem:partition-clique-size}, for all sufficiently large $n$ we have $k_1\ge\lambda_-\sqrt{n_1}$. Moreover, since $v_1=o(\log n/\log\log n)$, \eqref{eq:k1-proxy-close} gives, uniformly over all values of $k_1$ on this event,
\begin{align}
\left|\log\frac{\nu(\sH)(k_1-1)_{v_1-1}}
{\nu(\sH)(\widetilde k_1-1)_{v_1-1}}\right|
&\le
O\!\left(\frac{v_1|k_1-\widetilde k_1|}{k}\right)\notag\\
&=
o((\log n)^{-3}).
\label{eq:proxy-normalization}
\end{align}
For a value $K$ of $k_1$ on this event, put $\mu_{n_1,K}=\nu(\sH)(K-1)_{v_1-1}$. The implementable normalization in Algorithm~\ref{alg:cleanup} then satisfies
\[
a_{n,K}\colonequals\frac{\mu_{n_1,K}}{\widetilde\mu_1}
=1+o((\log n)^{-3}).
\]
Conditional on $k_1=K$, the graph $G[V_1]$ is distributed according to $\PP_{n_1,K}$. Hence, for $b\in\{0,1\}$,
\[
\E\!\left[(a_{n,K}\widehat p_i^{(K)}-b)^2\,\middle|\,
k_1=K,\ \mathbf 1_{\{i\in C_1\}}=b\right]
\le
2a_{n,K}^2\delta_{n_1}+2(a_{n,K}-1)^2.
\]
Consequently, uniformly over such $K$, the proxy-normalized scores satisfy the two conditional mean-square bounds in Lemma~\ref{lem:recovery} after enlarging $\delta_{n_1}^{1/2}$ by an $o((\log n)^{-3})$ term; below we continue to denote the enlarged bound by $\delta_{n_1}$.

We may now apply Lemma~\ref{lem:rounding} conditional on each $k_1=K$ in the event above, average the resulting bounds over $K$, and add the $o(1)$ probability that the event fails. This gives an algorithm that runs in time $T(n)$ and outputs $V_3 \subseteq V_1$ such that, with probability $1 - o(1)$,
\begin{align}
    |V_3| &\le k_1+O(n_1\delta_{n_1}^{1/2}), \label{eq:V3-size-bound} \\
    |V_3 \cap C_1| &\geq (1-\delta_{n_1}^{1/2}) k_1, \label{eq:V3-C1-big}\\
    \frac{|V_3\cap C_1|^2}{|V_3|}
    &=\omega((\log n_1)^3). \label{eq:V3-snr}
\end{align}
Indeed, the last claim follows from the first two since
\[
\frac{|V_3\cap C_1|^2}{|V_3|}
\ge
c\min\left\{k_1,\frac{k_1^2}{n_1\delta_{n_1}^{1/2}}\right\}
\ge
c\min\left\{\lambda_-\sqrt{n_1},\frac{\lambda_-^2}{\delta_{n_1}^{1/2}}\right\}.
\]

We next use the degrees between $V_2$ and $V_3 \subseteq V_1$ to recover $C_2$.
\begin{lemma}\label{lem:filter-V2}\label{lem:ideal-holdout-threshold}\label{lem:implementable-threshold}
Define the ideal threshold
\[
T_{\mathrm{ideal}} \colonequals \frac{|V_3|}{2} + \frac{|V_3\cap C_1|}{4},
\]
and let
\[
\widehat{C}_{2,\mathrm{ideal}} \colonequals\{v\in V_2: |E(\{v\}, V_3)| \ge T_{\mathrm{ideal}}\} \subseteq V_2.
\]
Also let
\[
  \widetilde k_1=\lfloor(1-n^{-2/5})k\rfloor,
  \qquad
  T=\frac{|V_3|}{2}+0.2499\,\widetilde k_1.
\]
Define
\[
  \widehat C_2
  =
  \{v\in V_2: |E(\{v\},V_3)|\ge T\}.
\]
Then, on the event that \eqref{eq:k1-big}, \eqref{eq:k2-big}, \eqref{eq:V3-C1-big}, and \eqref{eq:V3-snr} hold, we have
\[
\widehat C_{2,\mathrm{ideal}}=\widehat C_2=C_2
\]
with probability $1-o(1)$.
\end{lemma}
\noindent The proof is deferred to Appendix~\ref{apx:proof-recover2}.

\begin{remark}
\label{rem:threshold-T}
The threshold in Algorithm~\ref{alg:cleanup} is implementable because it uses only $k$, $n$, and $V_3$, not the unknown set $C_1$. Lemma~\ref{lem:implementable-threshold} shows that replacing the ideal threshold $T_{\mathrm{ideal}}$,
\[
|V_3|/2+|V_3\cap C_1|/4
\]
by the algorithmic threshold $T$,
\[
|V_3|/2+0.2499\lfloor(1-n^{-2/5})k\rfloor
\]
does not affect the cleanup step with probability $1-o(1)$.
\end{remark}

Finally, if we have identified $C_2$ exactly, then we may also identify $C_1$ exactly:
\begin{lemma}\label{lem:recover-K1-from-K2}\label{lem:final-adjacency-cleanup}\label{lem:recover-main-from-holdout}
Define
\[
\widehat{C}_1 \colonequals \{v\in V_1:\ |E(\{v\}, \widehat{C}_2)| = |\widehat{C}_2| \}.
\]
Then
\[
\Pr\bigl(\widehat C_2=C_2,\ \widehat C_1\ne C_1\bigr)=o(1).
\]
Consequently, if $\Pr(\widehat C_2=C_2)=1-o(1)$, then $\widehat C_1\sqcup\widehat C_2=C$ with probability $1-o(1)$.
\end{lemma}
\noindent The proof is deferred to Appendix~\ref{apx:proof-recover3}.

Lemma~\ref{lem:recovery} then follows by combining Lemmas~\ref{lem:partition-clique-size}, \ref{lem:rounding}, \ref{lem:implementable-threshold}, and \ref{lem:final-adjacency-cleanup} as described above to construct the required algorithm.

\section{Color coding approximation of low-degree statistics}
\label{sec:approx}

We now turn to the algorithmic part of the argument. Our goal in this section is to show how the detection and recovery statistics $F_{\sH}(X)$ and $F_{\sH, i}(X)$ can be approximated efficiently by means of color-coding. While previously we could take $\sH = \sH^{(v)}$ to be fairly generic families of graphs under only mild conditions, here we will need to rely on the much stronger treewidth bound (see Section~\ref{sec:treewidth} for the definition).

\subsection{Approximating detection statistics}
\label{sec:approx-detection}

To approximate the detection statistic $F_{\sH}(G)$, we use the classical color-coding method of Alon, Yuster, and Zwick \cite{alon1995color}, which reduces the problem of counting embeddings to counting only those embeddings whose vertices receive distinct colors under a random coloring.

We review the notation presented in Section~\ref{sec:intro-approx}. Let $\zeta:[n]\to [v]$ be a random coloring obtained by assigning independently to each vertex of $[n]$ a uniformly random color in $[v]$. Whenever $|W|=v$,
\begin{equation}
\Pr\!\bigl[\zeta \text{ is injective on } W\bigr]
\equalscolon
\rho_v = \frac{v!}{v^v}, \label{eq:rhov}
\end{equation}
the probability that $v$ independently sampled colors from $[v]$ are all distinct.

As in our algorithms, let $\sH$ be a set of graphs on $[v]$, which we call \emph{pattern graphs} following other references on color-coding. For a pattern graph $H\in\sH$, we define the associated colorful embedding sum by
\[
T_H^{\mathrm{col}}(G,\zeta)
\colonequals
\sum_{\phi:V(H)\hookrightarrow [n]}
\mathbf{1}\{\zeta \text{ is injective on } \phi(V(H))\}
\prod_{e\in E(H)} X_{\phi(e)}.
\]
In other words, $T_H^{\mathrm{col}}(G,\zeta)$ is obtained from the full embedding sum by restricting to those embeddings whose image is colorful under $\zeta$. Since each fixed embedding on $v$ vertices is colorful with probability $\rho_v$, averaging over $\zeta$ yields an unbiased scaled version of the original quantity. The key algorithmic advantage is that, once the colors are fixed, colorful embeddings can be computed efficiently by dynamic programming. There are two main differences between the simplest application of color-coding compared to ours: first, color-coding is simplest on trees and forests rather than bounded treewidth graphs, and second it is usually applied to detect or count subgraphs rather than to compute the signed counts we are interested in.

\paragraph{Warmup: Color-coding DP for tree patterns.}
Let us first explain the idea in the case of $H$ a tree. Suppose that we are interested in detecting whether the \emph{host graph} $G = (V, E)$ contains a colorful copy of $H$ (one where each vertex is given a different color by $\zeta$). Choose $r \in V(H)$ to serve as a root. For each $x \in V(G)$, we maintain a set of sets of colors $\Col_x(H) \subseteq 2^{[v]}$, where $C\in\Col_x(H)$ if and only if there exists a colorful embedding $\phi$ of $H$ into $G$ such that $\phi(r)=x$ and $\zeta(\phi(V(H)))=C$.

These collections can be computed recursively as follows. If $H$ consists of a single vertex, then for each $x \in V(G)$ we simply set $\Col_x(H) = \{\{\zeta(x)\}\}$. Otherwise, select an edge $e_T=(r,r')\in E(H)$. Deleting $e_T$ splits $H$ into two rooted subtrees $H_1$ and $H_2$, with roots $r$ and $r'$, respectively. By recursion, for each $w\in V(G)$ we obtain the family $\Col_w(H_1)$ of color sets realized by colorful embeddings of $H_1$ where $r$ is mapped to $w$, and for each $x\in V(G)$ we obtain $\Col_x(H_2)$ of color sets realized by colorful embeddings of $H_2$ where $r'$ is mapped to $x$. We then reconstruct $\Col_w(H)$ by scanning all edges $(w,x)\in E(G)$ and, for every pair of color sets $C_1\in\Col_w(H_1), C_2\in\Col_x(H_2)$ with $C_1\cap C_2=\emptyset$, we insert $C_1\cup C_2$ into $\Col_w(H)$. This bottom-up recursion over the tree $H$ yields the required color sets.

In the recursive algorithm below, $[v]$ is the color palette of the original pattern tree and remains fixed in every recursive call. Thus a recursive call may receive a proper subtree of the original pattern, but all of its tables are still indexed by subsets of the same palette $[v]$.

\begin{algorithm}[t!]
\LinesNumbered
\DontPrintSemicolon
\caption{\textsc{ColorfulCharacterTree}$(G,H,r,\zeta)$}
\label{alg:tree-colorful-character}
\KwIn{A host graph $G=([n],E)$ with edge-weights $X_{ab}\in\{\pm1\}$; a fixed color palette $[v]$; a rooted tree $(H,r)$ on at most $v$ vertices; a coloring $\zeta:[n]\to [v]$}
\KwOut{The table $\DP_{H,r}(x,C)$ for all $x\in[n]$ and $C\subseteq [v]$ as defined in \eqref{eq:DP-tree}}

\If{$|V(H)|=1$}{
    \ForEach{$x\in[n]$}{
        \ForEach{$C\subseteq [v]$}{
            $\DP_{H,r}(x,C)\gets 0$\;
        }
        $\DP_{H,r}(x,\{\zeta(x)\})\gets 1$\;
    }
    \Return $\DP_{H,r}$\;
}

Choose an edge $e_T=(r,r')\in E(H)$\;
Delete $e_T$, and let $(H_1,r)$ and $(H_2,r')$ be the two rooted components\;

$\DP_{H_1,r}\gets \textsc{ColorfulCharacterTree}(G,H_1,r,\zeta)$\;
$\DP_{H_2,r'}\gets \textsc{ColorfulCharacterTree}(G,H_2,r',\zeta)$\;

\ForEach{$x\in[n]$}{
    \ForEach{$C\subseteq [v]$}{
        $\DP_{H,r}(x,C)\gets 0$\;
    }
}

\ForEach{ordered pair $(a,b)\in[n]^2$ with $a\ne b$}{
    \ForEach{$C_1,C_2\subseteq [v]$ such that $C_1\cap C_2=\emptyset$}{
        $C\gets C_1\cup C_2$\;
        $\DP_{H,r}(a,C)\gets \DP_{H,r}(a,C)$
        $+\DP_{H_1,r}(a,C_1)\cdot \DP_{H_2,r'}(b,C_2)\cdot X_{ab}$\;
    }
}

\Return $\DP_{H,r}$\;
\end{algorithm}

We now extend the above framework to compute $T_H^{\mathrm{col}}(G, \zeta)$. Instead of storing only sets of achievable color sets by embeddings of $H$ mapping $r$ to each $x \in V(G)$, we replace $G$ with the complete graph and now also record the sum of $\prod_{e\in E(H')} X_{\phi(e)}$ over all colorful embeddings $\phi$ of the relevant rooted subtree $H'$ in the current recursive call, with its root mapped to $x$ and using exactly the color set $C$. The recursion is otherwise the same: when two partial embeddings are combined along an edge, the sums we are tracking are multiplied together, and then multiplied by the weight of the connecting edge. Thus the usual color-coding dynamic program extends in a straightforward way, except that Boolean OR of sets is replaced by signed addition, and the resulting tables compute the colorful embedding polynomial $T_H^{\mathrm{col}}(G,\zeta)$. See Algorithm~\ref{alg:tree-colorful-character}.

\begin{lemma}\label{lem:tree_algo}
Algorithm~\ref{alg:tree-colorful-character} computes, for every $x\in[n]$ and $C\subseteq [v]$, the quantity
\begin{equation}
\DP_{H,r}(x,C)
 \colonequals
\sum_{\substack{\phi:V(H)\hookrightarrow[n]\\ \phi(r)=x}}
\mathbf{1}\!\left\{
\begin{array}{l}
\zeta(\phi(V(H)))=C,\\
\zeta\circ\phi \text{ is injective on }V(H)
\end{array}
\right\}
\prod_{e\in E(H)} X_{\phi(e)}. \label{eq:DP-tree}
\end{equation}
Consequently, if $|V(H)|=v$, then
\[
T_H^{\mathrm{col}}(G,\zeta)
=
\sum_{x\in[n]} \DP_{H,r}(x,[v]).
\]
The total running time is $2^{O(v)}\cdot n^2$.
\end{lemma}
\noindent We give the proof in Appendix~\ref{apx:tree_algo}.

\paragraph{Color-coding DP for bounded-treewidth patterns.}
We now extend the signed color-coding framework from trees to a fixed pattern graph $H$ of treewidth at most $t$, the case relevant to our algorithms. Let $H$ be a graph on $v$ vertices with $\tw(H)\le t$, and fix a coloring $\zeta:[n]\to [v]$. Our objective is to compute the colorful embedding polynomial
\[
T_H^{\mathrm{col}}(G,\zeta)
 \colonequals
\sum_{\phi:V(H)\hookrightarrow[n]}
\mathbf{1}\{\zeta\circ \phi \text{ is injective on }V(H)\}
\prod_{\{a,b\}\in E(H)} X_{\phi(a)\phi(b)}.
\]

Assume a tree decomposition $(\mathsf T,\{B_x\subseteq V(H)\}_{x\in V(\mathsf T)})$ of $H$ of width at most $t$ is given, so every bag has size at most $t+1$. Choose any $u_r\in V(H)$ and root the given decomposition at a bag containing $u_r$. By a standard transformation (see, e.g., \cite{cygan2015parameterized,bodlaender1997treewidth,kloks1994treewidth}), we may further assume, without loss of generality, that this decomposition is a rooted \emph{nice} tree decomposition
\[
\mathcal D=(\mathcal T,\{B_x\}_{x\in V(\mathcal T)}),
\]
still of width at most $t$, with $O(v)$ bags, and with node types leaf, introduce, forget, and join. We additionally attach above its root a chain of forget nodes ending at $B_r=\{u_r\}$. This preserves the width and the $O(v)$ size bound. Concretely: a leaf node $x$ satisfies $|B_x|=1$; an introduce node $x$ has one child $y$ and $B_x=B_y\cup\{u_x\}$ for some $u_x\notin B_y$; a forget node $x$ has one child $y$ and $B_x=B_y\setminus\{u_x\}$ for some $u_x\in B_y$; and a join node $x$ has two children $y_1,y_2$ with $B_x=B_{y_1}=B_{y_2}$.

For each node $x\in V(\mathcal T)$, let $H_x$ denote the subgraph of $H$ induced by all vertices appearing in bags of the subtree of $\mathcal T$ rooted at $x$. The dynamic program keeps, for each node $x$, each injective assignment $\psi:B_x\hookrightarrow [n]$, and each color set $C\subseteq [v]$, the value
\[
\DP_x(\psi,C)
 \colonequals
\sum_{\phi:V(H_x)\hookrightarrow[n]}
\mathbf{1}\!\left\{
\begin{array}{l}
\phi|_{B_x}=\psi,\\
\zeta(\phi(V(H_x)))=C,\\
\zeta\circ\phi \text{ is injective on }V(H_x)
\end{array}
\right\}
\prod_{e\in E(H_x)\setminus E(H[B_x])} X_{\phi(e)}.
\]
Thus $\DP_x(\psi,C)$ is the signed contribution of all colorful embeddings of the partial graph $H_x$ that agree with $\psi$ on the current bag and use exactly the color set $C$. The product includes only those edges whose contribution has already been finalized below $x$; edges with both endpoints still inside the current bag are deferred until a later forget step.

If $r$ is the root of $\mathcal T$, then $B_r$ has size $1$. Therefore the full colorful embedding polynomial is recovered by summing over the possible images of the root bag vertex:
\[
T_H^{\mathrm{col}}(G,\zeta)
=
\sum_{z\in[n]} \DP_r(\psi_z,[v]),
\]
where $\psi_z:B_r\hookrightarrow[n]$ is the unique assignment sending the single vertex of $B_r$ to $z$. The algorithm runs bottom-up over the nice tree decomposition and is given in Algorithm~\ref{alg:treewidth-colorful-character}.

\begin{algorithm}[htbp]
\LinesNumbered
\DontPrintSemicolon
\caption{\textsc{ColorfulCharacterTreewidth}$(G,H,\zeta,\mathcal D)$}
\label{alg:treewidth-colorful-character}
\KwIn{
A graph $G=([n],E)$ with edge-weights $X_{ij}\in\{\pm1\}$;\\
a pattern graph $H$ on $v$ vertices with $\tw(H)\le t$;\\
a coloring $\zeta:[n]\to [v]$;\\
a rooted nice tree decomposition $\mathcal D=(\mathcal T,\{B_x\}_{x\in V(\mathcal T)})$ of $H$ of width at most $t$, normalized so that its root bag is $B_r=\{u_r\}$.
}
\KwOut{
The value $T_H^{\mathrm{col}}(G,\zeta)$.
}

Process the nodes of $\mathcal T$ in post-order\;

\ForEach{node $x\in V(\mathcal T)$}{

  \uIf{$x$ is a leaf node, say $B_x=\{u\}$}{
    \ForEach{injective assignment $\psi:B_x\hookrightarrow[n]$}{
      \ForEach{$C\subseteq [v]$}{
        $\DP_x(\psi,C)\gets 0$\;
      }
      $\DP_x(\psi,\{\zeta(\psi(u))\})\gets 1$\;
    }
  }

  \uElseIf{$x$ is an introduce node with child $y$, introducing $u$}{
    \ForEach{injective assignment $\psi:B_x\hookrightarrow[n]$}{
      \ForEach{$C\subseteq [v]$}{
        \uIf{$\zeta(\psi(u))\in C$}{
          $\DP_x(\psi,C)\gets
          \DP_y(\psi|_{B_y},C\setminus\{\zeta(\psi(u))\})$\;
        }
        \Else{
          $\DP_x(\psi,C)\gets 0$\;
        }
      }
    }
  }

  \uElseIf{$x$ is a forget node with child $y$, forgetting $u$}{
    \ForEach{injective assignment $\psi:B_x\hookrightarrow[n]$}{
      \ForEach{$C\subseteq [v]$}{
        $\DP_x(\psi,C)\gets 0$\;
        \ForEach{$z\in[n]\setminus \psi(B_x)$}{
          let $\psi':B_y\hookrightarrow[n]$ be the extension of $\psi$ with $\psi'(u)=z$\;
          let
          $
          m(u,\psi',\psi) \colonequals
          \prod_{\substack{w\in B_x\\ \{u,w\}\in E(H)}} X_{\psi'(u)\psi(w)}
          $\;
          $\DP_x(\psi,C)\gets
          \DP_x(\psi,C)+m(u,\psi',\psi)\cdot \DP_y(\psi',C)$\;
        }
      }
    }
  }

  \uElseIf{$x$ is a join node with children $y_1,y_2$}{
    \ForEach{injective assignment $\psi:B_x\hookrightarrow[n]$}{
      let $C_\psi \colonequals \zeta(\psi(B_x))$\;
      \ForEach{$C\subseteq [v]$}{
        $\DP_x(\psi,C)\gets 0$\;
        \ForEach{$C_1,C_2\subseteq [v]$ such that $C_1\cup C_2=C$ and $C_1\cap C_2=C_\psi$}{
          $\DP_x(\psi,C)\gets
          \DP_x(\psi,C)+\DP_{y_1}(\psi,C_1)\DP_{y_2}(\psi,C_2)$\;
        }
      }
    }
  }
}

Let $B_r=\{u_r\}$ be the root bag\;
$T\gets 0$\;
\ForEach{$z\in[n]$}{
  let $\psi_z:B_r\hookrightarrow[n]$ be given by $\psi_z(u_r)=z$\;
  $T\gets T+\DP_r(\psi_z,[v])$\;
}
\Return $T$\;
\end{algorithm}

\begin{lemma}\label{lem:treewidth_algo}
Algorithm~\ref{alg:treewidth-colorful-character} computes the colorful embedding polynomial
\begin{equation}
T_H^{\mathrm{col}}(G,\zeta)
=
\sum_{\phi:V(H)\hookrightarrow[n]}
\mathbf 1\{\zeta\circ\phi\text{ is injective on }V(H)\}
\prod_{\{a,b\}\in E(H)} X_{\phi(a)\phi(b)} \label{eq:THcol-treewidth}
\end{equation}
in time $2^{O(v)}\cdot n^{t+1}$. In particular, if $v=O(\log n/\log\log n)$ and $t$ is constant, then the running time is $n^{t+1+o(1)}$. Under the stronger condition $v=O(\log\log n)$, the per-pattern running time is $\widetilde O(n^{t+1})$.
\end{lemma}
\noindent The proof is deferred to Appendix~\ref{apx:treewidth-alg}.

We next obtain a better estimate of $T_H(G)$ by averaging over several random colorings. For a fixed embedding $\phi:V(H)\hookrightarrow[n]$, the event that $\phi(V(H))$ is colorful under a uniformly random coloring $\zeta:[n]\to [v]$ has probability $\rho_v$ as defined in \eqref{eq:rhov}. Therefore $T_H^{\mathrm{col}}(G,\zeta)$ has expectation $\rho_v T_H(G)$, so its rescaling by $1/\rho_v$ is unbiased for the full character sum. This motivates the estimator
\[
\widehat T_H(G)
 \colonequals
\frac{1}{\rho_v}\cdot \frac{1}{q}\sum_{s=1}^q T_H^{\mathrm{col}}(G,\zeta_s),
\]
where $\zeta_1,\dots,\zeta_q$ are i.i.d. uniform samples from $[v]^{[n]}$.

We take $q = \lceil 1 / \rho_v^2 \rceil$. Since $\rho_v=v!/v^v$, this gives $q=2^{O(v)}$. Combined with Lemma~\ref{lem:treewidth_algo}, the total running time for $\widehat T_H(G)$ is $q\cdot n^{t+1}2^{O(v)}=n^{t+1}2^{O(v)}$. Therefore, for the family of bounded treewidth graphs $\sH \subseteq \sG^{\conn}_{v, \leq t}$, the estimator
\[
\widehat F_{\sH}(G) \colonequals \sum_{H\in\sH}\widehat T_H(G)
\]
can be computed in time $|\sH|\cdot n^{t+1}\cdot 2^{O(v)}$. See Algorithm~\ref{alg:estimate-family-sum}.
For fixed $t$, Lemma~\ref{lem:enumerate-bounded-treewidth-patterns} below shows that the full family $\sG^{\conn}_{v,\leq t}$ and all required normalized decompositions can themselves be generated uniformly in time and space $\exp(O_t(v\log v))$.

\begin{algorithm}[t!]
\LinesNumbered
\DontPrintSemicolon
\caption{\textsc{EstimateFamilySum}$(G,\sH,q)$}
\label{alg:estimate-family-sum}
\KwIn{
A host graph $G=([n],E)$ with edge-weights $X_{ij}\in\{\pm1\}$;
a family $\sH$ of pattern graphs, each on $v$ vertices;
for each $H\in\sH$, a nice tree decomposition $\mathcal D_H$ of width at most $t$ with a singleton root bag; an integer $q\ge 1$.
}
\KwOut{
An estimator $\widehat F_{\sH}(G)$ for $F_{\sH}(G)$.
}

\BlankLine
Set
$\rho_v\gets v! / v^v$, and $\widehat F_{\sH}(G)\gets 0$.
Sample colorings $\zeta_1,\ldots,\zeta_q:[n]\to[v]$ independently and uniformly at random.

\ForEach{$H\in\sH$}{
    $S_H\gets 0$\;

    \For{$s\gets 1$ \KwTo $q$}{
        $S_H\gets S_H+\textsc{ColorfulCharacterTreewidth}(G,H,\zeta_s,\mathcal D_H)$\;
    }

    $\widehat T_H(G)\gets S_H / (\rho_v q)$\;
    $\widehat F_{\sH}(G)\gets \widehat F_{\sH}(G)+\widehat T_H(G)$\;
}

\Return $\widehat F_{\sH}(G)$\;
\end{algorithm}
The following lemma records the unbiasedness and concentration properties of this estimator.

\begin{lemma}\label{lemma_concentration}\label{lem:colorcoding-detection}
Fix an integer $t\geq 1$. Let $v=v(n)\to\infty$ satisfy
\[
v\leq \frac{\log n}{5\log\log n},
\]
let $\sH = \sG^{\conn}_{v, \leq t}$, and let $\widehat F_{\sH}(G)$ be computed by Algorithm~\ref{alg:estimate-family-sum} with $q=\lceil\rho_v^{-2}\rceil$. Then the estimator $\widehat F_{\sH}(G)$ is unbiased: for any given graph $G$,
\[
\E_{\zeta_1, \dots, \zeta_q}[\widehat F_{\sH}(G)]
=
F_{\sH}(G).
\]
Moreover, if $G \sim \PP_{n,k}$ with $k\geq\lambda\sqrt{n}$ and $\lambda > 1 / \sqrt{c(\sH)} = 1 / \sqrt{c(t)}$, and also when $G \sim \QQ_n$, we have the convergence
\[
\frac{\widehat F_{\sH}(G) - F_{\sH}(G)}{\E_{\PP}[\widehat F_{\sH}(G)]}
\xrightarrow{L^2} 0.
\]
\end{lemma}
\noindent The proof is deferred to Appendix~\ref{apx:concentration}.

\subsection{Approximating recovery statistics}
\label{sec:approx-recovery}

We now turn to the recovery statistic. Fix a distinguished pattern vertex $a\in V(H)$. For each host vertex $z\in[n]$, we would like to compute the signed contribution of all colorful embeddings $\phi$ satisfying $\phi(a)=z$. The key point is that this can be done with essentially the same dynamic program as in the detection case, without introducing any additional $n$-dimensional state.

Indeed, by a standard transformation, we may assume that the nice tree decomposition is rooted at a bag containing $a$, and that the root bag is $\{a\}$. By the connectedness property of tree decompositions, for every node $x$, either $a\in B_x$, or $a$ does not appear anywhere in the subtree rooted at $x$. Thus, whenever the image of $a$ is relevant, it is already determined by the current bag assignment $\psi$; no extra coordinate needs to be carried through the recursion. Consequently, all recovery values are obtained simultaneously from the root table with the same asymptotic running time as in the detection case. See Algorithm~\ref{alg:treewidth-recovery}.

\begin{algorithm}[t!]
\LinesNumbered
\DontPrintSemicolon
\caption{\textsc{ColorfulRecoveryTreewidth}$(G,H,a,\zeta,\mathcal D)$}
\label{alg:treewidth-recovery}
\KwIn{
A host graph $G=([n],E)$ with edge-weights $X_{ij}\in\{\pm1\}$;\\
a pattern graph $H$ on $v$ vertices with $\tw(H)\le t$;\\
a distinguished pattern vertex $a \in V(H)$;\\
a coloring $\zeta:[n]\to [v]$;\\
a rooted nice tree decomposition $\mathcal D=(\mathcal T,\{B_x\}_{x\in V(\mathcal T)})$ of $H$, with root bag $B_r=\{a\}$.
}
\KwOut{
For every $z\in[n]$, the value $T_{H,a}^{\mathrm{col}}(G,\zeta;z)$ defined in \eqref{eq:THaz}
}

Run the same dynamic program as in Algorithm~\ref{alg:treewidth-colorful-character}\;

\ForEach{$z\in[n]$}{
  let $\psi_z:B_r\hookrightarrow[n]$ be given by $\psi_z(a)=z$\;
  $T_{H,a}^{\mathrm{col}}(G,\zeta;z)\gets \DP_r(\psi_z,[v])$\;
}

\Return $\bigl(T_{H,a}^{\mathrm{col}}(G,\zeta;z)\bigr)_{z\in[n]}$\;
\end{algorithm}

\begin{lemma}\label{lem:treewidth_recovery_algo}
Algorithm~\ref{alg:treewidth-recovery} computes, for all $z\in[n]$, the colorful recovery polynomials
\begin{equation}
T_{H,a}^{\mathrm{col}}(G,\zeta;z)
=
\sum_{\phi:V(H)\hookrightarrow[n]}
\mathbf{1}\!\left\{
\begin{array}{l}
\zeta\circ\phi \text{ is injective on }V(H),\\
\phi(a)=z
\end{array}
\right\}
\prod_{\{u,v\}\in E(H)} X_{\phi(u)\phi(v)}. \label{eq:THaz}
\end{equation}
simultaneously in time $2^{O(v)} n^{t+1}$. More generally, if $\sH$ is a family of such pattern graphs, then the corresponding recovery scores $\sum_{H\in\sH}\ \sum_{a\in\mathcal R(H)} T_{H,a}^{\mathrm{col}}(G,\zeta;z)$ can be computed in time $|\sH|\cdot 2^{O(v)} n^{t+1}$.
\end{lemma}

\begin{proof}
Root the nice tree decomposition at a bag $\{a\}$. Then, by the connectedness property of tree decompositions, for every node $x$, either $a\in B_x$, or $a$ does not appear in the subtree rooted at $x$. Hence the image of the distinguished vertex $a$ is determined by the current bag assignment whenever it is relevant, and no additional $n$-dimensional bookkeeping is needed.

Therefore the recovery algorithm is exactly the same as the detection algorithm: it computes the same DP table $\DP_x(\psi,C)$ for every node $x$, bag assignment $\psi$, and color set $C$. At the root $r$, whose bag is $\{a\}$, the assignment $\psi_z(a)=z$ forces the distinguished vertex to map to $z$. Thus $\DP_r(\psi_z,[v])=T_{H,a}^{\mathrm{col}}(G,\zeta;z)$ for every $z\in[n]$, and all recovery values are read off simultaneously from the root table.

Since no new state parameter is introduced, the running time is identical to that in Lemma~\ref{lem:treewidth_algo}, namely $2^{O(v)} n^{t+1}$. Summing over all $H\in\sH$ and all $a\in\mathcal R(H)$ gives the second claim, since the additional factor $|\mathcal R(H)|\leq v$ is absorbed by $2^{O(v)}$.
\end{proof}

For the recovery statistic, we use the same sampling scheme as in the detection case. Namely, we sample $q= \lceil 1 / \rho_v^2 \rceil$ independent random colorings $\zeta_1,\dots,\zeta_q:[n]\to [v]$, compute the corresponding colorful recovery contributions for each $H\in\sH$, and then average and rescale by $1/\rho_v$. This gives an estimator $\widehat F_{\sH,i}(G)$ for each host vertex $i\in[n]$. See Algorithm~\ref{alg:estimate-family-recovery}. Since the recovery values for all $i\in[n]$ are obtained simultaneously from the same dynamic program, the total running time remains $|\sH|\cdot n^{t+1}2^{O(v)}$.
\begin{algorithm}[t!]
\LinesNumbered
\DontPrintSemicolon
\caption{\textsc{EstimateFamilyRecovery}$(G,\sH,q)$}
\label{alg:estimate-family-recovery}
\KwIn{
A host graph $G=([n],E)$ with edge-weights $X_{ij}\in\{\pm1\}$;\\
a family $\sH$ of pattern graphs, each on $v$ vertices;\\
for each $H\in\sH$ and each $a\in\mathcal R(H)$, a nice tree decomposition $\mathcal D_{H,a}$ of width at most $t$, rooted at the bag $\{a\}$;\\
an integer $q\ge1$.\\
}
\KwOut{
For every $i\in[n]$, an estimator $\widehat F_{\sH,i}(G)$ for $F_{\sH,i}(G)$.
}

\BlankLine
Set $\rho_v\gets v! / v^v$\;
Sample colorings $\zeta_1,\ldots,\zeta_q:[n]\to[v]$ independently and uniformly at random\;
\ForEach{$i\in[n]$}{
    $\widehat F_{\sH,i}(G)\gets 0$\;
}

\ForEach{$H\in\sH$}{
    \ForEach{$a\in\mathcal R(H)$}{
    \ForEach{$i\in[n]$}{
        $S_{H,a,i}\gets 0$\;
    }

    \For{$s\gets 1$ \KwTo $q$}{
        compute $(T_{H,a}^{\mathrm{col}}(G,\zeta_s;i))_{i\in[n]}$ using Algorithm~\ref{alg:treewidth-recovery}\;
        \ForEach{$i\in[n]$}{
            $S_{H,a,i}\gets S_{H,a,i}+T_{H,a}^{\mathrm{col}}(G,\zeta_s;i)$\;
        }
    }

    \ForEach{$i\in[n]$}{
        $\widehat T_{H,a}(G;i)\gets S_{H,a,i} / (\rho_v q)$\;
        $\widehat F_{\sH,i}(G)\gets \widehat F_{\sH,i}(G)+\widehat T_{H,a}(G;i)$\;
    }
}
}
\Return $(\widehat F_{\sH,i}(G))_{i\in[n]}$\;
\end{algorithm}

\begin{lemma}
\label{lemma_recovery_concentration}
\label{lem:colorcoding-recovery}
Let $\mathcal H=\mathcal G^{\mathrm{conn}}_{v,\le t}$, let $v=v(n)\to\infty$ satisfy $v=o(\log n/\log\log n)$, and let $\widehat F_{\mathcal H,i}$ be the estimator computed by Algorithm~\ref{alg:estimate-family-recovery} with $q=\lceil \rho_v^{-2}\rceil$. For every fixed graph $G$ and every $i\in[n]$,
\[
  \E_{\zeta_1,\ldots,\zeta_q}
  [\widehat F_{\mathcal H,i}(G)]
  =
  F_{\mathcal H,i}(G).
\]
Moreover, suppose $G\sim \PP_{n,k}$, $k\ge \lambda\sqrt n$, and $\lambda>1/\sqrt{c(t)}$. Fix $c_0<c(t)$ such that $c_0\lambda^2>1$. Let $\mu \colonequals \nu(\mathcal H)(k-1)_{v-1}$ be the planted conditional mean from Lemma~\ref{lem:rooted-moments}. For $b\in\{0,1\}$, let $\mathsf B_{i,b} \colonequals \{\mathbf 1_{\{i\in C\}}=b\}$.
Then, uniformly over all $i$ and $b$ for which $\Pr(\mathsf B_{i,b})>0$,
\[
  \E_{G,\zeta}
  \left[
    \left(
      \frac{\widehat F_{\mathcal H,i}(G)-F_{\mathcal H,i}(G)}{\mu}
    \right)^2
    \,\middle|\,\mathsf B_{i,b}
  \right]
  =O(\rho_v).
\]
Consequently, for $v=\lfloor\sqrt{\log n}\rfloor$,
\[
  \E_{G,\zeta}
  \left[
    \left(
      \frac{\widehat F_{\mathcal H,i}(G)}{\mu}
      -b
    \right)^2
    \,\middle|\,\mathsf B_{i,b}
  \right]
  \le
  \exp(-\Omega(\sqrt{\log n})).
\]
The same bound therefore holds without conditioning, with $b$ replaced by $\mathbf 1_{\{i\in C\}}$.
\end{lemma}
\noindent The proof is deferred to Appendix~\ref{apx:recovery-concentration}.

\begin{corollary}\label{cor:recovery-color-coding-supplies-cleanup}
For $\mathcal H=\mathcal G^{\mathrm{conn}}_{v,\le t}$, $v=\lfloor\sqrt{\log n}\rfloor$, and $k\ge \lambda\sqrt n$ with $\lambda>1/\sqrt{c(t)}$, Algorithm~\ref{alg:estimate-family-recovery} outputs the raw estimates $\widehat F_{\mathcal H,i}$. After the normalization $\widehat p_i:=\widehat F_{\mathcal H,i}/\mu$, they satisfy, uniformly over all $i$ and $b\in\{0,1\}$ for which $\Pr(\mathsf B_{i,b})>0$,
\[
\E\!\left[(\widehat p_i-b)^2\,\middle|\,\mathsf B_{i,b}\right]
\leq \exp(-\Omega(\sqrt{\log n})).
\]
In particular, these normalized estimates satisfy the hypothesis of Lemma~\ref{lem:recovery}, and the raw estimates can be computed in time $n^{t+1+o(1)}$.
\end{corollary}

\begin{proof}
The conditional mean-square guarantee is exactly the conclusion of Lemma~\ref{lemma_recovery_concentration} and gives the corresponding hypothesis of Lemma~\ref{lem:recovery}; averaging over the conditioning events having positive probability also gives the unconditional bound. The running time follows from Lemmas~\ref{lem:treewidth_recovery_algo} and~\ref{lem:enumerate-bounded-treewidth-patterns}, the choice $q=\lceil\rho_v^{-2}\rceil=2^{O(v)}=n^{o(1)}$, and the bound $|\mathcal H|=n^{o(1)}$ for $v=\lfloor\sqrt{\log n}\rfloor$ and fixed $t$.
\end{proof}

\begin{lemma}\label{lem:enumerate-bounded-treewidth-patterns}
For every fixed $t\geq 1$, all graphs in $\mathcal G^{\mathrm{conn}}_{v,\leq t}$ can be enumerated, together with a width-$t$ tree decomposition of each graph, in time and space $\exp(O_t(v\log v))$.
For every retained graph $H$ and every $a\in V(H)$, a nice decomposition normalized to have singleton root bag $\{a\}$ can be constructed within the same total bound.
\end{lemma}
\noindent The proof is deferred to Appendix~\ref{apx:enumerate-bounded-treewidth-patterns}.

Combining Lemma~\ref{lem:holdout-cleanup-generic} with Corollary~\ref{cor:recovery-color-coding-supplies-cleanup} proves the recovery bound in Theorem~\ref{thm:intro-main-general} for bounded-treewidth patterns when the clique location is random. The detection bound follows from Theorem~\ref{thm:general-separation} and Lemma~\ref{lem:colorcoding-detection}. Thus this section proves the randomized bounds obtained by applying color coding to bounded-treewidth patterns. To obtain deterministic algorithms, the argument in Section~\ref{sec:derandomization} replaces random colorings by balanced hash families. When the algorithm must work for every clique location, the results there also replace the random holdout split by a splitter family.
\section{Matrix multiplication acceleration}
\label{sec:mm-acceleration}

In this section, we use matrix multiplication to improve the above runtime bounds. Square matrix multiplication speeds up the computation of the full signed subgraph count statistics for graphs of treewidth at most $2$, while rectangular matrix multiplication computes seeded versions of these statistics for trees (of treewidth 1) for many collections of seed vertices at once. We write $\tau=\tau_{\mathrm{rand}}$ below unless otherwise stated.

\subsection{Square matrix multiplication acceleration}
\label{sec:mm-square}

As in Section~\ref{sec:approx-detection}, for a fixed coloring $\zeta:[n]\to[v]$, color coding computes
\[
T_H^{\mathrm{col}}(X,\zeta) = \sum_{\phi:V(H)\hookrightarrow[n]} \mathbf 1\{\zeta\circ\phi \text{ is injective}\}
\prod_{\{a,b\}\in E(H)} X_{\phi(a)\phi(b)}.
\]
The purpose of using matrix multiplication is that, once the coloring is fixed, it speeds up the dynamic program used to evaluate the colorful statistic for patterns of treewidth at most $2$. A $2$-tree is a graph built recursively as follows. Start from a single edge. At each step, add one new vertex and connect it to both endpoints of some existing edge. Equivalently, every new vertex is added on top of an existing edge, creating a new triangle. A graph has treewidth at most $2$ if and only if it is a subgraph of a $2$-tree; such graphs are sometimes called \emph{partial $2$-trees}.

\begin{lemma}\label{lem:enumerate-partial-2trees}
For integers $v\ge2$ satisfying $v=o(\log n/\log\log n)$, one may enumerate all connected labeled graphs $H\in\mathcal G^{\mathrm{conn}}_{v,\le 2}$ together with one $2$-tree completion $K_H$ and one rooted construction history of $K_H$, using $n^{o(1)}$ preprocessing time and space.
\end{lemma}

\begin{proof}
A graph has treewidth at most $2$ if and only if it is a partial $2$-tree, that is, a subgraph of a $2$-tree on the same vertex set. We enumerate labeled $2$-trees by construction histories. Such a history is specified by choosing an initial edge, which we regard as the root edge, ordering the remaining $v-2$ vertices, and, when each new vertex is introduced, choosing one of the current edges as its parent edge. At every step the current graph has at most $2v$ edges, so the number of such histories is at most
\[
\binom v2 (v-2)! (2v)^v=\exp(O(v\log v)).
\]
Each $2$-tree on $v$ vertices has exactly $2v-3$ edges, and therefore has at most $2^{2v-3}=\exp(O(v))$
spanning subgraphs. Thus, by enumerating all construction histories and all spanning subgraphs of the resulting $2$-trees, we produce every labeled partial $2$-tree on vertex set $[v]$, possibly with multiplicity. We then discard disconnected subgraphs and deduplicate the remaining graphs by the bit vectors of their edge sets, keeping, for each surviving graph $H$, one $2$-tree completion $K_H$ and the rooted construction history of $K_H$ from which it was produced. The number of generated objects is $\exp(O(v\log v))$, and each bit vector has length $O(v^2)$, so both the time and space used by this preprocessing are $\exp(O(v\log v))$. Since $v=o(\log n/\log\log n)$, this is $n^{o(1)}$.
\end{proof}

\begin{lemma}\label{lem:one-partial-2tree-mm}
Fix a connected partial $2$-tree $H$ on $[v]$, a $2$-tree completion $K$ of $H$, a rooted construction history of $K$, and a coloring $\zeta:[n]\to[v]$. Then $T_H^{\mathrm{col}}(X,\zeta)$ and, for any root vertex $r\in V(H)$, the vector $\bigl(T^{\mathrm{col}}_{H,r}(X,\zeta;i)\bigr)_{i\in[n]}$ can be computed in time $2^{O(v)}n^{\omega+o(1)}$.
\end{lemma}
\noindent
The proof is deferred to Appendix~\ref{app:proof-one-partial-2tree-mm}.

\begin{theorem}\label{thm:tw2-mm}
Let $v=\lfloor\sqrt{\log n}\rfloor$ and $\mathcal H_v=\mathcal G^{\mathrm{conn}}_{v,\le2}$.
Then the detection statistic $F_{\mathcal H_v}$ and the rooted recovery vector $(F_{\mathcal H_v,i})_{i\in[n]}$ can be estimated in time $n^{\omega+o(1)}$.
Consequently,
\[
\tau(1/\sqrt{c(2)})\le \omega.
\]
\end{theorem}
\noindent
We recall that the results of \cite{bodirsky2007enumeration} imply that $1/\sqrt{c(2)} \approx 0.3320$.

\begin{proof}
By Lemma~\ref{lem:enumerate-partial-2trees}, all connected partial $2$-trees $H\in\mathcal G^{\mathrm{conn}}_{v,\le2}$, together with one $2$-tree completion and one rooted construction history, can be enumerated in $n^{o(1)}$ time and space. Since $v=\lfloor\sqrt{\log n}\rfloor$, we have $v\to\infty$ and $v=o(\log n/\log\log n)$, so the moment separation from Theorem~\ref{thm:general-separation} applies to $\mathcal H_v$ whenever $\lambda>1/\sqrt{c(2)}$.

For each $H\in\mathcal H_v$ and each coloring $\zeta:[n]\to[v]$, Lemma~\ref{lem:one-partial-2tree-mm} computes $T_H^{\mathrm{col}}(X,\zeta)$ in time $2^{O(v)}n^{\omega+o(1)}$. The same lemma computes, for any distinguished root $r\in V(H)$, the full rooted vector $(T_{H,r}^{\mathrm{col}}(X,\zeta;i))_{i\in[n]}$ in the same time. Repeating this over all roots in $\mathcal R(H)$ contributes only a factor $v=n^{o(1)}$.

We use the usual color-coding estimator with $q=\lceil\rho_v^{-2}\rceil$ independent colorings, where $\rho_v=v!/v^v$. Since $\rho_v=\exp(-O(v))$, our choice of $v$ gives $q=2^{O(v)}=n^{o(1)}$. Also, Lemma~\ref{lem:enumerate-partial-2trees} implies $|\mathcal H_v|=n^{o(1)}$. Hence summing the estimates over all $H\in\mathcal H_v$, over all colorings, and, for recovery, over all roots in $\mathcal R(H)$ still takes total time $n^{\omega+o(1)}$.

For detection, Lemma~\ref{lemma_concentration} shows that the color-coded estimator is asymptotically equivalent to $F_{\mathcal H_v}$ in the required $L^2$ sense, and Theorem~\ref{thm:general-separation} gives strong separation. Thus the resulting test succeeds whenever $\lambda>1/\sqrt{c(2)}$.

For recovery, first choose, independently of the graph, a uniformly random partition $[n]=V_1\sqcup V_2$ with $|V_2|=\lfloor n^{3/5}\rfloor$, and put $N=|V_1|$.  Compute all rooted scores using only $G[V_1]$.  Conditional on the partition and on its planted vertices, $G[V_1]$ is again a planted clique instance; with probability $1-o(1)$, $|C\cap V_1|=(1-o(1))k$ and $|C\cap V_2|\gg\log n$.  Lemma~\ref{lemma_recovery_concentration}, together with the rooted moment bound, therefore gives normalized estimates satisfying, for every such realization of $(V_1,C\cap V_1)$ having positive probability,
\[
\E_{G,\zeta}\left[\left(\widehat p_i-\mathbf 1_{\{i\in C\cap V_1\}}\right)^2\mid V_1,C\cap V_1\right]\le \exp(-\Omega(v))=\exp(-\Omega(\sqrt{\log n}))
\]
uniformly in $i\in V_1$.  The normalizing factor is positive and common to all $i$, but may depend on the unknown value $|C\cap V_1|$.
The algorithm therefore takes the
\[
M_{N,k}=\min\left\{N,\left\lceil k+\frac{N}{(\log N)^3}\right\rceil\right\}
\]
largest raw scores, with deterministic tie breaking.
To analyze this, we threshold the normalized scores strictly above $1/2$.
The mean-square bound and Markov's inequality show that this threshold set has size at most $|C\cap V_1|+o(N/(\log N)^3)\le k+o(N/(\log N)^3)$ and contains all but $o(k)$ clique vertices with probability $1-o(1)$.
If $M_{N,k}=N$, containment in the top set is immediate; otherwise the threshold set has size less than $M_{N,k}$.  Hence the top-$M_{N,k}$ set contains all but $o(k)$ vertices of $C\cap V_1$.  Moreover,
\[
\frac{k^2}{M_{N,k}}
\ge
\Omega\!\left(\min\left\{k,\frac{k^2(\log N)^3}{N}\right\}\right)
=\Omega((\log n)^3).
\]
Because the partition was chosen before scoring and the scores use only edges inside $V_1$, the non-clique edges between $V_1$ and $V_2$ remain independent of this set conditional on $C$.  Applying the cleanup argument of Lemma~\ref{lem:recovery} now gives exact recovery. The cleanup costs only $O(n^2)$ time, which is dominated by $n^{\omega+o(1)}$. Thus the resulting randomized algorithms achieve strong detection and strong recovery for every $\lambda>1/\sqrt{c(2)}$ in time $n^{\omega+o(1)}$.
\end{proof}

Numerically, using the bounds $\omega<2.371177$ from \cite{dupont2026improving} and $1 / \sqrt{c(2)} < 0.3320$ from \cite{bodirsky2007enumeration}, we find
\[
\tau(0.3320) < 2.371177,
\]
compared to $\tau(0.3320) \leq 3$ that the unaccelerated color-coding computation would give us.

\begin{remark}
The choice $v=\lfloor\sqrt{\log n}\rfloor$ keeps all other runtime costs having subpolynomial order. The moment estimates only require $v\to\infty$ and $v=o(\log n/\log\log n)$, while enumerating partial $2$-trees costs $\exp(O(v\log v))$ and the color-coding repetitions contribute $\rho_v^{-2}=\exp(O(v))$. Thus any choice with $v\to\infty$ and $v\log v=o(\log n)$ would suffice for the $n^{\omega+o(1)}$ bound. By contrast, taking $v=\Theta(\log n/\log\log n)$ would make the enumeration cost polynomial in $n$, rather than $n^{o(1)}$.
\end{remark}

\begin{remark}\label{rem:why-only-tw2-mm}
    This use of square matrix multiplication is special to treewidth at most $2$. In the dynamic program for partial $2$-trees, every subproblem is indexed by an edge with two boundary vertices, so each table is an $n\times n$ matrix, and gluing two pieces along one vertex becomes an ordinary matrix product. For treewidth $t\ge 3$, the corresponding tables have $t$ boundary indices, and gluing two pieces requires a tensor contraction rather than an ordinary matrix product. A similar improvement would therefore require a faster method for this tensor contraction.
    We leave further investigation of this possibility as an intriguing question for future work.
\end{remark}

\subsection{Rectangular matrix multiplication acceleration}
\label{sec:mm-rectangular}

We next describe an application of fast rectangular matrix multiplication to an algorithm that combines our signed subgraph count technique with the boosting scheme of \cite{alon1998finding}, as detailed earlier.
Recall that the general premise is to run some algorithm (in our case our polynomial evaluations) on the induced subgraphs on mutual neighbors of a given set of \emph{seed vertices}.
The term ``seed'' always refers to such distinguished subsets below.

We first show how such a seed subgraph count may itself be expressed as a polynomial, in the case of trees that we will focus on in our algorithm.
Let $\mathcal T_v$ denote the family of labeled trees on $[v]$. Consider an ordered tuple of distinct seed vertices $Q=(q_1,\ldots,q_s)$.
We simply call these \emph{seed tuples} below.
Define
\[
w_Q(x)
=
\mathbf 1\{x\notin Q\}
\prod_{a=1}^s\frac{1+X_{q_ax}}2.
\]
Thus $w_Q(x)=1$ precisely when $x$ is adjacent to every seed vertex in $Q$, and $w_Q(x)=0$ otherwise. For $i\in[n]\setminus Q$, define the (rooted) \emph{seeded tree statistic}
\[
F_{Q,\mathrm{tree},i}(X)
=
\sum_{T\in\mathcal T_v}
\sum_{r\in\mathcal R(T)}
\sum_{\substack{\phi:V(T)\hookrightarrow[n]\setminus Q\\ \phi(r)=i}}
\left(
\prod_{u\in V(T)}w_Q(\phi(u))
\right)
\left(
\prod_{\{u,v\}\in E(T)}
X_{\phi(u)\phi(v)}
\right).
\]
We are then interested in computing this over many choices of seed $Q$.

Indeed, for seeded trees, it turns out that the dynamic program can be evaluated for many seeds at once. We view the seeds as indexing the rows of an $R\times |V_1|$ table, and each merge with a child subtree becomes a rectangular product with the adjacency matrix of $G[V_1]$. When $R=n^{s/2+o(1)}$, this gives the exponent $\omega(1,1,s/2)$.

We will use the following facts, easily shown by first moment method calculations.
Uniformly for $k\ge\lambda\sqrt n$ with fixed $\lambda>0$, the graph $\QQ_n$ contains no $k$-clique with high probability. And, uniformly in this range, the planted clique $C$ is the unique $k$-clique with high probability under $\PP_{n, k}$. Thus any clique of size $k$ that we find shows, with high probability, that our observation is drawn from the planted model and further must be equal to $C$.

\begin{proposition}\label{prop:randomized-seed-boosting}\label{prop:random-seed-boosting}
Suppose that for some $\lambda_0>0$ and $r\ge2$, there is a randomized algorithm running in $n^{r+o(1)}$ time that strongly recovers planted cliques of size $k\ge(\lambda_0+\varepsilon)\sqrt n$ for every fixed $\varepsilon>0$. Then for every fixed integer $a\ge0$,
\[
\tau(2^{-a/2}\lambda_0)\le r+\frac a2.
\]
Equivalently, if a base graph family has exponential base $c$ and can be used for recovery in time $n^{r+o(1)}$, then
\[
\tau\!\left((2^a c)^{-1/2}\right)\le r+\frac a2.
\]
\end{proposition}
\noindent
The proof is deferred to Appendix~\ref{app:proof-randomized-seed-boosting}.

\begin{algorithm}[t!]
\LinesNumbered
\DontPrintSemicolon
\caption{\textsc{BatchedSeededTreeRecovery}$(G,k,s,v,R)$}
\label{alg:batched-seeded-tree-recovery}
\KwIn{A graph $G=([n],E)$; a target clique size $k$; a seed size $s$; a tree size $v$; a number of seeds $R$}
\KwOut{A $k$-clique, or report failure to find a $k$-clique}
Set $n_2\gets \lfloor n^{3/5}\rfloor$ and choose, independently of $G$, a uniformly random partition $[n]=V_1\sqcup V_2$ with $|V_2|=n_2$\;
Independently sample, with replacement, $R$ uniformly random ordered $s$-tuples $Q$ of distinct vertices from $V_1$\;
Set $\rho_v\gets v!/v^v$ and $q\gets\lceil\rho_v^{-2}\rceil$, and sample colorings $\zeta_1,\ldots,\zeta_q:V_1\to[v]$ independently and uniformly, independently also of $G$, the partition, and the sampled seeds\;
For every sampled $Q$ and every $i\in V_1$, initialize $\widehat F_{Q,\mathrm{tree},i}\gets0$\;
For every $\ell\in[q]$, $T\in\mathcal T_v$, and $r\in\mathcal R(T)$, use the batched fixed-coloring dynamic program on $G[V_1]$ to compute $T^{\mathrm{col}}_{Q,T,r}(G,\zeta_\ell;i)$ simultaneously for all sampled $Q$ and $i\in V_1$, and add $T^{\mathrm{col}}_{Q,T,r}(G,\zeta_\ell;i)/(\rho_v q)$ to $\widehat F_{Q,\mathrm{tree},i}$\;
Set $N\gets |V_1|$ and $M\gets\min\{N,\lceil k+N/(\log N)^3\rceil\}$\;
For each sampled seed $Q$, form $V_3(Q)\subseteq V_1$ by taking the $M$ largest unnormalized scores, breaking ties by a fixed deterministic rule\;
Use the batched products in Lemma~\ref{lem:batched-cleanup} to recover candidate sets $\widehat C(Q)$ for all sampled seeds\;
\Return the first candidate $\widehat C(Q)$ such that $|\widehat C(Q)|=k$ and $G[\widehat C(Q)]$ is a clique; if no such seed exists, return failure\;
\end{algorithm}

\begin{lemma}\label{lem:good-seeded-row}
Fix $s=O(1)$, let $v=\lfloor\sqrt{\log n}\rfloor$, assume $k\ge\lambda\sqrt n$ for a fixed $\lambda>0$, and let $Q\subseteq C\cap V_1$ be a tuple of seed vertices. Define
\[
W=W_Q(V_1)=\{x\in V_1\setminus Q:\ x\text{ is adjacent to every vertex of }Q\}
\]
and $C_W=(C\cap V_1)\setminus Q$. Let $\widehat p_{Q,i}$ denote the color-coded estimate $\widehat F_{Q,\mathrm{tree},i}$ of the seeded tree statistic rooted at $i$, divided by the common planted-root mean $\mu_W=\nu(\mathcal T_v)(|C_W|-1)_{v-1}$. Suppose that, for some fixed constants $0<\alpha<1$ and $\eta>0$,
\[
|W|\ge n^\alpha, \qquad |C_W|\ge (e^{-1/2}+\eta)\sqrt{|W|}.
\]
Then, for every realization of $(W,C_W)$ satisfying these bounds
and having positive probability conditional on $Q\subseteq C\cap V_1$, the normalized color-coded estimates $\widehat p_{Q,i}$ satisfy
\[
\E\left[
\left(
\widehat p_{Q,i}-\mathbf 1_{\{i\in C_W\}}
\right)^2
\mid W,C_W
\right]
\le \delta_v,
\qquad
\delta_v=\exp(-\Omega(v)),
\]
uniformly over $i\in W$. Also, put $N=|V_1|$ and
\[
M_{N,k}=\min\left\{N,\left\lceil k+\frac{N}{(\log N)^3}\right\rceil\right\}.
\]
Then with conditional probability $1-o(1)$ the set $V_3(Q)$ of the $M_{N,k}$ largest unnormalized scores contains all but $o(k)$ vertices of $C_W$ and satisfies $|V_3(Q)|=M_{N,k}$ and $k^2/M_{N,k}=\Omega(\log^3 n)$.
\end{lemma}
\noindent
The proof is deferred to Appendix~\ref{app:proof-good-seeded-row}.

\begin{lemma}\label{lem:batched-seeded-tree-dp}
Fix $s=O(1)$, $R=n^{s/2+o(1)}$ ordered seed tuples $Q$, a coloring $\zeta:V_1\to[v]$, and a rooted tree $T$ on $v$ vertices. For every root $r\in V(T)$, the matrix
\[
\left(T^{\mathrm{col}}_{Q,T,r}(G,\zeta;i)\right)_{Q,i}
\]
whose rows are indexed by seed tuples and whose columns are indexed by vertices
can be computed in time $2^{O(v)}n^{\omega(1,1,s/2)+o(1)}$.
\end{lemma}
\noindent
The proof is deferred to Appendix~\ref{app:proof-batched-seeded-tree-dp}.

\begin{lemma}\label{lem:batched-cleanup}
Fix $\lambda>0$ and assume $k\ge\lambda\sqrt n$. Let $R=n^{s/2+o(1)}$, and let $[n]=V_1\sqcup V_2$ with $|V_2|=\lfloor n^{3/5}\rfloor$. Suppose the partition satisfies
\[
|C\cap V_1|=(1-o(1))k,
\qquad
|C\cap V_2|\ge \Omega(n^{1/10}).
\]
Suppose an algorithm has computed, for all sampled seed tuples $Q\subseteq V_1$, rooted score vectors $S(Q,i)$ for $i\in V_1$, and that these scores are functions only of edges inside $V_1$ together with the edges between the seed vertices and $V_1$ that define the seed masks. For each sampled seed, form $V_3(Q)$ by selecting the $M_{N,k}$ largest scores as in Algorithm~\ref{alg:batched-seeded-tree-recovery}, where $N=|V_1|$ and $M_{N,k}=\min\{N,\lceil k+N/(\log N)^3\rceil\}$. Suppose that for at least one seed $Q^\star\subseteq C\cap V_1$, the set $V_3(Q^\star)$ contains all but $o(k)$ vertices of $C\cap V_1$. Then the cleanup and checking whether we have found a clique or not, for all sampled seeds, can be performed in $n^{\omega(1,1,s/2)+o(1)}$ additional time, and the seed $Q^\star$ outputs the planted clique with probability $1-o(1)$.
\end{lemma}
\noindent
The proof is deferred to Appendix~\ref{app:proof-batched-cleanup}.

\begin{theorem}
\label{thm:seeded-tree-rectangular-mm}
\label{thm:seeded-rectangular}
For every fixed integer $s\ge0$,
\[
\tau\!\left((2^s e)^{-1/2}\right)
\le
\omega(1,1,s/2).
\]
More generally,
\[
\tau(\lambda)
\le
\min_{s\in\ZZ_{\ge0}}
\left[
\omega(1,1,s/2)
+
\frac12
\left\lceil
\log_2\frac{1}{\lambda^2\,2^s e}
\right\rceil_+
\right].
\]
\end{theorem}

\begin{proof}
Fix an integer $s\ge0$, and put
\[
\lambda_s=(2^s e)^{-1/2}.
\]
It suffices to prove recovery for every fixed $\lambda>\lambda_s$ in time $n^{\omega(1,1,s/2)+o(1)}$. Run Algorithm~\ref{alg:batched-seeded-tree-recovery} with $v=\lfloor\sqrt{\log n}\rfloor$. Put $N=|V_1|$ and let
\[
M_{N,k}=\min\left\{N,\left\lceil k+\frac{N}{(\log N)^3}\right\rceil\right\}
\]
be the truncation size used in Algorithm~\ref{alg:batched-seeded-tree-recovery}. By the same hypergeometric concentration used in Lemma~\ref{lem:holdout-partition}, the random partition satisfies
\[
|C\cap V_1|=(1-o(1))k,
\qquad
|C\cap V_2|\ge \Omega(n^{1/10})
\]
with probability $1-o(1)$. Condition on this event. If $s=0$, use the single empty seed tuple. If $s\ge1$, sample $R=n^{s/2}\log n$ ordered seed tuples from $V_1$. Since $|C\cap V_1|=(1-o(1))k$, a uniform ordered $s$-tuple from $V_1$ is contained in $C\cap V_1$ with probability
\[
\frac{(|C\cap V_1|)_s}{(|V_1|)_s}
\ge
(1-o(1))\lambda^s n^{-s/2}.
\]
Thus, with probability $1-o(1)$, at least one sampled seed $Q^\star$ satisfies $Q^\star\subseteq C\cap V_1$.

For this seed, put $K=|(C\cap V_1)\setminus Q^\star|$. With probability $1-o(1)$ over the exposed edges between $Q^\star$ and $V_1\setminus Q^\star$,
\[
K=(1-o(1))k,
\qquad
|W_{Q^\star}(V_1)|
=(1+o(1))\bigl(2^{-s}N+(1-2^{-s})K\bigr).
\]
The function $x\mapsto x/\sqrt{2^{-s}N+(1-2^{-s})x}$ is increasing. Since $N\le n$ and $K\ge(1-o(1))\lambda\sqrt n$, the effective signal in the induced subgraph on the common neighbors of the seed set $Q^\star$ in $V_1$ is therefore at least
\[
\frac{K}{\sqrt{|W_{Q^\star}(V_1)|}}
\ge
(1-o(1))\frac{\lambda\sqrt n}{\sqrt{2^{-s}n+(1-2^{-s})\lambda\sqrt n}}
=
(1-o(1))\lambda 2^{s/2}.
\]
Since $\lambda>\lambda_s$, this is larger than $e^{-1/2}$ by a fixed margin. Lemma~\ref{lem:good-seeded-row} therefore shows that the set $V_3(Q^\star)$ of the $M_{N,k}$ largest scores contains all but $o(k)$ vertices of $C_W=(C\cap V_1)\setminus Q^\star$. Since $|Q^\star|=s=O(1)=o(k)$, the same set contains all but $o(k)$ vertices of $C\cap V_1$ and has the signal-to-noise ratio required for cleanup using the holdout set.

It remains to account for the runtime. For each coloring, rooted tree, and root, Lemma~\ref{lem:batched-seeded-tree-dp} computes the score matrix in time $2^{O(v)}n^{\omega(1,1,s/2)+o(1)}$.
The number of labeled trees and admissible roots is at most $v^{v-2}\cdot v=n^{o(1)}$, and the number of color-coding repetitions is $\lceil\rho_v^{-2}\rceil=2^{O(v)}=n^{o(1)}$. Thus computing all scores takes time $n^{\omega(1,1,s/2)+o(1)}$. By Lemma~\ref{lem:batched-cleanup}, the cleanup and checking for cliques over all sampled seeds take the same asymptotic time, and the seed $Q^\star$ outputs the planted clique with probability $1-o(1)$.

Strong detection follows by running the recovery algorithm and accepting if and only if a clique of the target size $k$ is found. Uniformly for $k\ge\lambda\sqrt n$, under $\QQ_n$ no such $k$-clique exists with high probability, while under $\PP_{n,k}$ the recovery argument outputs $C$. This proves
\[
\tau(\lambda_s)\le \omega(1,1,s/2).
\]
For a general $\lambda>0$, fix an integer $s\ge0$. If $\lambda<\lambda_s$, apply Proposition~\ref{prop:randomized-seed-boosting} with an additional seed of size
\[
a=
\left\lceil
\log_2\frac{\lambda_s^2}{\lambda^2}
\right\rceil_+
=
\left\lceil
\log_2\frac{1}{\lambda^2 2^s e}
\right\rceil_+.
\]
By the definition of $a$,
\[
2^{-a/2}\lambda_s\le\lambda.
\]
Proposition~\ref{prop:randomized-seed-boosting} bounds $\tau(2^{-a/2}\lambda_s)$ by the base exponent plus $a/2$, and monotonicity of the tradeoff in the signal strength then gives
\[
\tau(\lambda)
\le
\omega(1,1,s/2)+\frac a2.
\]
Optimizing over $s$ gives the result.
\end{proof}

\begin{remark}
    In principle one may extend this idea beyond just counting signed trees.
    Trees are special in this setting because the rooted dynamic program has one boundary vertex: after putting $R=n^{s/2+o(1)}$ seeds into one table, each transition is a product of shape $(R\times n)\cdot(n\times n)$, giving the exponent $\omega(1,1,s/2)$.

    For a bounded-treewidth family, the tables have several boundary vertices. After putting the seeds into one table and combining the boundary indices, the transitions would involve more general rectangular products, for example $(n^{s/2+a}\times n^b)\cdot(n^b\times n^c)$, with $a,b,c$ depending on the chosen decomposition. Thus one would need a separate algorithm using exponents such as $\omega(s/2+a,b,c)$ and, for larger treewidth, possibly faster methods for tensor contractions. We do not pursue these variants here.
\end{remark}
\section{Derandomization of algorithms}
\label{sec:derandomization}

This section gives the methods needed to make the above algorithms deterministic.
Recall that there are three sources of randomness that have been involved: (1) the random colorings in color-coding, (2) the sets of seed vertices in boosted algorithms, and (3) the choice of holdout sets in recovery algorithms.
We describe how each of these can be derandomized below.

\subsection{Balanced hash families for color coding}
\label{sec:derand}

We first describe how the color-coding approximations above can be derandomized by replacing random colorings with explicit families of hash functions.
This follows an argument of \cite{alon2010balanced}.

The key point is that, for our color-coding argument, we do not need every $v$-subset of vertices to be uniformly colored by every hash function; rather, it suffices to guarantee that each such subset is assigned distinct colors by roughly the same number of hash functions. This leads to the notion of a balanced family of perfect hash functions.

\begin{definition}
Let $1 \le v \le n$ and $\delta \ge 1$. A family $\mathcal{T}$ of functions $\zeta:[n]\to[v]$ is called a \emph{$\delta$-balanced $(n,v)$-family of perfect hash functions} if there exists a real number $T>0$ such that for every $S \in \binom{[n]}{v}$, the number of functions in $\mathcal{T}$ that are injective on $S$,
\[
N_{\mathcal T}(S)\colonequals
\bigl|\{\zeta\in\mathcal{T}:\ \zeta|_S \text{ is one-to-one}\}\bigr|,
\]
satisfies
\[
\frac{T}{\delta} \le N_{\mathcal T}(S) \le \delta T.
\]
\end{definition}

\begin{theorem}[Theorem~7 of \cite{alon2010balanced}]
\label{thm:balanced-phf-main}
Let $1<\delta\le 2$. There is an explicit construction of a $\delta$-balanced $(n,v)$-family of perfect hash functions of size
\[
2^{O(v\log\log v)}\cdot (\delta-1)^{-O(\log v)}\cdot \log n,
\]
and the construction, together with a balancing parameter $T$ satisfying the definition above, can be produced in time
\[
2^{O(v\log\log v)}\cdot (\delta-1)^{-O(\log v)}\cdot n\log n
 +(\delta-1)^{-O(v/\log v)}.
\]
In particular, for any fixed $1<\delta\le2$, the size is $2^{O(v\log\log v)}\log n$ and the construction time is $2^{O(v\log\log v)}n\log n$.
\end{theorem}

When $v=o(\log n/\log\log n)$, the family size is $n^{o(1)}$ for every fixed $\delta>1$. To derandomize rooted recovery below, where $v=\lfloor\sqrt{\log n}\rfloor$, we instead take $\delta=1+\eta_v$ with $\eta_v=\exp(-\Theta(v))$ small enough. The factors involving $(\delta-1)^{-1}$ are then still $n^{o(1)}$. Let $\mathcal T$ be such a balanced family, viewed as colorings $\zeta:[n]\to[v]$. Define
\[
\widetilde{T}_H(G)
=
\frac{1}{\beta_v|\mathcal T|}
\sum_{\zeta\in\mathcal T}T_H^{\mathrm{col}}(G,\zeta),
\]
where $\beta_v=T/|\mathcal T|$. For each injective embedding $\phi:V(H)\hookrightarrow[n]$, the weight of the corresponding embedding term in $\widetilde T_H$ lies in $[1/\delta,\delta]$. For a family $\sH$, define
\[
\widetilde F_{\sH}(G)=\sum_{H\in\sH}\widetilde T_H(G).
\]
The derandomization increases the running time by only a subpolynomial factor.

\begin{lemma}\label{lem:derand-second-moment}\label{lem:balanced-hash-detection}\label{lem:balanced-hash-variance}
Let $\sH=\sG_{v,\le t}^{\conn}$, and let $\mathcal T$ be a $\delta$-balanced $(n,v)$-family of perfect hash functions with $1\le\delta\le2$. Then
\begin{align*}
\E_{\QQ}[\widetilde F_{\sH}] &= \E_{\QQ}[F_{\sH}] = 0, \\
\frac{1}{\delta}\E_{\PP}[F_{\sH}]
&\le \E_{\PP}[\widetilde F_{\sH}]
\le \delta \E_{\PP}[F_{\sH}], \\
\Var_{\QQ}(\widetilde F_{\sH})
&\le \delta^2\Var_{\QQ}(F_{\sH}).
\end{align*}
Moreover, if $G\sim\PP$ with $k\ge\lambda\sqrt n$, $\lambda>1/\sqrt{c(\sH)}$, $v^2=o(k)$, and
\[
\frac{\Var_{\PP}(F_{\sH})}{(\E_{\PP}[F_{\sH}])^2}=o(1),
\]
then
\[
\frac{\Var_{\PP}(\widetilde F_{\sH})}
{(\E_{\PP}[\widetilde F_{\sH}])^2}=o(1).
\]
The planted-mean bounds and the final normalized-variance conclusion also hold with $\PP$ replaced by $\PP_{n,C}$, uniformly over every fixed $C\subseteq[n]$ of size $k$.
\end{lemma}
\noindent
The proof is deferred to Appendix~\ref{app:proof-balanced-hash-detection}.

Thus all detection results based on the first two moments also have deterministic versions with the same asymptotic guarantees.
In particular, for every fixed $t\ge1$, choosing $v\to\infty$ with $v\log v=O(\log\log n)$ gives a deterministic algorithm that achieves strong detection for every clique location at threshold $c(t)^{-1/2}$ in time $\widetilde O(n^{t+1})$.

\begin{lemma}\label{lem:balanced-hash-rooted}\label{lem:derandomized-rooted}
Let $\mathcal H=\mathcal G^{\mathrm{conn}}_{v,\le t}$, and define the rooted statistic obtained from the balanced hash family by
\[
  \widetilde F_{\mathcal H,i}
  =
  \sum_{H\in\mathcal H}
  \sum_{r\in\mathcal R(H)}
  \sum_{\substack{\phi:V(H)\hookrightarrow[n]\\ \phi(r)=i}}
  c_{\phi,H}\chi_{H,\phi},
\]
where the coefficients arise from a $(1+\eta_v)$-balanced hash family and therefore satisfy
\[
(1+\eta_v)^{-1}\le c_{\phi,H}\le 1+\eta_v.
\]
Let
\[
\mu=\nu(\mathcal H)(k-1)_{v-1}.
\]
If $v=\lfloor\sqrt{\log n}\rfloor$, $k\ge\lambda\sqrt n$, and $c_0\lambda^2>1$ for some $c_0<c(t)$, put
\[
\delta_v=K_{\lambda,c_0,t}(c_0\lambda^2)^{-(v-1)}.
\]
Here $K_{\lambda,c_0,t}$ is a sufficiently large positive constant depending only on $\lambda,c_0,t$.
Choose $\eta_v=o(\delta_v^{1/2})$. Then, uniformly over every fixed $C\subseteq[n]$ of size $k$ and every $i\in[n]$, the following casewise bounds hold (the expectation is only over the background edges):
\begin{align*}
\E_{\PP_{n,C}}\left[
  \left(
    \frac{\widetilde F_{\mathcal H,i}}{\mu}-1
  \right)^2
  \right]
&\le (1+o(1))\delta_v,
&& i\in C,\\
\E_{\PP_{n,C}}\left[
  \left(
    \frac{\widetilde F_{\mathcal H,i}}{\mu}
  \right)^2
  \right]
&\le (1+o(1))\delta_v,
&& i\notin C.
\end{align*}
\end{lemma}
\noindent
The proof is deferred to Appendix~\ref{app:proof-balanced-hash-rooted}.

\subsection{Seed subsets and Tur\'an derandomization}
\label{sec:turan-seed-derandomization}

For the standard planted clique model, the clique location is random, so random choices of seed vertex subsets can be replaced by a fixed seed list inside a moderately large set of vertices without changing the runtime exponent.
This gives the same upper bounds for $\tau_{\det,\mathrm{avg}}$ as for $\tau_{\mathrm{rand}}$ in Theorem~\ref{thm:intro-main-bounds}.

\begin{proposition}
\label{prop:det-random-location-seeds}
\label{prop:det-seeds-random-location}
For every fixed integer $s\ge0$,
\[
\tau_{\det,\mathrm{avg}}\!\left((2^s e)^{-1/2}\right)
\le
\omega(1,1,s/2).
\]
\end{proposition}

\begin{proof}
Fix $\lambda>(2^s e)^{-1/2}$ and let $k\ge\lambda\sqrt n$. For $s=0$ there is no seed to choose. For $s\ge1$, fix
\[
U_0=\{1,\ldots,\lceil \sqrt n\log n\rceil\}
\]
and enumerate all ordered $s$-tuples of distinct vertices in $U_0$. This list has size $n^{s/2+o(1)}$. If the planted clique is uniformly random and $k\ge\lambda\sqrt n$, then $|C\cap U_0|$ is hypergeometric with mean at least $(1+o(1))\lambda\log n$, and hence $|C\cap U_0|\ge s$ with probability $1-o(1)$. Thus, with high probability, the fixed list contains a seed from the planted clique.

We also fix the holdout partition used by the cleanup step, for example by taking $V_2$ to be the last $\lfloor n^{3/5}\rfloor$ vertices. Since the clique location is uniformly random, the usual hypergeometric concentration for $C\cap V_1$ and $C\cap V_2$ still holds.

Use this fixed seed list, replace random colorings by the balanced hash family above, and run the batched algorithm for seeded trees from Theorem~\ref{thm:seeded-tree-rectangular-mm}. This gives the claimed bound on $\tau_{\det,\mathrm{avg}}$.
\end{proof}

\begin{proposition}
\label{prop:det-avg-square-mm}
For graphs of treewidth at most $2$, $\tau_{\det,\mathrm{avg}}(1/\sqrt{c(2)})\le \omega$.
\end{proposition}

\begin{proof}
Use the algorithm from Theorem~\ref{thm:tw2-mm} and replace its random colorings by the balanced hash families from Lemmas~\ref{lem:balanced-hash-detection} and~\ref{lem:balanced-hash-rooted}. For each fixed coloring, the dynamic program of Lemma~\ref{lem:one-partial-2tree-mm} is deterministic. Its total running time remains $n^{\omega+o(1)}$ after summing over the hash family. For cleanup, fix the holdout set in advance, for instance $V_2=\{n-\lfloor n^{3/5}\rfloor+1,\ldots,n\}$. Since the planted clique location is uniformly random, the same hypergeometric concentration used for the random holdout set holds for this fixed set. The same cleanup procedure therefore gives strong recovery, and hence strong detection, by a deterministic algorithm with runtime exponent $\omega$.
\end{proof}

The same argument allows us to add more seed subsets when the clique location is uniformly random. If a deterministic algorithm succeeds for a uniformly random clique location at threshold $\lambda_0$ in time $n^{r+o(1)}$ for some $r\ge2$, then, for every fixed integer $a\ge0$,
\[
\tau_{\det,\mathrm{avg}}(2^{-a/2}\lambda_0)
\le
r+\frac a2.
\]
This justifies the ceiling term in the bound on $\tau_{\det,\mathrm{avg}}$ in Theorem~\ref{thm:intro-main-bounds}.

To see this, fix $\lambda>2^{-a/2}\lambda_0$, let $k\ge\lambda\sqrt n$, fix a vertex set $U_0$ of size $\lceil\sqrt n\log n\rceil$, and enumerate all ordered $a$-tuples of distinct vertices in $U_0$. This list has size $n^{a/2+o(1)}$. Run the assumed algorithm on the induced subgraph on common neighbors of every tuple and perform the standard cleanup. For a fixed listed tuple $Q$, conditional on $Q\subseteq C$ and on the set of common neighbors, $C \setminus Q$ is uniformly distributed among the $K$-subsets of that set, where $K=k-a$. Writing $q=2^{-a}$, the set of common neighbors has size $(1+o(1))(qn+(1-q)K)$ with probability $1-o(1)$. Since $x/\sqrt{qn+(1-q)x}$ is increasing and $\lambda2^{a/2}>\lambda_0$, the resulting planted clique instance has signal greater than $\lambda_0$ by a fixed margin. Thus, conditional on $Q\subseteq C$, the call to the algorithm associated with $Q$ fails to return $C$ with probability $o(1)$.

Let $Z$ be the number of listed tuples contained in $C$, and let $Z_{\mathrm{bad}}$ be the number of these planted tuples whose calls fail. Hypergeometric concentration gives $Z=(1+o(1))\E Z$ with probability $1-o(1)$, where $\E Z=(|U_0|)_a(k)_a/(n)_a$ and $\E Z\ge(1+o(1))(\lambda\log n)^a$. The preceding conditional bound and symmetry give $\E Z_{\mathrm{bad}}=o(\E Z)$. Hence Markov's inequality implies $Z_{\mathrm{bad}}<Z$ with probability $1-o(1)$, so at least one call returns the planted clique. The base calls take $n^{r+a/2+o(1)}$ time and checking that the output is a clique costs $O(k^2)\le O(n^2)$ per tuple.

To obtain a deterministic guarantee for every clique location, the seed list must contain a seed from every possible clique, and the cleanup step must provide a suitable holdout set for every possible clique. We now construct families of sets with these properties.

\begin{definition}[Universal seed family]
For integers $n,k,s$, an \emph{$(n,k,s)$-universal seed family} is a family $\mathcal S\subseteq \binom{[n]}s$ such that every $C\subseteq[n]$ with $|C|=k$ contains at least one $S\in\mathcal S$.
\end{definition}

\begin{lemma}
\label{lem:explicit-universal-seed-family}
Fix an integer $s\ge1$, a constant $\lambda>0$, and $k\ge\lambda\sqrt n$. There is an explicit $(n,k,s)$-universal seed family of size $O_s(n^s/k^{s-1})\le n^{(s+1)/2+o(1)}$.
\end{lemma}

\begin{proof}
For $s=1$, take all singleton seeds. For $s\ge2$, put $b=\left\lfloor\frac{k-1}{s-1}\right\rfloor$, partition $[n]$ into $b$ consecutive parts whose sizes differ by at most one, and include every $s$-subset contained in one part. If a $k$-set $C$ contained fewer than $s$ vertices in every part, then $|C|\le b(s-1)\le k-1$, a contradiction. Thus some part contains an $s$-subset of $C$. The number of listed seeds is $O_s\!\left(b(n/b)^s\right) =O_s(n^s/k^{s-1}) \le n^{(s+1)/2+o(1)}$.
\end{proof}

We continue to use $\tau_{\det}$ for the tradeoff defined in Definition~\ref{def:tau-variants} for deterministic algorithms that work for every clique location. By definition, it requires both strong recovery and strong detection.

\begin{lemma}\label{lem:deterministic-holdout-splitter}
Fix a constant $\lambda>0$, let $k\ge\lambda\sqrt n$, and set $h=\left\lfloor \frac{k}{10\log n}\right\rfloor$.
There is an explicit family $\mathfrak P$ of subsets $B\subseteq[n]$, with
\[
|\mathfrak P|=O(n\log n/k)\le n^{1/2+o(1)},
\]
such that every $C\subseteq[n]$ with $|C|=k$ has some $B\in\mathfrak P$ satisfying $h\le |C\cap B|\le 2h$ and $|B|\le \frac{2n}{3}$.
Consequently, if $V_2=B$ and $V_1=[n]\setminus B$, then
\[
|C\cap V_2|\ge h\gg \log n,\qquad
|C\cap V_1|=k-o(k),\qquad
|V_1|\ge n/3.
\]
\end{lemma}

\begin{proof}
Partition $[n]$ into consecutive blocks of size between $h$ and $2h$. Build a binary tree by recursively splitting each consecutive list of these blocks into two groups whose total sizes differ by at most $2h$, and let $\mathfrak P$ be the collection of all non-root nodes, each identified with the union of the blocks of leaves below that node. The number of leaves, and hence the number of tree nodes, is $O(n/h)=O(n\log n/k)\le n^{1/2+o(1)}$.
At the root, each child has size at most $n/2+2h\le 2n/3$ for all sufficiently large $n$. Every non-root node is contained in one of these two children and therefore also has size at most $2n/3$.

Fix $C\subseteq[n]$ with $|C|=k$. The root contains $k>2h$ vertices of $C$, while each leaf contains at most $2h$ vertices of $C$. Choose a minimal node $u$, with respect to inclusion, such that $|C\cap u|>2h$.
This node is not a leaf. Since $u$ is minimal, each child $u_1,u_2$ of $u$ satisfies $|C\cap u_j|\le 2h$. But $|C\cap u_1|+|C\cap u_2|=|C\cap u|>2h$, so at least one child, say $u_j$, satisfies $|C\cap u_j|>h$. Therefore $h<|C\cap u_j|\le 2h$.
Taking $B=u_j$ proves the claim.
\end{proof}

\begin{lemma}\label{lem:universal-holdout-wrapper}
Fix $\lambda>0$ and $\varepsilon>0$. Suppose that, for every induced subgraph $G[U]$ with $|U|=N$ and every fixed planted clique $C_U\subseteq U$ of size
\[
|C_U|\ge(\lambda+\varepsilon/2)\sqrt N,
\]
there is a deterministic algorithm running in time $T(N)$, where $T$ is nondecreasing, that outputs raw scores $S_i$, $i\in U$, for which there exists a positive factor $\mu_{U,C_U}$, common to all $i$ and not required to be known to the algorithm, such that the normalized scores $\widehat p_i=S_i/\mu_{U,C_U}$ satisfy
\[
\E\left[
\left(\widehat p_i-\mathbf 1_{\{i\in C_U\}}\right)^2
\right]
\le
\delta_N
\]
uniformly in $i\in U$, where $\delta_N^{1/2}=o((\log N)^{-3})$.
Then there is a deterministic exact-recovery algorithm that succeeds with high probability for every clique location and, for every planted clique size $k\ge(\lambda+\varepsilon)\sqrt n$, runs in time $n^{1/2+o(1)}T(n)+O(n^{5/2+o(1)})$.
\end{lemma}
\noindent
The proof is deferred to Appendix~\ref{app:proof-universal-holdout-wrapper}.

\begin{lemma}\label{lem:seeded-rooted-derandomization}
Fix a constant $\eta>0$ and a deterministic pair $(B,Q)$ consisting of a holdout set $B$ and a seed $Q$. Put $V_1=[n]\setminus B$, let $v=\lfloor\sqrt{\log n}\rfloor$, and let
\[
W_{B,Q}
=
\{x\in V_1\setminus Q:\ x\text{ is adjacent to every vertex of }Q\}.
\]
Let $\widetilde F_{B,Q,\mathrm{tree},i}$ be the seeded tree statistic rooted at $i$ for the pair $(B,Q)$. Compute it with a $(1+\eta_v)$-balanced hash family in place of random colorings, where the family is chosen so that $\eta_v=o(\delta_v^{1/2})$ for the $\delta_v$ appearing below. If $Q\subseteq C\cap V_1$, $|W_{B,Q}|=n^{\Omega(1)}$, and
\[
|C_W|\ge(e^{-1/2}+\eta)\sqrt{|W_{B,Q}|},
\qquad
C_W=(C\cap V_1)\setminus Q,
\]
then, after dividing by the same positive normalizing factor as in Lemma~\ref{lem:good-seeded-row}, the estimates $\widetilde p_{B,Q,i}$ satisfy
\[
\E\left[
\left(
\widetilde p_{B,Q,i}-\mathbf 1_{\{i\in C_W\}}
\right)^2
\mid W_{B,Q},C_W
\right]
\le \delta_v,
\qquad
\delta_v=\exp(-\Omega(v)),
\]
uniformly over $i\in W_{B,Q}$.
\end{lemma}

\begin{proof}
With $v=\lfloor\sqrt{\log n}\rfloor$ as in the statement, condition on the deterministic pair $(B,Q)$, on the exposed set of common neighbors $W_{B,Q}$, and on $C_W$. The hypothesis $|W_{B,Q}|=n^{\Omega(1)}$ ensures that $v=o(\log |W_{B,Q}|/\log\log |W_{B,Q}|)$. The vertex weights $w_{B,Q}$ only restrict the host set to $W_{B,Q}$, so this statistic is the ordinary rooted tree statistic for the planted clique instance $G[W_{B,Q}]$, as in Lemma~\ref{lem:good-seeded-row}. The proof of Lemma~\ref{lem:balanced-hash-rooted}, which compares the coefficients before and after replacing random colorings, therefore applies with $\mathcal H=\mathcal T_v$ and with vertex set $W_{B,Q}$. We have $O(\eta_v^2)=o(\delta_v)$ and the variance changes by a factor $1+o(1)$, giving the displayed bound.
\end{proof}

The following theorem collects the bounds for deterministic algorithms and arbitrary clique locations from Theorems~\ref{thm:intro-main-general}, \ref{thm:tau-fast-tw2}, and~\ref{thm:intro-main-bounds}.
\begin{theorem}\label{thm:universal-deterministic-exact-recovery}\label{thm:universal-det}\label{thm:universal-deterministic}
For deterministic algorithms that must work for every clique location, the following bounds hold. For every fixed integer $t\ge1$,
\[
\tau_{\det}(c(t)^{-1/2})\le t+\frac32.
\]
For treewidth at most $2$, the algorithm using square matrix multiplication gives
\[
\tau_{\det}(1/\sqrt{c(2)})\le \omega+\frac12.
\]
For every fixed $\lambda>c(2)^{-1/2}$, a deterministic algorithm can achieve strong detection for every clique location in time $n^{\omega+o(1)}$.
At $c(1)^{-1/2}=e^{-1/2}$, the deterministic message-passing algorithm gives the sharper bound $\tau_{\det}(e^{-1/2}) \le 2$.
Lastly, for every fixed integer $s\ge0$, define
\[
\kappa_s=
\begin{cases}
0, & s=0,\\
(s+1)/2, & s\ge1.
\end{cases}
\]
Then
\[
\tau_{\det}\!\left((2^s e)^{-1/2}\right)
\le
\min\{\omega(1,1,\kappa_s+1/2),\,5/2+\kappa_s\}.
\]
\end{theorem}
\noindent
The proof is deferred to Appendix~\ref{app:proof-universal-deterministic-upper-bounds}.

\section*{Acknowledgments}
\addcontentsline{toc}{section}{Acknowledgments}

Thanks to Alex Wein and Ilias Zadik for helpful discussions with the first author that eventually led to this work. All mathematical proofs in this paper are due to the authors. ChatGPT 5.6 was used for language polishing and final proofreading.

\addcontentsline{toc}{section}{References}
\bibliographystyle{alpha}
\bibliography{references}

@article{alon1998finding,
  title={Finding a large hidden clique in a random graph},
  author={Alon, Noga and Krivelevich, Michael and Sudakov, Benny},
  journal={Random Structures \& Algorithms},
  volume={13},
  number={3-4},
  pages={457--466},
  year={1998},
  publisher={Wiley Online Library}
}

@article{kuvcera1995expected,
  title={Expected complexity of graph partitioning problems},
  author={Ku{\v{c}}era, Lud{\v{e}}k},
  journal={Discrete Applied Mathematics},
  volume={57},
  number={2-3},
  pages={193--212},
  year={1995},
  publisher={Elsevier}
}

@inproceedings{feige2010finding,
  author    = {Uriel Feige and Dorit Ron},
  title     = {Finding hidden cliques in linear time},
  booktitle = {Proceedings of the 21st International Meeting on Probabilistic, Combinatorial, and Asymptotic Methods in the Analysis of Algorithms (AofA'10)},
  series    = {Discrete Mathematics \& Theoretical Computer Science Proceedings},
  volume    = {AM},
  pages     = {189--204},
  year      = {2010},
  doi       = {10.46298/dmtcs.2802}
}

@article{dekel2014finding,
  title={Finding hidden cliques in linear time with high probability},
  author={Dekel, Yael and Gurel-Gurevich, Ori and Peres, Yuval},
  journal={Combinatorics, Probability and Computing},
  volume={23},
  number={1},
  pages={29--49},
  year={2014},
  publisher={Cambridge University Press}
}

@article{feige2000finding,
  title={Finding and certifying a large hidden clique in a semirandom graph},
  author={Feige, Uriel and Krauthgamer, Robert},
  journal={Random Structures \& Algorithms},
  volume={16},
  number={2},
  pages={195--208},
  year={2000},
  publisher={Wiley Online Library}
}

@article{deshpande2015finding,
  title={Finding hidden cliques of size N/e in nearly linear time},
  author={Deshpande, Yash and Montanari, Andrea},
  journal={Foundations of Computational Mathematics},
  volume={15},
  number={4},
  pages={1069--1128},
  year={2015},
  publisher={Springer}
}

@article{ames2011nuclear,
  title={Nuclear norm minimization for the planted clique and biclique problems},
  author={Ames, Brendan PW and Vavasis, Stephen A},
  journal={Mathematical programming},
  volume={129},
  number={1},
  pages={69--89},
  year={2011},
  publisher={Springer}
}

@article{ames2015guaranteed,
  title={Guaranteed recovery of planted cliques and dense subgraphs by convex relaxation},
  author={Ames, Brendan PW},
  journal={Journal of Optimization Theory and Applications},
  volume={167},
  number={2},
  pages={653--675},
  year={2015},
  publisher={Springer}
}

@article{li2025smooth,
  title={A Smooth Computational Transition in Tensor PCA},
  author={Li, Zhangsong},
  journal={arXiv preprint arXiv:2509.09904},
  year={2025}
}

@article{alon1995color,
  title={Color-coding},
  author={Alon, Noga and Yuster, Raphael and Zwick, Uri},
  journal={Journal of the ACM (JACM)},
  volume={42},
  number={4},
  pages={844--856},
  year={1995},
  publisher={ACM New York, NY, USA}
}

@article{bringmann2021current,
  title={Current algorithms for detecting subgraphs of bounded treewidth are probably optimal},
  author={Bringmann, Karl and Slusallek, Jasper},
  journal={arXiv preprint arXiv:2105.05062},
  year={2021}
}

@article{alman2024asymmetry,
  title={More asymmetry yields faster matrix multiplication},
  author={Alman, Josh and Duan, Ran and Vassilevska Williams, Virginia and Xu, Yinzhan and Xu, Zixuan and Zhou, Renfei},
  journal={arXiv preprint arXiv:2404.16349},
  year={2024}
}

@inproceedings{vassilevska2024new,
  title={New Bounds for Matrix Multiplication: from Alpha to Omega},
  author={Vassilevska Williams, Virginia and Xu, Yinzhan and Xu, Zixuan and Zhou, Renfei},
  booktitle={Proceedings of the 2024 Annual ACM-SIAM Symposium on Discrete Algorithms},
  pages={3792--3835},
  year={2024},
  organization={SIAM}
}

@article{fomin2012faster,
  title={Faster algorithms for finding and counting subgraphs},
  author={Fomin, Fedor V and Lokshtanov, Daniel and Raman, Venkatesh and Saurabh, Saket and Rao, BV Raghavendra},
  journal={Journal of Computer and System Sciences},
  volume={78},
  number={3},
  pages={698--706},
  year={2012},
  publisher={Elsevier}
}

@article{alon2010balanced,
  title={Balanced families of perfect hash functions and their applications},
  author={Alon, Noga and Gutner, Shai},
  journal={ACM Transactions on Algorithms (TALG)},
  volume={6},
  number={3},
  pages={1--12},
  year={2010},
  publisher={ACM New York, NY, USA}
}

@article{baste2017number,
  title={On the number of labeled graphs of bounded treewidth},
  author={Baste, Julien and Noy, Marc and Sau, Ignasi},
  journal={European Journal of Combinatorics},
  volume={71},
  pages={12--21},
  year={2018},
  publisher={Elsevier}
}

@article{norine2006proper,
  title={Proper minor-closed families are small},
  author={Norine, Serguei and Seymour, Paul and Thomas, Robin and Wollan, Paul},
  journal={Journal of Combinatorial Theory, Series B},
  volume={96},
  number={5},
  pages={754--757},
  year={2006},
  publisher={Elsevier}
}

@inproceedings{sohn2025sharp,
  title={Sharp phase transitions in estimation with low-degree polynomials},
  author={Sohn, Youngtak and Wein, Alexander S},
  booktitle={Proceedings of the 57th Annual ACM Symposium on Theory of Computing},
  pages={891--902},
  year={2025}
}

@article{lugosi2017lectures,
  title={Lectures on combinatorial statistics},
  author={Lugosi, G{\'a}bor},
  journal={47th Probability Summer School, Saint-Flour},
  pages={1--91},
  year={2017}
}

@article{chvatal1979tail,
  title={The tail of the hypergeometric distribution},
  author={Chv{\'a}tal, Vasek},
  journal={Discrete Mathematics},
  volume={25},
  number={3},
  pages={285--287},
  year={1979},
  publisher={Elsevier}
}

@phdthesis{Shamos-1978-ComputationalGeometry,
  title={Computational geometry},
  author={Shamos, Michael Ian},
  year={1978},
  school={Yale University}
}

@book{FK-2016-RandomGraphs,
  title={Introduction to random graphs},
  author={Frieze, Alan and Karo{\'n}ski, Micha{\l}},
  year={2016},
  publisher={Cambridge University Press}
}

@article{mao2024testing,
  title={Testing network correlation efficiently via counting trees},
  author={Mao, Cheng and Wu, Yihong and Xu, Jiaming and Yu, Sophie H},
  journal={The Annals of Statistics},
  volume={52},
  number={6},
  pages={2483--2505},
  year={2024},
  publisher={Institute of Mathematical Statistics}
}

@inproceedings{mao2023random,
  title={Random graph matching at Otter’s threshold via counting chandeliers},
  author={Mao, Cheng and Wu, Yihong and Xu, Jiaming and Yu, Sophie H},
  booktitle={Proceedings of the 55th Annual ACM symposium on theory of computing},
  pages={1345--1356},
  year={2023}
}

@article{DHS-2020-SpikedMatrixHeavyTailed,
  title={Estimating Rank-One Spikes from Heavy-Tailed Noise via Self-Avoiding Walks},
  author={Ding, Jingqiu and Hopkins, Samuel B and Steurer, David},
  journal={arXiv preprint arXiv:2008.13735},
  year={2020}
}

@article{feldman2017statistical,
  title={Statistical algorithms and a lower bound for detecting planted cliques},
  author={Feldman, Vitaly and Grigorescu, Elena and Reyzin, Lev and Vempala, Santosh S and Xiao, Ying},
  journal={Journal of the ACM (JACM)},
  volume={64},
  number={2},
  pages={1--37},
  year={2017},
  publisher={ACM New York, NY, USA}
}

@article{barak2019nearly,
  title={A nearly tight sum-of-squares lower bound for the planted clique problem},
  author={Barak, Boaz and Hopkins, Samuel and Kelner, Jonathan and Kothari, Pravesh K and Moitra, Ankur and Potechin, Aaron},
  journal={SIAM Journal on Computing},
  volume={48},
  number={2},
  pages={687--735},
  year={2019},
  publisher={SIAM}
}

@article{berthet2013computational,
  title={Computational lower bounds for sparse PCA},
  author={Berthet, Quentin and Rigollet, Philippe},
  journal={arXiv preprint arXiv:1304.0828},
  year={2013}
}

@article{ma2015computational,
  title={Computational barriers in minimax submatrix detection},
  author={Ma, Zongming and Wu, Yihong},
  journal={The Annals of Statistics},
  pages={1089--1116},
  year={2015},
  publisher={JSTOR}
}

@inproceedings{brennan2020reducibility,
  title={Reducibility and statistical-computational gaps from secret leakage},
  author={Brennan, Matthew and Bresler, Guy},
  booktitle={Conference on Learning Theory},
  pages={648--847},
  year={2020},
  organization={PMLR}
}

@inproceedings{bresler2023detection,
  title={Detection-recovery and detection-refutation gaps via reductions from planted clique},
  author={Bresler, Guy and Jiang, Tianze},
  booktitle={The Thirty Sixth Annual Conference on Learning Theory},
  pages={5850--5889},
  year={2023},
  organization={PMLR}
}

@inproceedings{brennan2019universality,
  title={Universality of computational lower bounds for submatrix detection},
  author={Brennan, Matthew and Bresler, Guy and Huleihel, Wasim},
  booktitle={Conference on Learning Theory},
  pages={417--468},
  year={2019},
  organization={PMLR}
}

@article{chen2016statistical,
  title={Statistical-computational tradeoffs in planted problems and submatrix localization with a growing number of clusters and submatrices},
  author={Chen, Yudong and Xu, Jiaming},
  journal={Journal of Machine Learning Research},
  volume={17},
  number={27},
  pages={1--57},
  year={2016}
}

@article{abbe2018community,
  title={Community detection and stochastic block models: recent developments},
  author={Abbe, Emmanuel},
  journal={Journal of Machine Learning Research},
  volume={18},
  number={177},
  pages={1--86},
  year={2018}
}

@article{hopkins2017bayesian,
  title={Bayesian estimation from few samples: community detection and related problems},
  author={Hopkins, Samuel B and Steurer, David},
  journal={arXiv preprint arXiv:1710.00264},
  year={2017}
}

@book{hopkins2018statistical,
  title={Statistical inference and the sum of squares method},
  author={Hopkins, Samuel},
  year={2018},
  publisher={Cornell University}
}

@inproceedings{kunisky2019notes,
  title={Notes on computational hardness of hypothesis testing: Predictions using the low-degree likelihood ratio},
  author={Kunisky, Dmitriy and Wein, Alexander S and Bandeira, Afonso S},
  booktitle={ISAAC Congress (International Society for Analysis, its Applications and Computation)},
  pages={1--50},
  year={2019},
  organization={Springer}
}

@article{wein2025computational,
  title={Computational Complexity of Statistics: New Insights from Low-Degree Polynomials},
  author={Wein, Alexander S},
  journal={arXiv preprint arXiv:2506.10748},
  year={2025}
}

@article{FVRS-2021-TutorialAMP,
  title={A unifying tutorial on Approximate Message Passing},
  author={Feng, Oliver Y and Venkataramanan, Ramji and Rush, Cynthia and Samworth, Richard J},
  journal={arXiv preprint arXiv:2105.02180},
  year={2021}
}

@article{montanari2022equivalence,
  title={Equivalence of approximate message passing and low-degree polynomials in rank-one matrix estimation},
  author={Montanari, Andrea and Wein, Alexander S},
  journal={arXiv preprint arXiv:2212.06996},
  year={2022}
}

@article{Borchardt-1860-CountingTrees,
  author   = {Borchardt, C. W.},
  title    = {Ueber eine der Interpolation entsprechende Darstellung der Eliminations-Resultante},
  journal  = {Journal f{\"u}r die reine und angewandte Mathematik},
  volume   = {57},
  year     = {1860},
  pages    = {111--121},
  doi      = {10.1515/crll.1860.57.111}
}

@article{Cayley-1889-CountingTrees,
  author  = {Cayley, Arthur},
  title   = {A Theorem on Trees},
  journal = {Quarterly Journal of Pure and Applied Mathematics},
  volume  = {23},
  year    = {1889},
  pages   = {376--378}
}

@article{bodirsky2007enumeration,
  title={Enumeration and limit laws for series--parallel graphs},
  author={Bodirsky, Manuel and Gim{\'e}nez, Omer and Kang, Mihyun and Noy, Marc},
  journal={European Journal of Combinatorics},
  volume={28},
  number={8},
  pages={2091--2105},
  year={2007},
  publisher={Elsevier}
}

@book{cygan2015parameterized,
  title={Parameterized algorithms},
  author={Cygan, Marek and Fomin, Fedor V and Kowalik, {\L}ukasz and Lokshtanov, Daniel and Marx, D{\'a}niel and Pilipczuk, Marcin and Pilipczuk, Micha{\l} and Saurabh, Saket},
  volume={5},
  year={2015},
  publisher={Springer}
}

@article{yu2025counting,
  title={Counting stars is constant-degree optimal for detecting any planted subgraph},
  author={Yu, Xifan and Zadik, Ilias and Zhang, Peiyuan},
  journal={Mathematical Statistics and Learning},
  volume={8},
  number={1},
  pages={105--164},
  year={2025}
}

@article{ma2025nonlinear,
  title={Nonlinear Laplacians: Tunable principal component analysis under directional prior information},
  author={Ma, Yuxin and Kunisky, Dmitriy},
  journal={arXiv preprint arXiv:2505.12528},
  year={2025}
}

@inproceedings{bodlaender1997treewidth,
  title={Treewidth: Algorithmic techniques and results},
  author={Bodlaender, Hans L},
  booktitle={International symposium on mathematical foundations of computer science},
  pages={19--36},
  year={1997},
  organization={Springer}
}

@book{kloks1994treewidth,
  title={Treewidth: computations and approximations},
  author={Kloks, Ton},
  year={1994},
  publisher={Springer}
}

@article{schramm2022computational,
  title={Computational barriers to estimation from low-degree polynomials},
  author={Schramm, Tselil and Wein, Alexander S},
  journal={The Annals of Statistics},
  volume={50},
  number={3},
  pages={1833--1858},
  year={2022},
  publisher={Institute of Mathematical Statistics}
}

@article{chen2026computational,
  title={A computational transition for detecting correlated stochastic block models by low-degree polynomials},
  author={Chen, Guanyi and Ding, Jian and Gong, Shuyang and Li, Zhangsong},
  journal={The Annals of Statistics},
  volume={54},
  number={1},
  pages={226--251},
  year={2026},
  publisher={Institute of Mathematical Statistics}
}

@article{li2025algorithmic,
  title={The Algorithmic Phase Transition in Correlated Spiked Models},
  author={Li, Zhangsong},
  journal={arXiv preprint arXiv:2511.06040},
  year={2025}
}

@article{robertson1986graph,
  title={Graph minors. II. Algorithmic aspects of tree-width},
  author={Robertson, Neil and Seymour, Paul D.},
  journal={Journal of algorithms},
  volume={7},
  number={3},
  pages={309--322},
  year={1986},
  publisher={Elsevier}
}

@article{yedidia2001understanding,
  title={Understanding belief propagation and its generalizations},
  author={Yedidia, Jonathan S and Freeman, William T and Weiss, Yair and others},
  journal={Exploring artificial intelligence in the new millennium},
  volume={8},
  number={236-239},
  pages={0018--9448},
  year={2003}
}

@inproceedings{jones2024fourier,
  title={Fourier Analysis of Iterative Algorithms},
  author={Jones, Chris and Pesenti, Lucas},
  booktitle={52nd International Colloquium on Automata, Languages, and Programming (ICALP 2025)},
  pages={102--1},
  year={2025},
  organization={Schloss Dagstuhl--Leibniz-Zentrum f{\"u}r Informatik}
}

@article{zdeborova2015statphys,
  title={Statistical physics of inference: Thresholds and algorithms},
  author={Zdeborov{\'a}, Lenka and Krzakala, Florent},
  journal={Advances in Physics},
  volume={65},
  number={5},
  pages={453--552},
  year={2016},
  publisher={Taylor \& Francis}
}

@article{castellvi2024chordal,
  title={Chordal graphs with bounded tree-width},
  author={Castellv{\'\i}, Jordi and Drmota, Michael and Noy, Marc and Requil{\'e}, Cl{\'e}ment},
  journal={Advances in Applied Mathematics},
  volume={157},
  pages={102700},
  year={2024},
  publisher={Elsevier}
}

@article{otter1948number,
  title={The number of trees},
  author={Otter, Richard},
  journal={Annals of Mathematics},
  volume={49},
  number={3},
  pages={583--599},
  year={1948},
  publisher={JSTOR}
}

@article{mao2023exact,
  title={Exact matching of random graphs with constant correlation},
  author={Mao, Cheng and Rudelson, Mark and Tikhomirov, Konstantin},
  journal={Probability Theory and Related Fields},
  volume={186},
  number={1},
  pages={327--389},
  year={2023},
  publisher={Springer}
}

@article{chen2025detecting,
  title={Detecting correlation efficiently in stochastic block models: breaking Otter's threshold by counting decorated trees},
  author={Chen, Guanyi and Ding, Jian and Gong, Shuyang and Li, Zhangsong},
  journal={arXiv e-prints},
  pages={arXiv--2503},
  year={2025}
}

@article{Jerrum-1992-LargeCliques,
  title={Large cliques elude the {Metropolis} process},
  author={Jerrum, Mark},
  journal={Random Structures \& Algorithms},
  volume={3},
  number={4},
  pages={347--359},
  year={1992},
  publisher={Wiley Online Library}
}

@inproceedings{CMZ-2023-LinearPlantedCliquesMetropolis,
  title={Almost-linear planted cliques elude the {Metropolis} process},
  author={Chen, Zongchen and Mossel, Elchanan and Zadik, Ilias},
  booktitle={Proceedings of the 2023 Annual ACM-SIAM Symposium on Discrete Algorithms (SODA)},
  pages={4504--4539},
  year={2023},
  organization={SIAM}
}

@article{GJX-2023-PlantedCliqueMCMC,
  title={Finding planted cliques using {Markov} chain {Monte} {Carlo}},
  author={Gheissari, Reza and Jagannath, Aukosh and Xu, Yiming},
  journal={arXiv preprint arXiv:2311.07540},
  year={2023}
}

@article{MPW-2015-PlantedClique,
  title={Sum-of-squares lower bounds for planted clique},
  author={Meka, Raghu and Potechin, Aaron and Wigderson, Avi},
  journal={arXiv preprint arXiv:1503.06447},
  year={2015}
}

@article{DM-2015-SOSPlantedClique,
  title={Improved sum-of-squares lower bounds for hidden clique and hidden submatrix problems},
  author={Deshpande, Yash and Montanari, Andrea},
  journal={arXiv preprint arXiv:1502.06590},
  year={2015}
}

@inproceedings{HKPRS-2018-PlantedCliqueSOS4,
  title={The power of sum-of-squares for detecting hidden structures},
  author={Hopkins, Samuel B. and Kothari, Pravesh K. and Potechin, Aaron and Raghavendra, Prasad and Schramm, Tselil and Steurer, David},
  booktitle={2017 IEEE 58th Annual Symposium on Foundations of Computer Science (FOCS)},
  pages={720--731},
  year={2017},
  organization={IEEE}
}

@article{dupont2026improving,
  title={Improving the matrix multiplication exponent with modern optimization and {AlphaEvolve}},
  author={Dupont, Emilien and Eisenberger, Marvin and Kozlovskii, Borislav and Mehrabian, Abbas and Ruiz, Francisco J. R. and See, Abigail and Zhou, Renfei and Alman, Josh and Vassilevska Williams, Virginia and Balog, Matej},
  journal={arXiv preprint arXiv:2608.16884},
  year={2026},
  doi={10.48550/arXiv.2608.16884}
}

\clearpage

\appendix

\section{Combinatorial estimates}

\begin{lemma}[Lemma 24.1(e) of \cite{FK-2016-RandomGraphs}]\label{lem:fall}
    Let $1 \leq k \leq n$.
    Then
    \[
    \left(1 - \frac{k(k - 1)}{2n}\right)n^k \leq (n)_k \leq n^k.
    \]
\end{lemma}

\begin{lemma}\label{lem:alpha-diff}
Let $k = k(n)$ and $v = v(n)$ be positive integers such that $2v \leq k < n$ for all sufficiently large $n$, $k = o(n)$, and $v = o(\sqrt{k})$.
For $0\le t\le2v$, write $\alpha_t \colonequals (k)_t / (n)_t$.
Then, as $n \to \infty$,
\begin{align*}
    \frac{\alpha_{2v} - \alpha_v^2}{\alpha_v^2}
    &=
    -\frac{v^2}{k}\,(1+o(1)), \\
    \frac{\alpha_t}{\alpha_v^2}
    &=
    \left(\frac{k}{n}\right)^{t-2v}(1+o(1)) \text{ uniformly over } 0 \leq t \leq 2v - 1,
\end{align*}
where the second $o(1)$ term is bounded in absolute value by a single $\varepsilon(n)=o(1)$ for all such $t$.
\end{lemma}
\begin{proof}
    Lemma~\ref{lem:fall} gives
    \[
    \alpha_t = \left(\frac{k}{n}\right)^t\left(1 + O\left(\frac{v^2}{k}\right)\right)
    \]
    uniformly over $0 \leq t \leq 2v$.
    Since $v = o(\sqrt{k})$, this proves the second statement.
    For the first,
    \[
    \frac{\alpha_{2v}}{\alpha_v^2} = \frac{(k)_{2v}}{(k)_v^2}\frac{(n)_v^2}{(n)_{2v}} = \prod_{i=0}^{v-1}\frac{k-v-i}{k-i}\frac{n-i}{n-v-i} = \prod_{i=0}^{v-1}\left(1 - x_i\right),
    \]
    where
\[
x_i \colonequals \frac{v(n-k)}{(k-i)(n-v-i)} = \frac{v}{k}(1 + o(1))
\]
uniformly over $0 \leq i \leq v - 1$.
With $M \colonequals \max_{i = 0}^{v - 1} x_i = \frac{v}{k}(1 + o(1))$, we obtain
\begin{align*}
    \frac{\alpha_{2v}}{\alpha_v^2} - 1
    &= -\sum_{i = 0}^{v - 1} x_i + O\left((1 + M)^v - 1 - vM\right) \\
    &= -\frac{v^2}{k}(1 + o(1)) + O\left(\exp(vM) - 1 - vM\right) \\
    &= -\frac{v^2}{k}(1 + o(1)) + O((vM)^2) \\
    &= -\frac{v^2}{k}(1 + o(1)) + O\left(\frac{v^4}{k^2}\right) \\
    &= -\frac{v^2}{k}(1 + o(1)).
\end{align*}
Here Taylor's theorem applies because $vM=o(1)$, and the last step uses $v=o(\sqrt{k})$.
\end{proof}
\section{Omitted proofs from Section \ref{sec:stats}}
\subsection{Proof of Lemma \ref{lem:sufficient-conditions-overlap}}\label{apx:sufficient-condition}
\begin{proof}
Since $\E_{\QQ}F_{\sH}=0$ by Lemma~\ref{lem:detection-null}, we must bound
\[
\frac{\Var_{\PP}[F_{\sH}(X)] + \Var_{\QQ}[F_{\sH}(X)]}{ (\E_{\PP} F_{\sH}(X) - \E_{\QQ} F_{\sH}(X))^2} = \frac{\Var_{\PP}[F_{\sH}(X)]}{(\E_{\PP} F_{\sH}(X))^2} + \frac{\Var_{\QQ}[F_{\sH}(X)]}{(\E_{\PP} F_{\sH}(X))^2}.
\]
For the null variance, Lemmas~\ref{lem:detection-null}, \ref{lem:detection-planted}, and~\ref{lem:fall} give
\[
\frac{\Var_{\QQ}[F_{\sH}(X)]}{(\E_{\PP} F_{\sH}(X))^2} = \frac{|\sH| \cdot (n)_v \cdot v!}{|\sH|^2 \cdot ((k)_v)^2} = (1 + o(1))\frac{(n / k^2)^v \cdot v!}{|\sH|},
\]
which is absorbed by the first term of the stated bound.

We bound the planted variance using $p_{u,\ell}(\sH)$.
For each integer $t$, define
\[
N_t
\colonequals
\Bigl|
\Bigl\{
(H,H',\phi,\psi):
H,H'\in\mathcal H,\
\phi,\psi\in\Inj([v],[n]),\
\sigma(H,H';\phi,\psi)=t
\Bigr\}
\Bigr|,
\]
and define
\begin{align*}
M_{u,\ell}
&\colonequals \Bigl|\Bigl\{(H,H',\phi,\psi):\,H,H'\in\mathcal H,\ \phi,\psi\in\Inj([v],[n]),\ |\phi([v])\cup\psi([v])|=2v-u, \\
&\hspace{3.5cm} \sigma(H,H';\phi,\psi)=2v-u-\ell\Bigr\}\Bigr|.
\end{align*}
Then
\[
\sum_{t=0}^{2v-1}N_t\alpha_t
=
\sum_{\substack{0\le u\le v\\0\le \ell\le u\\(u,\ell)\ne(0,0)}}
M_{u,\ell}\alpha_{2v-u-\ell}.
\]

For all sufficiently large $n$, $2v\le k$.
The contribution from disjoint embeddings is nonpositive, since
\[
\frac{\alpha_{2v}}{\alpha_v^2}
=
\prod_{j=0}^{v-1}
\frac{k-v-j}{k-j}\frac{n-j}{n-v-j}
\le1.
\]
Grouping the terms in Lemma~\ref{lem:detection-planted} by the support size $t$ gives
\begin{align*}
\frac{\Var_{\PP}[F_{\sH}(X)]}{(\E_{\PP} F_{\sH}(X))^2}
&= \frac{1}{|\sH|^2 ((k)_v)^2} \sum_{H,H'\in\mathcal H}\ \sum_{\phi, \psi \in \Inj([v], [n])}
\Bigl(\alpha_{\sigma(H,H';\phi,\psi)}-\alpha_v^2\Bigr)
\\
&\leq \frac{1}{|\sH|^2 ((k)_v)^2}
\sum_{t=0}^{2v-1} N_t\bigl(\alpha_t-\alpha_v^2\bigr)
\\
&\leq
\frac{1}{|\sH|^2 ((k)_v)^2} \sum_{t=0}^{2v-1} N_t\,\alpha_t \\
&= \frac{1}{|\sH|^2 ((k)_v)^2}
\sum_{\substack{0\le u\le v\\0\le \ell\le u\\(u,\ell)\ne(0,0)}}
M_{u,\ell}\alpha_{2v-u-\ell}.
\end{align*}
The cases $u<v$ and $u=v$ give the other two terms of the bound.

\paragraph{Case 1: $u<v$.}
For $0\le\ell\le u\le v-1$,
\[
M_{u,\ell}=\binom{n}{v}\binom{v}{u}\binom{n-v}{v-u}(v!)^2|\mathcal H|^2\,p_{u,\ell}(\mathcal H).
\]
Since $v^2=o(k)$,
\[
\begin{aligned}
\frac{M_{u,\ell}\,\alpha_{2v-u-\ell}}{|\mathcal H|^2\,(k)_v^2}
&=\binom{n}{v}\binom{v}{u}\binom{n-v}{v-u}(v!)^2\cdot\frac{(k)_{2v-u-\ell}}{(n)_{2v-u-\ell}(k)_v^2}\cdot p_{u,\ell}(\mathcal H)\\
&=\frac{n!(v!)^2}{u!(v-u)!^2(n-2v+u)!}\cdot\frac{(k)_{2v-u-\ell}}{(n)_{2v-u-\ell}(k)_v^2}\cdot p_{u,\ell}(\mathcal H)\\
&=\frac{(v)_u^2}{u!}\cdot\frac{(n)_{2v-u}}{(n)_{2v-u-\ell}}\cdot\frac{(k)_{2v-u-\ell}}{(k)_v^2}\cdot p_{u,\ell}(\mathcal H)\\
&=\frac{(v)_u^2}{u!}\cdot(n-2v+u+\ell)_{\ell}\cdot\frac{(k)_{2v-u-\ell}}{(k)_v^2}\cdot p_{u,\ell}(\mathcal H)\\
&\le (1+o(1))\cdot \frac{1}{u!}\,v^{2u}n^{\ell}k^{-(u+\ell)}\,p_{u,\ell}(\mathcal H),
\end{aligned}
\]
where the last step uses Lemma~\ref{lem:fall}.
Consequently,
\[
\sum_{\substack{0\le u\le v-1\\0\le \ell\le u\\(u,\ell)\ne(0,0)}}
\frac{M_{u,\ell}\,\alpha_{2v-u-\ell}}{|\mathcal H|^2\,(k)_v^2}
\le
(1+o(1))
\sum_{\substack{0\le u\le v-1\\0\le \ell\le u\\(u,\ell)\ne(0,0)}}
\frac{1}{u!}\,v^{2u}n^{\ell}k^{-(u+\ell)}\,p_{u,\ell}(\mathcal H).
\]
Since $\sum_{u=0}^{\infty}(u+1)/u!<\infty$, we obtain
\[
\sum_{\substack{0\le u\le v-1\\0\le \ell\le u\\(u,\ell)\ne(0,0)}}
\frac{M_{u,\ell}\,\alpha_{2v-u-\ell}}{|\mathcal H|^2\,(k)_v^2}
\le
O\left(
\max_{\substack{0\le \ell\le u<v\\(u,\ell)\ne(0,0)}}
\Bigl\{v^{2u}n^{\ell}k^{-(u+\ell)}p_{u,\ell}(\mathcal H)\Bigr\}
\right).
\]
\paragraph{Case 2: $u=v$.}
For $0\le\ell\le v$, we similarly have
\[
M_{v,\ell}=\binom{n}{v}(v!)^2|\mathcal H|^2\,p_{v,\ell}(\mathcal H).
\]
Again using $v^2=o(k)$,
\[
\begin{aligned}
\frac{M_{v,\ell}\,\alpha_{v-\ell}}{|\mathcal H|^2\,(k)_v^2}
&=\binom{n}{v}(v!)^2\cdot\frac{(k)_{v-\ell}}{(n)_{v-\ell}(k)_v^2}\cdot p_{v,\ell}(\mathcal H)\\
&=(n)_v\,v!\cdot\frac{(k)_{v-\ell}}{(n)_{v-\ell}(k)_v^2}\cdot p_{v,\ell}(\mathcal H)\\
&=v!\cdot\frac{(n)_v}{(n)_{v-\ell}}\cdot\frac{(k)_{v-\ell}}{(k)_v^2}\cdot p_{v,\ell}(\mathcal H)\\
&\le (1+o(1))\,v!\,n^{\ell}k^{-(v+\ell)}\,p_{v,\ell}(\mathcal H).
\end{aligned}
\]
Therefore,
\[
\sum_{\ell=0}^{v}\frac{M_{v,\ell}\,\alpha_{v-\ell}}{|\mathcal H|^2\,(k)_v^2}
\le
(1+o(1))\sum_{\ell=0}^{v}v!\,n^{\ell}k^{-(v+\ell)}\,p_{v,\ell}(\mathcal H)
\le
O\left(\max_{0\le \ell\le v}\Bigl\{(v+1)!\,n^{\ell}k^{-(v+\ell)}\,p_{v,\ell}(\mathcal H)\Bigr\}\right).
\]

For $\ell=v$, the probability $p_{v,v}(\sH)$ is that a uniform permutation in $S_v$ maps independent uniform $H,H'\in\sH$ to one another.
By Lemma~\ref{lem:iso-sum-vfact}, this is
\[
\frac{\sum_{H, H' \in \sH} |\Iso(H, H')|}{|\sH|^2 v!} = \frac{1}{|\sH|},
\]
so
\[
(v + 1)! n^{v} k^{-2v} p_{v,v}(\sH) = \frac{(v + 1)! (n / k^2)^v}{|\sH|},
\]
which is absorbed by the first term of the bound.
The remaining maximum is over $0\le\ell\le v-1$, as in Lemma~\ref{lem:sufficient-conditions-overlap}.
Together with the null variance estimate, these bounds prove the lemma.
\end{proof}

\subsection{Proof of Theorem \ref{thm:separation-hardness}}\label{apx:separation}

\begin{proof}
Let $a$ be the number of vertices shared by two embeddings, and let $\ell$ count those absent from the symmetric difference.
Then $\sigma(H,H';\phi,\psi)=2v-a-\ell$, with $u=a$ in the notation of Lemma~\ref{lem:sufficient-conditions-overlap}.

Choose $0<\eta,\eta'\ll\sqrt\theta$, then $0<C_2,\varepsilon\ll\eta,\eta'$ with $C_2<1/4-\varepsilon/2$, and finally $C_1$ sufficiently large depending on $\beta,\theta,\eta,\eta'$.
We take $k=n^{1/2-\varepsilon}$, so $v=o(\sqrt k)$.
For sufficiently large $v$, the condition $\mathrm{VC}(\beta)$ excludes isolated vertices: if $x$ were isolated, taking $S=\{x\}$ in \eqref{eq:vcb} would give
\[
v-1\le\max\{2^{-\beta}(v-1),2\log v/\beta\}<v-1,
\]
a contradiction. Lemma~\ref{lem:sufficient-conditions-overlap} therefore applies.
It remains to show that
\begin{align*}
E_0&\colonequals\frac{(n/k^2)^v (v+1)!}{|\mathcal H|},\\
E_1&\colonequals
\max_{\substack{0\le \ell\le a\le v-1\\(a,\ell)\ne(0,0)}}
\frac{v^{2a}n^\ell}{k^{a+\ell}}
p_{a,\ell}(\mathcal H),\\
E_2&\colonequals
\max_{0\le \ell\le v-1}
\frac{(v+1)!n^\ell}{k^{v+\ell}}
p_{v,\ell}(\mathcal H)
\end{align*}
are $o(1)$.
The case $a=\ell=0$ is excluded because the embeddings are disjoint, so $\sigma=2v$ and their covariance under the planted model is nonpositive.

For $E_0$, the bounds $|\mathcal H|\ge2^{\theta v^2}$ and $v\ge C_1\log n$ give $|\mathcal H|\ge n^{\gamma v}$ for some $\gamma=\gamma(\theta)>0$, after increasing $C_1$ if necessary.
Since $v\le n^{C_2}$ and $k=n^{1/2-\varepsilon}$,
\[
E_0
\le
\frac{n^{2\varepsilon v}(v+1)!}{|\mathcal H|}
\le
n^{(2\varepsilon+C_2+o(1)-\gamma)v}
=o(1)
\]
provided $C_2$ and $\varepsilon$ are chosen sufficiently small relative to $\gamma$.

To bound $E_1$ and $E_2$, we use the following claim.
\begin{claim}[Overlap cancellation bound]\label{clm:hardness-overlap-cancel}
There is a constant $c_\beta>0$, depending only on $\beta$, such that for all $0\le \ell\le a\le v$,
\[
p_{a,\ell}(\mathcal H)
\le
\binom a\ell
\min\left\{
2^{-c_\beta \ell(v-a)},
2^{-\theta v^2+(v-\ell)^2}
\right\}.
\]
\end{claim}

\begin{proof}
Choose the $\ell$ cancelled vertices in the overlap, giving the factor $\binom a\ell$.
We bound the probability of cancellation for each choice in two ways.

First, fix the two image sets and a set $L_0$ of $\ell$ vertices in their intersection, leaving the preimages random.
Put
\[
I=\phi^{-1}(\phi([v])\cap\psi([v])),
\qquad
L=\phi^{-1}(L_0),
\]
and define $I',L'$ analogously using $\psi$.
The bijections defining the embeddings are independent and uniform, so $L$ is a uniform $\ell$-subset of $[v]$ and, conditional on $L$, $[v]\setminus I$ is a uniform $(v-a)$-subset of $[v]\setminus L$.
The pair $(L',[v]\setminus I')$ has the same law and is independent of $(L,[v]\setminus I)$.

For a vertex in $L_0$ to disappear from $V(\phi(H)\triangle\psi(H'))$, it can have no edge to a vertex outside the overlap in either copy, since that edge would remain in the symmetric difference.
Thus cancellation requires
\[
E_H(L,[v]\setminus I)=\emptyset,
\qquad
E_{H'}(L',[v]\setminus I')=\emptyset.
\]
Applying Lemma~\ref{lem:VC-implies} to these two independent random pairs, and then averaging over $H,H'$, gives
\[
\Pr[\text{this necessary condition holds}]
\le 2^{-(\beta/4)\ell(v-a)}
\le 2^{-c_\beta \ell(v-a)}
\]
for $c_\beta=\beta/8$.

Second, fix the cancelled vertices, both embeddings, and $H$.
Cancellation determines every adjacency of $H'$ incident to $L'$.
Only edges among the remaining $v-\ell$ vertices are free, giving at most $2^{\binom{v-\ell}{2}}\le2^{(v-\ell)^2}$ choices for $H'$.
Since $|\mathcal H|\ge2^{\theta v^2}$, the conditional probability is at most $2^{-\theta v^2+(v-\ell)^2}$.
Taking the minimum of the two bounds proves the claim.
\end{proof}

We now bound $E_1$.
For $k=n^{1/2-\varepsilon}$ and $v\le n^{C_2}$, the remaining factor satisfies
\[
\frac{v^{2a}n^\ell}{k^{a+\ell}}
\le
n^{a(2C_2-1/2+\varepsilon)+\ell(1/2+\varepsilon)}.
\]

If $\ell\le(1-\eta)a$, this exponent is at most $a(2C_2+2\varepsilon-\eta(1/2+\varepsilon))=-\Omega_\theta(a)$ by the choice $C_2,\varepsilon\ll\eta$.
The binomial factor in Claim~\ref{clm:hardness-overlap-cancel} is at most $2^a=n^{o(a)}$, so the contribution is uniformly $o(1)$.

Next suppose $\ell>(1-\eta)a$ but $a\le(1-\eta')v$.
The first bound in Claim~\ref{clm:hardness-overlap-cancel} gives
\[
p_{a,\ell}(\mathcal H)
\le
2^a\,2^{-c_\beta \ell(v-a)}
\le
2^a\,2^{-\Omega_{\beta,\eta'}(av)}.
\]
Since $v\ge C_1\log n$, choosing $C_1$ sufficiently large makes this $n^{-\Omega(a)}$, which dominates the polynomial prefactor.

Finally suppose $\ell>(1-\eta)a$ and $a>(1-\eta')v$.
Then
\[
v-\ell\le (v-a)+(a-\ell)\le(\eta'+\eta)v.
\]
Choosing $\eta,\eta'$ sufficiently small in terms of $\theta$, the second bound in Claim~\ref{clm:hardness-overlap-cancel} yields
\[
p_{a,\ell}(\mathcal H)
\le
2^a\,2^{-\theta v^2+(v-\ell)^2}
\le
2^v\,2^{-\Omega_\theta(v^2)}.
\]
Since $v\ge C_1\log n$, this decay dominates the prefactor, giving $E_1=o(1)$.

For $E_2$, we have $a=v$ and $0\le\ell\le v-1$.
At $\ell=0$,
\[
(v+1)!k^{-v}p_{v,0}(\mathcal H)
\le
(v+1)!n^{-(1/2-\varepsilon)v}
\le
n^{-(1/2-\varepsilon-C_2-o(1))v}
=o(1)
\]
by the choice of $C_2,\varepsilon$.
For the remaining cases, the prefactor is
\[
\frac{(v+1)!n^\ell}{k^{v+\ell}}
=
(v+1)!n^{-(1/2-\varepsilon)v+(1/2+\varepsilon)\ell}.
\]
If $\ell\le(1-\eta)v$, the exponent of $n$ is at most $(2\varepsilon-\eta(1/2+\varepsilon))v=-\Omega_\theta(v)$ for $\varepsilon\ll\eta$.
Also, $(v+1)!\le n^{(C_2+o(1))v}$ because $v\le n^{C_2}$.
Taking $C_2,\varepsilon$ sufficiently small relative to $\eta$ makes the contribution $o(1)$.

If $\ell>(1-\eta)v$, then the bound in Claim~\ref{clm:hardness-overlap-cancel} that uses $|\mathcal H|$ gives
\[
p_{v,\ell}(\mathcal H)
\le
2^v\,2^{-\theta v^2+(v-\ell)^2}
\le
2^v\,2^{-\Omega_\theta(v^2)}
\]
after choosing $\eta\ll\sqrt\theta$.
This dominates $(v+1)!$ and the polynomial prefactor because $v\ge C_1\log n$.
Hence $E_2=o(1)$.

All three terms in Lemma~\ref{lem:sufficient-conditions-overlap} are $o(1)$, so $F_{\mathcal H}$ strongly separates $\PP_{n,k}$ from $\QQ_n$.
Conjecture~\ref{conj:clique} therefore rules out a polynomial-time algorithm for computing $F_{\mathcal H}$ in this regime.
\end{proof}

\subsection{Proof of Lemma \ref{lem:rooted-moments}}\label{sec:proof_root}

\begin{proof}
We first compute the mean under the planted model.
The character $\chi_{H,\phi}(G)$ has nonzero expectation exactly when every edge of $H$ is mapped inside the planted clique, so
\begin{equation}\label{eq:char-expect}
\E\bigl[\chi_{H,\phi}(G)\mid C\bigr]
=
\begin{cases}
1, & \text{if }\phi(V(H))\subseteq C,\\
0, & \text{otherwise.}
\end{cases}
\end{equation}

\paragraph{Case 1: $i\in C$.} For every fixed planted set $C$ containing $i$, taking expectation in \eqref{eq:rooted-sum} and using \eqref{eq:char-expect} gives
\begin{equation*}
\begin{aligned}
\E\bigl[F_{\sH,i}\mid C\bigr]
&=
\sum_{H\in\sH}\sum_{r\in\mathcal R(H)}
\sum_{\substack{\phi:V(H)\hookrightarrow[n]\\ \phi(r)=i}}
\E\bigl[\chi_{H,\phi}(G)\mid C\bigr] \\
&=
\sum_{H\in\sH}\sum_{r\in\mathcal R(H)}
\#\Bigl\{\phi:V(H)\hookrightarrow C:\ \phi(r)=i\Bigr\}\\
&=
\nu(\sH)\cdot (k-1)_{v-1}.
\end{aligned}
\end{equation*}
This value is the same for every $C$ containing $i$, so it also equals $\E[F_{\sH,i}\mid i\in C]$.

\paragraph{Case 2: $i\notin C$.} Every embedding with $\phi(r)=i$ has $\phi(V(H))\nsubseteq C$.
By \eqref{eq:char-expect}, each term has mean zero, so $\E[F_{\sH,i}\mid C]=0$ for every such $C$ and $\E[F_{\sH,i}\mid i\notin C]=0$.

To compute the variance, fix $C$ and index the summands by triples $a=(H,r,\phi)$ with $H\in\sH$, $r\in\mathcal R(H)$, and $\phi\in\Inj(V(H),[n])$ satisfying $\phi(r)=i$.
Let $\sI_i$ be the set of these triples and write $\chi_a(G)\colonequals\chi_{H,\phi}(G)$.
Then
\begin{equation}\label{eq:var-expand}
\Var\bigl[F_{\sH,i}\mid C\bigr]
=
\sum_{a,b\in\sI_i}
\Bigl(\E[\chi_a\chi_b\mid C]-\E[\chi_a\mid C]\E[\chi_b\mid C]\Bigr).
\end{equation}

For $a=(H,r,\phi)$, define its set of random edges by
\[
R(a)\colonequals\Bigl\{\{\phi(u),\phi(w)\}:\ \{u,w\}\in E(H)\ \text{and}\ \{\phi(u),\phi(w)\}\not\subseteq C\Bigr\}.
\]
For any $a,b$,
\[
\E[\chi_a\chi_b\mid C]=
\begin{cases}
1, & \text{if }R(a)=R(b)\text{ as sets of edges,}\\
0, & \text{otherwise.}
\end{cases}
\]
Similarly, $\E[\chi_a\mid C]=\mathbf 1\{R(a)=\emptyset\}$, where $R(a)=\emptyset$ is equivalent to $\phi(V(H))\subseteq C$.

\paragraph{Case 1: $i\notin C$.}
Every $a\in\sI_i$ has $R(a)\ne\emptyset$ and $\E[\chi_a\mid C]=0$.
Thus \eqref{eq:var-expand} becomes
\[
\Var\bigl[F_{\sH,i}\mid C\bigr]
=
\sum_{a,b\in\sI_i}\E[\chi_a\chi_b\mid C]
=
\#\{(a,b)\in\sI_i^2:\ R(a)=R(b)\}.
\]

For $a=(H,r,\phi)\in\sI_i$, put $s(a)\colonequals|\phi(V(H))\cap C|\in\{0,1,\dots,v\}$ and $\sI_{i,s}\colonequals\{a\in\sI_i:s(a)=s\}$.
Write $E^{\mathrm{out}}(a)\colonequals\phi(E(H))\setminus\binom C2$ for the embedded edges not contained in $C$.
Conditional on $C$, the character depends only on these edges, and $\E[\chi_a\chi_b\mid C]=1$ exactly when $E^{\mathrm{out}}(a)=E^{\mathrm{out}}(b)$.
Keeping all root choices in $\sI_i$, we obtain
\[
\#\{(a,b)\in\sI_i^2:\ R(a)=R(b)\}
=
\sum_{s,t=0}^{v}\#\{(a,b)\in\sI_{i,s}\times\sI_{i,t}:\ E^{\mathrm{out}}(a)=E^{\mathrm{out}}(b)\}.
\]

For $s=0$, we have $\phi(V(H))\cap C=\emptyset$ and $E^{\mathrm{out}}(a)=\phi(E(H))$.
Collisions arise only from relabelings, so each $a\in\sI_{i,0}$ has exactly $v!$ matches $b\in\sI_{i,0}$.
Hence
\[
\#\{(a,b)\in\sI_{i,0}^2:\ R(a)=R(b)\}=|\sI_{i,0}|\cdot v!.
\]

\begin{claim}\label{clm:bounded-treewidth-types}
For every fixed $t\ge1$, there is a constant $A_t<\infty$ such that the number of isomorphism classes of graphs on $s$ vertices with treewidth at most $t$ is at most $A_t^s$ for every $s\ge1$.
\end{claim}

\begin{proof}
It suffices to consider $s\ge t+1$, since increasing $A_t$ covers the remaining cases.
Every graph of treewidth at most $t$ is a spanning subgraph of a $t$-tree.
A $t$-tree starts from a clique of size $t+1$ and repeatedly adds a vertex adjacent to all vertices of an existing $t$-clique.

To count isomorphism classes of $t$-trees, root a construction at its initial $(t+1)$-clique and order the vertices in each clique.
Each later maximal clique shares exactly $t$ vertices with a parent maximal clique, giving a rooted tree with $s-t$ nodes.
Ordering the children gives a plane rooted tree, with at most $4^s$ choices of shape.
For each child and its parent, there are only a bounded number of ways, depending on $t$, to identify the shared vertices and specify the position of the new vertex.
Thus there are at most $D_t^s$ encodings, and hence isomorphism classes, for some $D_t<\infty$.

A $t$-tree on $s$ vertices has $\binom{t+1}{2}+t(s-t-1)$ edges and at most $2^{ts+O_t(1)}$ spanning subgraphs.
Combining these counts gives $A_t^s$ after increasing the constant.
\end{proof}

We next bound the number of collisions with a fixed $a=(H,r,\phi)\in\sI_i$.
Put $J\colonequals\phi^{-1}(C)$ and $s\colonequals|J|$.
Let
\[
\partial J\colonequals\{u\in J:\ N_H(u)\cap(V(H)\setminus J)\ne\emptyset\}
\]
be the vertices in $J$ with a neighbor outside $J$.
Put $d\colonequals|\partial J|$ and $h\colonequals s-d$.
Up to relabeling, $E^{\mathrm{out}}(a)$ determines the part of the rooted pattern outside $J$ and the images of $\partial J$.
Preserving this edge set allows only the images of $J\setminus\partial J$ and the edges inside $J$ to change.
For $\mathcal H\subseteq\mathcal G^{\mathrm{conn}}_{v,\le t}$, there is therefore a constant $A_t<\infty$ depending only on $t$ such that
\begin{equation}\label{eq:local-rooted-collision}
\#\{b\in\sI_i:\ E^{\mathrm{out}}(b)=E^{\mathrm{out}}(a)\}
\le
v!\,A_t^s\,s!\binom{k-d}{h}
\le
v!\,A_t^s\,(s)_d(k)_h.
\end{equation}
Indeed, there are at most $v!$ assignments of pattern labels and, by Claim~\ref{clm:bounded-treewidth-types}, at most $A_t^s$ isomorphism classes for the induced graph on $J$.
There are then at most $\binom{k-d}{h}$ choices for the remaining images in $C$ and $s!$ ways to assign the vertices of $J$ to these images and the fixed boundary images.
The last inequality uses $s!\binom{k-d}{h}=(s)_d(k-d)_h\le(s)_d(k)_h$.
Since $H$ is connected and $s<v$, we have $d\ge1$.

If $i\notin C$, then for each fixed $(H,r)$,
\[
\#\{\phi:\phi(r)=i,\ |\phi(V(H))\cap C|=s\}
=
\binom{v-1}{s}(k)_s (n-k-1)_{v-1-s}.
\]
Multiplying by the number of choices of $(H,r)$, namely $\nu(\mathcal H)$, yields
\begin{equation}\label{eq:Iis-notinC}
|\sI_{i,s}|
=
\nu(\mathcal H)\cdot \binom{v-1}{s}(k)_s (n-k-1)_{v-1-s},
\qquad i\notin C,\ \ s=0,1,\dots,v-1.
\end{equation}

\begin{claim}\label{lem:ratio-Is}
Assume $k=\lambda\sqrt n$ for fixed $\lambda>0$ and $v=O(\log n/\log\log n)$.
Fix $i\notin C$.
Then
\[
\sum_{s=1}^{v-1}
\sum_{a\in\sI_{i,s}}
\#\{b\in\sI_i:\ E^{\mathrm{out}}(b)=E^{\mathrm{out}}(a)\}
=
o(|\sI_{i,0}|\,v!).
\]
\end{claim}

\begin{proof}
By \eqref{eq:local-rooted-collision} and $d\ge1$, uniformly over $1\le s\le v-1$ we have
\[
(s)_d(k)_{s-d}\le(1+o(1))v(k)_{s-1}.
\]
Thus the left-hand side is at most
\[
\sum_{s=1}^{v-1}
(1+o(1))|\sI_{i,s}|v!A_t^s v(k)_{s-1}.
\]
Dividing by $|\sI_{i,0}|v!$ and using \eqref{eq:Iis-notinC}, we get
\[
\frac{|\sI_{i,s}|v!A_t^s v(k)_{s-1}}{|\sI_{i,0}|v!}
\le
\binom{v-1}{s}vA_t^s\frac{(k)_s(k)_{s-1}}{(n-k-v+s)_s}.
\]
Since $v=o(n)$ and $k=\lambda\sqrt n$, the last expression is bounded by $n^{-1/2}v\binom{v-1}{s}(C_t)^s$ for a constant $C_t<\infty$.
Summing over $1\le s\le v=O(\log n/\log\log n)$ gives at most $n^{-1/2}v(1+C_t)^v=n^{-1/2+o(1)}=o(1)$.
\end{proof}

Therefore,
\begin{align*}
\Var\bigl[F_{\sH,i}\mid C\bigr]
&\le
|\sI_{i,0}|\cdot v!+
\sum_{s=1}^{v-1}
\sum_{a\in\sI_{i,s}}
\#\{b\in\sI_i:\ E^{\mathrm{out}}(b)=E^{\mathrm{out}}(a)\} \\
&=
(1+o(1))|\sI_{i,0}|\cdot v! \\
&=
(1+o(1))\nu(\mathcal H)\cdot (n-k-1)_{v-1}\cdot v!.
\end{align*}
Every collision contributes a nonnegative term, and those with $s=0$ contribute $|\sI_{i,0}|v!$.
This gives the matching lower bound.

\paragraph{Case 2: $i\in C$.}
Some summands now have $R(a)=\emptyset$ and nonzero mean.
Recall that $\E[\chi_a\chi_b\mid C]=\mathbf 1\{R(a)=R(b)\}$ and $\E[\chi_a\mid C]=\mathbf 1\{R(a)=\emptyset\}$.
Let $\sI_i^{\mathrm{in}}\colonequals\{a\in\sI_i:R(a)=\emptyset\}$ and $\sI_i^{\mathrm{out}}\colonequals\sI_i\setminus\sI_i^{\mathrm{in}}$.
Substituting into~\eqref{eq:var-expand} gives
\[
\begin{aligned}
\Var\bigl[F_{\mathcal H,i}\mid C\bigr]
&=
\sum_{a,b\in\sI_i}
\Bigl(\mathbf 1\{R(a)=R(b)\}-\mathbf 1\{R(a)=\emptyset\}\mathbf 1\{R(b)=\emptyset\}\Bigr)\\
&=
\#\bigl\{(a,b)\in\sI_i^2:\ R(a)=R(b)\neq\emptyset\bigr\}\\
&=
\#\bigl\{(a,b)\in(\sI_i^{\mathrm{out}})^2:\ R(a)=R(b)\bigr\}\\
&\le
\sum_{s=1}^{v-1}\sum_{a\in\sI_{i,s}}
\#\bigl\{b\in\sI_i:\ R(b)=R(a)\bigr\},
\end{aligned}
\]
where $s(a)\colonequals|\phi(V(H))\cap C|\in\{1,2,\dots,v\}$ and $\sI_{i,s}\colonequals\{a\in\sI_i:s(a)=s\}$ for $a=(H,r,\phi)\in\sI_i$.
For $s=1$, no edge lies fully inside the planted part of the embedded pattern, so collisions are only relabelings and contribute at most $|\sI_{i,1}|v!$.
For $s\ge2$, we use \eqref{eq:local-rooted-collision}.

If $i\in C$, then for each fixed $(H,r)$,
\[
\#\{\phi:\phi(r)=i,\ |\phi(V(H))\cap C|=s\}
=
\binom{v-1}{s-1}(k-1)_{s-1}(n-k)_{v-s}.
\]
Multiplying by $\nu(\mathcal H)$ yields
\begin{equation}\label{eq:Iis-inC}
|\sI_{i,s}|
=
\nu(\mathcal H)\cdot \binom{v-1}{s-1}(k-1)_{s-1}(n-k)_{v-s},
\qquad i\in C,\ \ s=1,2,\dots,v.
\end{equation}

\begin{claim}\label{lem:dominant-s1-inC}
Assume $k=\lambda\sqrt n$ for fixed $\lambda>0$ and $v=O(\log n/\log\log n)$.
Fix $i\in C$.
Then
\[
\sum_{s=2}^{v-1}
\sum_{a\in\sI_{i,s}}
\#\{b\in\sI_i:\ E^{\mathrm{out}}(b)=E^{\mathrm{out}}(a)\}
=
o(|\sI_{i,1}|\,v!).
\]
\end{claim}

\begin{proof}
Since $i\in C$, the root $r$ belongs to $J$.
If $r\in\partial J$ and $d=1$, then $H-r$ has no edge between the nonempty sets $J\setminus\{r\}$ and $V(H)\setminus J$, contradicting $r\in\mathcal R(H)$.
Thus $d\ge2$ in this case, and
\[
(s)_d(k)_{s-d}\le(1+o(1))v^2(k)_{s-2}.
\]
If $r\notin\partial J$, then the fixed root image is determined in addition to the images of the nonempty set $\partial J$, and the same counting argument instead has the factor
\[
(s)_{d+1}(k)_{s-d-1}\le(1+o(1))v^2(k)_{s-2}.
\]
In either case, at least two images in $C$ are fixed.
Summing over the possible boundary sets $\partial J$ bounds the $s$th summand by $(1+o(1))|\sI_{i,s}|v!A_t^s v^2(k)_{s-2}$, after adjusting $A_t$ by a constant depending only on $t$.
Using \eqref{eq:Iis-inC} and dividing by $|\sI_{i,1}|v!$ gives
\[
\frac{|\sI_{i,s}|v!A_t^s v^2(k)_{s-2}}{|\sI_{i,1}|v!}
\le
\binom{v-1}{s-1}v^2A_t^s\frac{(k-1)_{s-1}(k)_{s-2}}{(n-k-v+s)_{s-1}}.
\]
Since $k=\lambda\sqrt n$ and $v=o(n)$, this is at most $n^{-1/2}v^2\binom{v-1}{s-1}(C_t)^s$.
Summing over $2\le s\le v=O\!\left(\frac{\log n}{\log\log n}\right)$ gives $n^{-1/2+o(1)}=o(1)$.
\end{proof}

This implies
\begin{align*}
\Var\bigl[F_{\sH,i}\mid C\bigr]
&\le
|\sI_{i,1}|v!+
\sum_{s=2}^{v-1}
\sum_{a\in\sI_{i,s}}
\#\{b\in\sI_i:\ E^{\mathrm{out}}(b)=E^{\mathrm{out}}(a)\}\\
&=
(1+o(1))|\sI_{i,1}|v!\\
&=
(1+o(1))\nu(\mathcal H)\cdot (n-k)_{v-1}\cdot v!,
\end{align*}
as required for each fixed $C$.
The conditional mean $\E[F_{\sH,i}\mid C]$ depends only on whether $i\in C$.
The law of total variance therefore gives the claimed bounds conditional on $i\in C$ or $i\notin C$ by averaging over $C$.
The matching lower bound when $i\notin C$ is preserved.
\end{proof}

\subsection{Proof of Lemma~\ref{lem:p-var}}\label{apx:rooted-mse}

\begin{proof}
Put $r\colonequals n/k^2\le\lambda^{-2}$.
We group collisions as in the proof of Lemma~\ref{lem:rooted-moments}, according to the number $s$ of pattern vertices mapped into $C$.
The bound \eqref{eq:local-rooted-collision} gives the following estimates for some $A_t<\infty$ depending only on $t$.

First suppose $i\notin C$.
After division by $\mu^2$, the contribution from $s=0$ is at most $(1+o(1))(v!/\nu(\sH))r^{v-1}$.
For $1\le s\le v-1$, equations \eqref{eq:local-rooted-collision} and \eqref{eq:Iis-notinC}, together with $(k-1)_{v-1}=(1+o(1))k^{v-1}$, give the uniform bound
\[
\frac{|\sI_{i,s}|\,v!A_t^s v(k)_{s-1}}{\mu^2}
\le
(1+o(1))\frac{v!}{\nu(\sH)}
\binom{v-1}{s}\frac{vA_t^s}{k}r^{v-1-s}.
\]
Summing over $s$ yields
\begin{equation}
\frac{\Var(F_{\sH,i}\mid i\notin C)}{\mu^2}
\le
(1+o(1))\frac{v!}{\nu(\sH)}
\left(r^{v-1}+\frac{v(r+A_t)^{v-1}}{k}\right).
\label{eq:rooted-normalized-notin}
\end{equation}

Now suppose $i\in C$.
The contribution from $s=1$ has the same bound $(1+o(1))(v!/\nu(\sH))r^{v-1}$.
For $2\le s\le v-1$, at least two images in $C$ are fixed by the root condition, as in Claim~\ref{lem:dominant-s1-inC}, so
\[
\frac{|\sI_{i,s}|\,v!A_t^s v^2(k)_{s-2}}{\mu^2}
\le
(1+o(1))\frac{v!}{\nu(\sH)}
\binom{v-1}{s-1}\frac{v^2A_t^s}{k}r^{v-s}.
\]
After increasing $A_t$ by a constant factor if necessary, this gives
\begin{equation}
\frac{\Var(F_{\sH,i}\mid i\in C)}{\mu^2}
\le
(1+o(1))\frac{v!}{\nu(\sH)}
\left(r^{v-1}+\frac{v^2A_t(r+A_t)^{v-1}}{k}\right).
\label{eq:rooted-normalized-in}
\end{equation}

Choose a fixed $c_1$ with $c_0<c_1<c(\sH)$.
For all sufficiently large $v$, $\nu(\sH)\ge2|\sH|\ge2c_1^v v!$.
The first term in each of \eqref{eq:rooted-normalized-notin} and \eqref{eq:rooted-normalized-in} is therefore at most $\frac{1+o(1)}{2c_1}(c_1\lambda^2)^{-(v-1)}$.
The second is $n^{-1/2+o(1)}=o((c_0\lambda^2)^{-(v-1)})$, since $k\ge\lambda\sqrt n$ and $v=o(\log n/\log\log n)$.
By the conditional means in Lemma~\ref{lem:rooted-moments}, these normalized variances equal the stated mean square errors.
Averaging over the two conditioning events gives the unconditional bound.
\end{proof}

\subsection{Proof of Lemma \ref{lem:partition-clique-size}}\label{apx:proof-recover1}

\begin{proof}
The uniform partition with prescribed part sizes makes $k_1$ and $k_2=k-k_1$ hypergeometric, with means $(n_1/n)k$ and $(n_2/n)k$.
Since $n_2=\lfloor n^{3/5}\rfloor$, we have $n_1/n=1-n^{-2/5}+o(1)=1-o(1)$ and $\E[k_1]=(1-o(1))k$.
The hypergeometric tail bound of \cite{chvatal1979tail} gives, for each fixed $\eta>0$,
\[
\Pr\bigl[|k_1-\E[k_1]|\ge \eta k\bigr]
\le
2\exp(-\Omega_\eta(k)).
\]
Since $k\ge\lambda\sqrt n\to\infty$, we obtain $k_1=(1-o(1))k$ with probability $1-o(1)$.

For the sharper estimate of $k_1-\widetilde k_1$, the same bound gives
\[
\Pr\!\left[\left|k_1-\frac{n_1}{n}k\right|>\sqrt{k}\log n\right]
\le
2\exp(-2(\log n)^2)
=o(1).
\]
Also,
\[
\left|\frac{n_1}{n}k-\lfloor(1-n^{-2/5})k\rfloor\right|
=O(k/n+1).
\]
The triangle inequality proves \eqref{eq:k1-proxy-close}.

For $k_2$, we have $\E[k_2]=(n_2/n)k\ge(1-o(1))\lambda n^{1/10}\ge\lambda n^{1/10}/2$ for all sufficiently large $n$.
From the multiplicative Chernoff bound for hypergeometric variables \cite{chvatal1979tail},
\[
\Pr\!\left[k_2\le \frac12 \E[k_2]\right]\le \exp\!\left(-\frac{\E[k_2]}{8}\right),
\]
we obtain
\[
\Pr\!\left[k_2\le \frac{\lambda}{4} n^{1/10}\right]
\le
\exp\!\bigl(-\Omega(n^{1/10})\bigr)
=o(1).
\]
Thus $k_2>\lambda n^{1/10}/4$ with probability $1-o(1)$.
\end{proof}
\subsection{Proof of Lemma~\ref{lem:filter-V2}}\label{apx:proof-recover2}

\begin{proof}
Condition on $C$, the partition $[n]=V_1\sqcup V_2$, the estimator's independent randomness (including its colorings), and $V_3$.
We suppress the partition and the estimator's randomness from the notation below.
The set $V_3$ depends only on $G[V_1]$ and this randomness, so the edges between $V_1$ and $V_2$ retain their conditional independence and Bernoulli laws.
Put $S=V_3\cap C_1$ and $s=|S|$.

For $v\in V_2\setminus C_2$, the edges from $v$ to $V_3$ are independent Bernoulli$(1/2)$ variables, so $|E(\{v\},V_3)|\sim\Bin(|V_3|,1/2)$.
Hoeffding's inequality gives
\[
\Pr\!\left[v\in \widehat{C}_{2,\mathrm{ideal}}\mid C,V_3\right]
=
\Pr\!\left[|E(\{v\}, V_3)|-\frac{|V_3|}{2}\ge \frac{s}{4}\,\middle|\,C,V_3\right]
\le
\exp\!\left(-\frac{2(s/4)^2}{|V_3|}\right)
=
\exp\!\left(-\frac{s^2}{8|V_3|}\right).
\]
By \eqref{eq:V3-snr}, $s^2/|V_3|=\omega((\log n)^3)$, so this probability is at most $\exp(-\omega((\log n)^3))$.
A union bound over $v\in V_2\setminus C_2$, using $|V_2|=n^{0.6}+O(1)$, gives
\[
\Pr\!\left[\widehat{C}_{2,\mathrm{ideal}}\setminus C_2\neq\varnothing\mid C,V_3\right]=o(1).
\]
For $v\in C_2$, all edges to $S$ are present, so $|E(\{v\},V_3)|=s+\Bin(|V_3|-s,1/2)$.
Its conditional mean is $|V_3|/2+s/2$, exceeding $T_{\mathrm{ideal}}$ by $s/4$.
Hoeffding's inequality gives
\begin{align*}
\Pr\!\left[v\notin \widehat{C}_{2,\mathrm{ideal}}\mid C,V_3\right]
&\le
\Pr\!\left[|E(\{v\}, V_3)|-\E[|E(\{v\}, V_3)|\mid C,V_3]\le -\frac{s}{4}\,\middle|\,C,V_3\right] \\
&\le
\exp\!\left(-\frac{2(s/4)^2}{|V_3|}\right) \\
&\le
\exp\!\bigl(-\omega((\log n)^3)\bigr).
\end{align*}
Taking a union bound over all $v\in C_2$, we obtain
\[
\Pr\!\left[C_2\not\subseteq \widehat{C}_{2,\mathrm{ideal}}\mid C,V_3\right]=o(1).
\]
The two bounds give $\widehat C_{2,\mathrm{ideal}}=C_2$ with probability $1-o(1)$.

For the implementable threshold, the assumed event gives $s=(1-o(1))k_1=(1-o(1))k$, and also $\widetilde k_1=(1-o(1))k$.
Thus, for all sufficiently large $n$,
\[
0.2499\,\widetilde k_1\le \frac{s}{3},
\qquad
\frac{s}{2}-0.2499\,\widetilde k_1\ge \frac{s}{5},
\]
and also $0.2499\,\widetilde k_1\ge c_0s$ for some absolute constant $c_0>0$.

For $v\in V_2\setminus C_2$, the preceding lower bound on $0.2499\,\widetilde k_1$ and Hoeffding's inequality give
\[
\Pr[v\in \widehat C_2\mid C,V_3]
\le
\exp\!\left(-\Omega\!\left(\frac{s^2}{|V_3|}\right)\right).
\]
For $v\in C_2$, the mean remains $|V_3|/2+s/2$.
The second inequality above and Hoeffding's inequality give
\[
\Pr[v\notin \widehat C_2\mid C,V_3]
\le
\exp\!\left(-\Omega\!\left(\frac{s^2}{|V_3|}\right)\right).
\]
Using $s^2/|V_3|=\omega((\log n)^3)$ from \eqref{eq:V3-snr}, a union bound over $|V_2|\le n$ vertices gives $\widehat C_2=C_2$ with probability $1-o(1)$.
\end{proof}
\subsection{Proof of Lemma~\ref{lem:recover-K1-from-K2}}\label{apx:proof-recover3}

\begin{proof}
On the event $\widehat C_2=C_2$, every $v\in C_1$ is adjacent to all vertices of $\widehat C_2$, so $C_1\subseteq\widehat C_1$.
By Lemma~\ref{lem:partition-clique-size}, $k_2\ge\lambda n^{1/10}/4$ with probability $1-o(1)$.
On this event, we bound the joint failure probability by
\begin{align*}
&\Pr\bigl[
\widehat C_2=C_2,\ \widehat C_1\setminus C_1\ne\emptyset
\bigr]\\
&\qquad\le
\Pr\bigl[
\exists v\in V_1\setminus C_1:
v\text{ is adjacent to every vertex of }C_2
\bigr]\\
&\qquad\le
n\,2^{-k_2}
\le
n\,2^{-\lambda n^{1/10}/4}
=o(1).
\end{align*}
Since $C_1\subseteq\widehat C_1$ whenever $\widehat C_2=C_2$, this proves $\Pr[\widehat C_2=C_2,\widehat C_1\ne C_1]=o(1)$.
If $\Pr[\widehat C_2=C_2]=1-o(1)$, as in Lemma~\ref{lem:implementable-threshold}, a union bound gives $\Pr[\widehat C_1\sqcup\widehat C_2=C]=1-o(1)$.
\end{proof}
\section{Omitted proofs from Section \ref{sec:approx}}
\subsection{Proof of Lemma~\ref{lem:tree_algo}}
\label{apx:tree_algo}
\begin{proof}
We prove correctness by induction on $|V(H)|$, using the color palette $[v]$ of the original tree throughout.

If $|V(H)|=1$, then $H$ consists only of the root $r$.
For each $x\in[n]$, the unique embedding with $\phi(r)=x$ uses the color set $\{\zeta(x)\}$ and has empty edge product $1$.
Thus $\DP_{H,r}(x,C)=\mathbf 1\{C=\{\zeta(x)\}\}$, as initialized by the algorithm.

Assume $|V(H)|>1$ and the claim holds for smaller rooted trees.
Choose an edge $e_T=(r,r')\in E(H)$, and let $(H_1,r)$ and $(H_2,r')$ be the rooted trees obtained by deleting it.
Fix $x\in[n]$ and $C\subseteq [v]$.
An embedding of $H$ with root image $x$ is colorful exactly when its restrictions to $H_1$ and $H_2$ are colorful and use disjoint color sets $C_1,C_2$.
Conversely, each such pair gives a unique colorful embedding of $H$.
Its signed contribution is the product of the two subtree contributions and the weight $X_{xy}$ of the connecting edge, where $y$ is the image of $r'$.
Therefore
\[
\DP_{H,r}(x,C)
=
\sum_{y\in[n]\setminus\{x\}}
\ \sum_{\substack{C_1,C_2\subseteq [v]\\
C_1\cap C_2=\emptyset\\
C_1\cup C_2=C}}
\DP_{H_1,r}(x,C_1)
\DP_{H_2,r'}(y,C_2)
X_{xy},
\]
which is the algorithm's recurrence.
The induction hypothesis gives the correct tables for $H_1$ and $H_2$, so the table for $H$ is also correct.

Each merge scans at most $n(n-1)$ ordered pairs of host vertices and $3^v$ disjoint ordered pairs of color sets, taking $2^{O(v)}n^2$ time.
The recursion tree has $2v-1$ nodes, and initializing all leaves takes $O(v2^vn)$ time.
The total running time is therefore $O(v)2^{O(v)}n^2+O(v2^vn)=2^{O(v)}n^2$.
\end{proof}

\subsection{Proof of Lemma~\ref{lemma_recovery_concentration}}
\label{apx:recovery-concentration}

\begin{proof}
Fix $H\in\mathcal H$, a root $a\in\mathcal R(H)$, and a host vertex $i$.
Let $\Phi_{H,a,i}\colonequals\{\phi:V(H)\hookrightarrow[n]:\phi(a)=i\}$, and for $\phi\in\Phi_{H,a,i}$ define $I_\phi(\zeta)=\mathbf 1\{\zeta\circ\phi\text{ is injective on }V(H)\}$.
The colors on the image of $\phi$ are independent and uniform in $[v]$, so $\Pr_\zeta[I_\phi(\zeta)=1]=\rho_v$.
Linearity of expectation over all triples $(H,a,\phi)$ gives
\[
\E_{\zeta_1,\ldots,\zeta_q}
[\widehat F_{\mathcal H,i}(G)]
=F_{\mathcal H,i}(G).
\]

For one coloring, define
\[
Y_i(G,\zeta)
\colonequals
\sum_{H\in\mathcal H}\sum_{a\in\mathcal R(H)}
\sum_{\phi\in\Phi_{H,a,i}}
I_\phi(\zeta)\chi_{H,\phi}(G).
\]
Then
\[
\widehat F_{\mathcal H,i}(G)
=
\frac{1}{q\rho_v}\sum_{r=1}^qY_i(G,\zeta_r).
\]
Conditioning on $G$ and using independence of the colorings,
\[
\Var_\zeta(\widehat F_{\mathcal H,i}\mid G)
\le
\frac{1}{q\rho_v^2}\,
\E_\zeta[Y_i(G,\zeta)^2\mid G].
\]

For $b\in\{0,1\}$, recall $\mathsf B_{i,b}=\{\mathbf 1_{\{i\in C\}}=b\}$.
We average over $G\sim\PP_{n,k}$ conditional on $\mathsf B_{i,b}$.
For any two rooted embeddings $\phi,\psi$, we have $\E_\zeta[I_\phi(\zeta)I_\psi(\zeta)]\le\rho_v$ and
\[
\E_G[\chi_{H,\phi}(G)\chi_{H',\psi}(G)\mid\mathsf B_{i,b}]\geq 0
\]
because the latter expectation is the conditional probability that the vertices of the symmetric difference lie in $C$.
Hence
\[
\E_{G,\zeta}[Y_i(G,\zeta)^2\mid\mathsf B_{i,b}]
\le
\rho_v\,\E_G[F_{\mathcal H,i}(G)^2\mid\mathsf B_{i,b}].
\]
Since $q\ge\rho_v^{-2}$,
\[
\E_G[\Var_\zeta(\widehat F_{\mathcal H,i}\mid G)\mid\mathsf B_{i,b}]
\le
\frac{1}{q\rho_v}\E_G[F_{\mathcal H,i}(G)^2\mid\mathsf B_{i,b}]
\le
\rho_v\,\E_G[F_{\mathcal H,i}(G)^2\mid\mathsf B_{i,b}].
\]
The estimator is conditionally unbiased, so the left side is
\[
\E_{G,\zeta}
\left[
\left(\widehat F_{\mathcal H,i}(G)-F_{\mathcal H,i}(G)\right)^2
\,\middle|\,\mathsf B_{i,b}
\right].
\]

It remains to normalize.
Lemma~\ref{lem:p-var} gives, uniformly in $i$ and $b$,
\[
\E_G\left[
\left(\frac{F_{\mathcal H,i}}{\mu}-b\right)^2
\,\middle|\,\mathsf B_{i,b}
\right]
=\frac{\Var(F_{\mathcal H,i}\mid\mathsf B_{i,b})}{\mu^2}
\leq \delta_v,
\qquad
\delta_v\colonequals(1+o(1))(c_0\lambda^2)^{-(v-1)}=o(1).
\]
Here we used the conditional means
$\E[F_{\mathcal H,i}\mid\mathsf B_{i,b}]=b\mu$.
It follows that
\[
\frac{\E_G[F_{\mathcal H,i}^2\mid\mathsf B_{i,b}]}{\mu^2}
=b+\frac{\Var(F_{\mathcal H,i}\mid\mathsf B_{i,b})}{\mu^2}
\leq 1+\delta_v.
\]
Dividing the preceding color-coding variance bound by $\mu^2$ therefore gives
\[
\E_{G,\zeta}\left[
\left(
\frac{\widehat F_{\mathcal H,i}-F_{\mathcal H,i}}{\mu}
\right)^2
\,\middle|\,\mathsf B_{i,b}
\right]
\leq \rho_v(1+\delta_v)
=O(\rho_v),
\]
which proves the claimed bound on the error from color coding.

Finally, $\E_\zeta[\widehat F_{\mathcal H,i}-F_{\mathcal H,i}\mid G,C]=0$ by conditional unbiasedness.
The cross term therefore vanishes after conditioning on $\mathsf B_{i,b}$, giving
\begin{align*}
\E_{G,\zeta}\left[
\left(
\frac{\widehat F_{\mathcal H,i}}{\mu}-b
\right)^2
\,\middle|\,\mathsf B_{i,b}
\right]
&=
\E_{G,\zeta}\left[
\left(
\frac{\widehat F_{\mathcal H,i}-F_{\mathcal H,i}}{\mu}
\right)^2
\,\middle|\,\mathsf B_{i,b}
\right]\\
&\quad+
\E_G\left[
\left(
\frac{F_{\mathcal H,i}}{\mu}-b
\right)^2
\,\middle|\,\mathsf B_{i,b}
\right]\\
&\leq O(\rho_v)+\delta_v.
\end{align*}
For $v=\lfloor\sqrt{\log n}\rfloor$, Stirling's formula gives $\rho_v=\exp(-\Theta(v))$, and $c_0\lambda^2>1$ gives $\delta_v=\exp(-\Omega(v))$.
The last display is therefore $\exp(-\Omega(\sqrt{\log n}))$, uniformly in $i$ and $b$.
Averaging over $\mathsf B_{i,0}$ and $\mathsf B_{i,1}$ gives the unconditional bound.
\end{proof}

\subsection{Proof of Lemma \ref{lem:treewidth_algo}}\label{apx:treewidth-alg}

\begin{proof}
Choose a pattern vertex $u_r$, root the given tree decomposition at a bag containing it, and convert it to a nice tree decomposition $\sD=(\sT,\{B_x\}_{x\in V(\sT)})$ of width at most $t$ with $O(v)$ nodes.
Attach a chain of forget nodes above the root ending at $B_r=\{u_r\}$.
This preserves the width and the $O(v)$ size bound.

For each node $x\in V(\sT)$, let $V_x$ be the set of pattern vertices appearing in bags of the subtree rooted at $x$, and let $H_x\colonequals H[V_x]$.
For each injective assignment $\psi:B_x\hookrightarrow [n]$ and color set $C\subseteq [v]$, the algorithm maintains
\[
\DP_x(\psi,C)
=
\sum_{\phi:V(H_x)\hookrightarrow[n]}
\mathbf{1}\!\left\{
\begin{array}{l}
\phi|_{B_x}=\psi,\\
\zeta(\phi(V(H_x)))=C,\\
\zeta\circ\phi \text{ is injective on }V(H_x)
\end{array}
\right\}
\prod_{e\in E(H_x)\setminus E(H[B_x])} X_{\phi(e)}.
\]
The product includes each edge once one of its endpoints has been forgotten.

We prove that the algorithm computes $\DP_x(\psi,C)$ for every $x,\psi,C$.
At the root, $B_r=\{u_r\}$ and $E(H[B_r])=\emptyset$, so this gives
\[
\sum_{z\in[n]} \DP_r(\psi_z,[v])
=
\sum_{\phi:V(H)\hookrightarrow[n]}
\mathbf{1}\{\zeta\circ\phi\text{ is injective on }V(H)\}
\prod_{e\in E(H)} X_{\phi(e)}
=
T_H^{\mathrm{col}}(G,\zeta),
\]
where $\psi_z(u_r)=z$.

We proceed by induction from the leaves to the root.

\medskip \noindent \textbf{Leaf node.} Here $B_x=\{u\}$ and $H_x$ consists only of $u$.
The edge product is empty, so $\DP_x(\psi,C)=1$ if $C=\{\zeta(\psi(u))\}$ and $0$ otherwise, as initialized by the algorithm.

\medskip \noindent \textbf{Introduce node.} Let $y$ be the child of $x$, with $B_x=B_y\cup\{u\}$.
The vertex $u$ does not appear in $H_y$, and $V(H_x)=V(H_y)\cup\{u\}$.
Every edge of $H_x$ incident to $u$ has both endpoints in $B_x$, so none contributes yet.
Each colorful embedding counted by $\DP_x(\psi,C)$ restricts to one of $H_y$ extending $\psi|_{B_y}$ and using the colors $C\setminus\{\zeta(\psi(u))\}$.
Conversely, each such embedding extends uniquely by mapping $u$ to $\psi(u)$.
Therefore
\[
\DP_x(\psi,C)=
\DP_y\bigl(\psi|_{B_y},\, C\setminus\{\zeta(\psi(u))\}\bigr)
\]
whenever $\zeta(\psi(u))\in C$, and $\DP_x(\psi,C)=0$ otherwise, as in the algorithm.

\medskip \noindent \textbf{Forget node.} Let $y$ be the child of $x$, with $B_x=B_y\setminus\{u\}$.
Fix $\psi:B_x\hookrightarrow[n]$ and $C\subseteq[v]$.
For each image $z\in[n]\setminus\psi(B_x)$ of $u$, extend $\psi$ to $\psi':B_y\hookrightarrow[n]$ by setting $\psi'(u)=z$.
By induction, $\DP_y(\psi',C)$ gives the contribution of the corresponding embeddings before $u$ is forgotten.

The edges that now contribute are exactly $\{u,w\}\in E(H)$ with $w\in B_x$: both endpoints were in $B_y$, but only $w$ remains in $B_x$.
All other edges retain their previous status.

Thus we multiply by
\[
m(u,\psi',\psi)
=
\prod_{\substack{w\in B_x\\ \{u,w\}\in E(H)}} X_{\psi'(u)\psi(w)}.
\]
Summing over $z$ gives
\[
\DP_x(\psi,C)
=
\sum_{z\in[n]\setminus \psi(B_x)}
m(u,\psi',\psi)\,\DP_y(\psi',C),
\]
as in the algorithm.

\medskip \noindent \textbf{Join node.} Let $y_1,y_2$ be the children of $x$, with $B_x=B_{y_1}=B_{y_2}$.
Fix $\psi:B_x\hookrightarrow[n]$, and put $C_\psi\colonequals\zeta(\psi(B_x))$.
The graphs $H_{y_1}$ and $H_{y_2}$ intersect exactly in $B_x$.
A colorful embedding of $H_x$ extending $\psi$ is therefore equivalent to a pair of colorful embeddings of these graphs extending $\psi$, with color sets satisfying $C_1\cup C_2=C$ and $C_1\cap C_2=C_\psi$.
The intersection condition ensures that the embeddings share only the images of the bag vertices.
Their signed contributions multiply because their edge sets are disjoint outside the bag, giving
\[
\DP_x(\psi,C)
=
\sum_{\substack{C_1,C_2\subseteq [v]\\
C_1\cup C_2=C\\
C_1\cap C_2=C_\psi}}
\DP_{y_1}(\psi,C_1)\,\DP_{y_2}(\psi,C_2),
\]
as in the algorithm.

This completes the induction and proves correctness.

\medskip \noindent \textbf{Running time.} Each bag has size at most $t+1$, giving at most $n^{t+1}$ assignments and $2^v$ color sets per assignment.
Thus each node has at most $2^v n^{t+1}$ states.

At a leaf or introduce node, each state takes constant time, giving $2^v n^{t+1}$ time per node.

At a forget node, the parent bag has size at most $t$, giving at most $2^v n^t$ states.
For each, we sum over at most $n$ images of the forgotten vertex and compute $m(u,\psi',\psi)$ in $O(t)$ time.
The total time per node is $2^v n^{t+1}$, absorbing the factor $O(t)$ into the bound.

At a join node, for fixed $\psi,C$, there are at most $3^v$ pairs $(C_1,C_2)$ with $C_1\cup C_2=C$ and $C_1\cap C_2=C_\psi$.
Each join node therefore takes $2^{O(v)}n^{t+1}$ time.

There are $O(v)$ nodes, so the total running time is $O(v)2^{O(v)}n^{t+1}=2^{O(v)}n^{t+1}$.
\end{proof}

\subsection{Proof of Lemma~\ref{lem:enumerate-bounded-treewidth-patterns}}
\label{apx:enumerate-bounded-treewidth-patterns}

\begin{proof}
Handle the finitely many cases $v\leq t$ by exhaustive enumeration, and assume $v\geq t+1$.
Every graph of treewidth at most $t$ is a spanning subgraph of a $t$-tree on the same vertex set.
To enumerate labeled $t$-trees, choose the initial $(t+1)$-clique, order the remaining vertices, and choose an existing $t$-clique to which each new vertex is attached.
At each step there are $O_t(v)$ available $t$-cliques, so there are $\exp(O_t(v\log v))$ construction histories.

Each $t$-tree has $O_t(v)$ edges, so enumerating its spanning subgraphs adds a factor $2^{O_t(v)}$.
Discard disconnected subgraphs and remove duplicates by comparing the bit vectors of their edge sets.
For each remaining graph, retain the decomposition of width $t$ given by its construction history: the initial clique is one bag, and each added vertex forms a bag with the $t$-clique to which it was attached.
This produces all graphs in $\mathcal G^{\mathrm{conn}}_{v,\leq t}$ and their decompositions in time and space $\exp(O_t(v\log v))$.

For each graph $H$ and vertex $a\in V(H)$, root the decomposition at a bag containing $a$, convert it to a nice decomposition, and attach a chain of forget nodes ending at $\{a\}$.
This takes $\poly_t(v)$ time per pair $(H,a)$, within the same total bound.
\end{proof}

\subsection{Proof of Lemma \ref{lemma_concentration}}
\label{apx:concentration}
\begin{proof}
Fix $G$.
For each $H\in\sH$, let $\Phi_H\colonequals\{\phi:V(H)\hookrightarrow[n]\}$, and for $\phi\in\Phi_H$ define $I_\phi(\zeta)\colonequals\mathbf 1\{\zeta\circ\phi\text{ is injective on }V(H)\}$.
The colors on the image of $\phi$ are independent and uniform in $[v]$, so $\E_\zeta[I_\phi(\zeta)\mid G]=v!/v^v=\rho_v$.
By linearity of expectation,
\[
\E_\zeta[T_H^{\mathrm{col}}(G,\zeta)\mid G] = \sum_{\phi\in \Phi_H}\E_\zeta[I_\phi(\zeta)\mid G]\chi_\phi(G) = \rho_v\,T_H(G).
\]
Define $F_{\sH}^{\mathrm{col}}(G,\zeta)\colonequals\sum_{H\in\sH}T_H^{\mathrm{col}}(G,\zeta)$.
Summing over $H$ gives $\E_\zeta[F_{\sH}^{\mathrm{col}}(G,\zeta)\mid G]=\rho_v F_{\sH}(G)$.
For independent colorings $\zeta_1,\dots,\zeta_q$ as in the statement, put $Y_s(G)\colonequals F_{\sH}^{\mathrm{col}}(G,\zeta_s)$.
Then
\[
\widehat F_{\sH}(G)=\frac{1}{q\rho_v}\sum_{s=1}^q Y_s(G),
\]
so $\E_{\zeta_1,\dots,\zeta_q}[\widehat F_{\sH}(G)\mid G]=F_{\sH}(G)$.

Conditional on $G$, the variables $Y_s(G)$ are independent and identically distributed, so
\[
\Var_{\zeta_1,\dots,\zeta_q}(\widehat F_{\sH}(G)\mid G)
=
\frac{1}{\rho_v^2}\cdot \frac{1}{q}\,\Var_\zeta(Y_1(G)\mid G).
\]
To bound this variance, let $\DD$ be either $\PP$ or $\QQ$.
Expanding $Y_1(G)=\sum_{H\in\sH}\sum_{\phi\in\Phi_H}I_\phi(\zeta)\chi_\phi(G)$ and using independence of $G$ and $\zeta$ gives
\[
\E_{\DD,\zeta}[Y_1^2] = \sum_{H,H'\in \sH}\sum_{\phi\in \Phi_H}\sum_{\psi\in \Phi_{H'}}\E_\zeta[I_\phi(\zeta)I_\psi(\zeta)]\,\E_{\DD}[\chi_\phi\chi_\psi].
\]
For every pair of embeddings, $\E_\zeta[I_\phi I_\psi]\le\E_\zeta[I_\phi]=\rho_v$ and $\E_{\DD}[\chi_\phi\chi_\psi]\ge0$.
Therefore
\[
\E_{\DD,\zeta}[Y_1^2]
\le
\rho_v
\sum_{H,H'\in \sH}\sum_{\phi\in \Phi_H}\sum_{\psi\in \Phi_{H'}} \E_{\DD}[\chi_\phi\chi_\psi] = \rho_v\,\E_{\DD}[F_{\sH}(G)^2].
\]
Since variance is bounded by the second moment and $q=\lceil\rho_v^{-2}\rceil$, conditional unbiasedness gives
\begin{align*}
\E_{\DD,\zeta}\!\left[(\widehat F_{\sH}-F_{\sH})^2\right]
&=\E_{G\sim\DD}\left[\Var_{\zeta_1,\dots,\zeta_q}(\widehat F_{\sH}(G)\mid G)\right] \\
&\le \frac{1}{q\rho_v}\E_{\DD}[F_{\sH}(G)^2]
\le \rho_v\E_{\DD}[F_{\sH}(G)^2].
\end{align*}
Since $v\to\infty$, we have $\rho_v=\exp(-\Theta(v))=o(1)$.
Theorem~\ref{thm:general-separation}, with $v\leq\log n/(5\log\log n)$, also gives
\begin{align*}
\frac{\E_{\QQ}[F_{\sH}(G)^2]}{\E_{\PP}[F_{\sH}(G)]^2}&=o(1),\\
\frac{\E_{\PP}[F_{\sH}(G)^2]}{\E_{\PP}[F_{\sH}(G)]^2}&=1+o(1).
\end{align*}
Therefore, for both $\DD=\PP$ and $\DD=\QQ$,
\[
\frac{\E_{\DD,\zeta}\!\left[(\widehat F_{\sH}-F_{\sH})^2\right]}
{\E_{\PP}[F_{\sH}(G)]^2}
\to 0. \qedhere
\]
\end{proof}
\section{Omitted proofs from Section \ref{sec:mm-acceleration}}
\subsection{Proof of Lemma~\ref{lem:one-partial-2tree-mm}}\label{app:proof-one-partial-2tree-mm}

\begin{proof}
Fix a rooted construction of the $2$-tree completion $K$.
For each oriented boundary edge $e=(p,q)$, let $V_e$ consist of $p,q$ and all vertices introduced below $e$.
For each color set $S\subseteq[v]$, our dynamic program computes an $n\times n$ matrix $D_e^S$ satisfying
\[
D_e^S(a,b)
=
\sum_{\substack{
\phi:V_e\hookrightarrow[n]\\
\phi(p)=a,\ \phi(q)=b\\
\zeta\circ\phi\text{ injective}\\
\zeta(\phi(V_e))=S
}}
\prod_{\{x,y\}\in E(H[V_e])} X_{\phi(x)\phi(y)}.
\]
We prove this invariant by induction from the leaves toward the root.

For $e=(p,q)$ and $a,b\in[n]$, let $L_e(a,b)=X_{ab}$ if $\{p,q\}\in E(H)$ and $L_e(a,b)=1$ otherwise.
List the vertices introduced directly on $e$ as $w_1,\ldots,w_m$, in the order of construction, and initialize
\[
D_{e,0}^S(a,b)
=
\mathbf 1\{a\ne b,\ \zeta(a)\ne\zeta(b),\ S=\{\zeta(a),\zeta(b)\}\}L_e(a,b).
\]
If no vertex is introduced below $e$, then $V_e=\{p,q\}$, and the invariant holds.

For a vertex $w=w_j$ introduced directly on $e=(p,q)$, write $e_L=(p,w)$ and $e_R=(w,q)$ for its child edges.
For $\gamma\in[v]$, let $P_\gamma$ be the diagonal matrix with $(P_\gamma)_{zz}=\mathbf 1\{\zeta(z)=\gamma\}$.
Using the previously computed child tables, define
\[
B_w^S
=
\sum_{\substack{
S_L,S_R,\gamma\\
S_L\cup S_R=S\\
S_L\cap S_R=\{\gamma\}
}}
D_{e_L}^{S_L}P_\gamma D_{e_R}^{S_R}.
\]
By induction, $D_{e_L}^{S_L}(a,z)$ and $D_{e_R}^{S_R}(z,b)$ sum over the two child pieces with $w$ mapped to $z$.
Thus
\[
(D_{e_L}^{S_L}P_\gamma D_{e_R}^{S_R})(a,b)
=\sum_{z\in[n]}D_{e_L}^{S_L}(a,z)\mathbf 1\{\zeta(z)=\gamma\}D_{e_R}^{S_R}(z,b).
\]
The conditions $S_L\cup S_R=S$ and $S_L\cap S_R=\{\gamma\}$ ensure that the child embeddings overlap only at $w$ and use distinct colors elsewhere.
Hence $B_w^S(a,b)$ sums the colorful signed contributions of the whole branch below $w$, with $p\mapsto a$ and $q\mapsto b$.

After computing $B_{w_j}^S$, we update the accumulated table by
\[
D_{e,j}^S(a,b)
=
\sum_{\substack{S_0,S_1\subseteq[v]\\
S_0\cup S_1=S\\
S_0\cap S_1=\{\zeta(a),\zeta(b)\}}}
D_{e,j-1}^{S_0}(a,b)B_{w_j}^{S_1}(a,b),
\]
where the right-hand side is zero when $a=b$ or $\zeta(a)=\zeta(b)$.
Distinct branches below $e$ intersect only at $p$ and $q$, so their color sets must intersect exactly at $\{\zeta(a),\zeta(b)\}$.
The entrywise product multiplies their signed contributions, and the sum records the union of their colors.
The initial table $D_{e,0}$ contributes the factor for $\{p,q\}$ exactly once if this edge belongs to $H$.
Thus every edge of $H[V_e]$ is counted exactly once, and no completion edge contributes a factor.

After processing all branches directly below $e$, set $D_e^S=D_{e,m}^S$.
This proves the invariant.

For each $w$, computing $B_w^S$ for all $S$ requires $2^{O(v)}$ products of $n\times n$ matrices.
The entrywise products and sums over color sets cost $2^{O(v)}n^2$, which is dominated by the matrix products.
There are $O(v)$ boundary edges and branch vertices, so the total running time is $2^{O(v)}n^{\omega+o(1)}$.

Let $e_0=(p,q)$ be the root edge.
The invariant gives
\[
T_H^{\mathrm{col}}(X,\zeta)
=
\sum_{a,b\in[n]}D_{e_0}^{[v]}(a,b).
\]
If the desired root vertex $r$ lies on the root edge, say $e_0=(r,r')$, then fixing the image of $r$ to be $i$ and summing over the image of $r'$ gives
\[
T^{\mathrm{col}}_{H,r}(X,\zeta;i)
=
\sum_{j\in[n]}D_{e_0}^{[v]}(i,j).
\]
If $r$ is not on the root edge, choose any edge of $K$ containing $r$ as the initial edge.
A $2$-tree admits a construction from any prescribed edge: repeatedly remove a simplicial vertex of degree $2$ outside that edge, then reverse the elimination order.
Finding and storing these orders for the $O(v)$ choices needed here takes polynomial time in $v$.
Thus the same runtime bound holds for the full vector at any prescribed root $r$.
\end{proof}

\subsection{Proof of Proposition~\ref{prop:randomized-seed-boosting}}
\label{app:proof-randomized-seed-boosting}

\begin{proof}
For $a=0$ there is nothing to prove.
Fix $a\ge1$ and $\lambda>2^{-a/2}\lambda_0$, put $\varepsilon=(2^{a/2}\lambda-\lambda_0)/2>0$, and assume $k\ge\lambda\sqrt n$.
Independently sample $R=A n^{a/2}\log n$ uniform random $a$-sets with replacement, where $A=A(a,\lambda_0,\varepsilon)$ is sufficiently large.
A sampled seed lies in $C$ with probability $p=(k)_a/(n)_a\ge\Omega_{a,\lambda}(n^{-a/2})$.
Let $Z$ count these occurrences and put $\mu_Z=Rp$.
Conditional on $C$, we have $Z\sim\operatorname{Bin}(R,p)$ and $\mu_Z=\Omega(\log n)$, so $Z\ge\mu_Z/2$ with probability $1-o(1)$.
For each sampled $S$ that is a clique, form
\[
G_S=G\!\left[\bigcap_{u\in S}N(u)\setminus S\right].
\]
If $S\subseteq C$, put $K=k-a$ and $q=2^{-a}$.
Then
\[
|V(G_S)|=(1+o(1))\bigl(qn+(1-q)K\bigr),
\]
and $G_S$ contains the planted clique $C\setminus S$.
The function $x\mapsto x/\sqrt{qn+(1-q)x}$ is increasing, so the effective signal is minimized at the smallest allowed $k$ and is at least
\[
(1-o(1))\frac{\lambda\sqrt n}{\sqrt{qn+(1-q)\lambda\sqrt n}}
=(1-o(1))2^{a/2}\lambda.
\]
For sufficiently large $n$, this signal is at least $\lambda_0+\varepsilon$.
For any fixed sampled occurrence, conditional on its seed $S\subseteq C$ and vertex set $W=V(G_S)$, the clique $C\setminus S$ is uniform among the $K$-subsets of $W$.
Conditional on this clique, $G_S$ has the ordinary planted clique law on $W$.
Indeed, the edges from $S$ that determine $W$ are independent of the edges inside $W$, and their likelihood is the same for every $K$-subset of $W$.
Using fresh randomness for each base call, the probability of failure for any fixed occurrence of a seed in $C$ is therefore $o(1)$, including the exceptional probability in the size estimate above.

Let $Z_{\mathrm{bad}}$ count the occurrences of seeds in $C$ on which the base call fails.
By linearity of expectation, $\E[Z_{\mathrm{bad}}]=o(Rp)=o(\mu_Z)$.
Consequently, Markov's inequality and the concentration of $Z$ give
\[
\Pr[Z_{\mathrm{bad}}\ge Z]
\le
\Pr[Z<\mu_Z/2]+\Pr[Z_{\mathrm{bad}}\ge\mu_Z/2]
=o(1).
\]
Thus, with probability $1-o(1)$, at least one call recovers $C\setminus S$.
For each output $R_S$, we check whether $R_S\cup S$ is a clique of size $k$.
A successful call with $S\subseteq C$ produces such a clique.
Uniformly for $k\ge\lambda\sqrt n$, with high probability no candidate passes under $\QQ_n$, which contains no $k$-clique, and every candidate that passes under the planted law equals the unique planted $k$-clique $C$.
Finding the common neighbors of $S$ takes $O(n)$ time per seed, and checking a candidate takes $O(k^2)\le O(n^2)$ time.
Since $r\ge2$, the base calls dominate these costs, giving total runtime $n^{r+a/2+o(1)}$.
\end{proof}

\subsection{Proof of Lemma~\ref{lem:good-seeded-row}}
\label{app:proof-good-seeded-row}

\begin{proof}
Expose the edges between $Q$ and $V_1\setminus Q$, which determine $W=W_Q(V_1)$.
Conditional on $Q\subseteq C\cap V_1$ and any possible values of $W$ and $C_W=(C\cap V_1)\setminus Q$, the graph $G[W]$ has the planted clique law with clique $C_W$: all edges inside $C_W$ are present, and all other edges inside $W$ are independent fair coin flips.

For $x\in V_1$, the mask $w_Q(x)=\mathbf 1\{x\notin Q\}\prod_{q\in Q}\frac{1+X_{qx}}2$ is the indicator of $x\in W$.
Thus only embeddings with image in $W$ contribute, each with its ordinary signed tree character on $G[W]$.
The seeded statistic rooted at $i$ therefore equals the ordinary rooted tree statistic on $G[W]$ for $i\in W$ and is zero otherwise.

The hypotheses $|W|\ge n^\alpha$ and $|C_W|\ge(e^{-1/2}+\eta)\sqrt{|W|}$ ensure that $v=o(\log |W|/\log\log |W|)$ and $v^2=o(|C_W|)$, so the rooted moment and color-coding estimates apply with ambient size $|W|$, and the normalizing factor is positive for all sufficiently large $n$.
Since $|C_W|\ge (\exp(-1/2)+\eta)\sqrt{|W|}$, the rooted tree moment estimates give, after the usual positive normalization,
\[
\E\left[
\left(
p_{Q,i}-\mathbf 1_{\{i\in C_W\}}
\right)^2
\mid W,C_W
\right]
\le \exp(-\Omega(v))
\]
uniformly over $i\in W$.
The color-coding estimator is that of Lemma~\ref{lemma_recovery_concentration}, with $\mathcal H=\mathcal T_v$ and ambient vertex set $W$.
Combining its variance bound and unbiasedness with the bound for the exact statistic gives
\[
\E\left[
\left(
\widehat p_{Q,i}-\mathbf 1_{\{i\in C_W\}}
\right)^2
\mid W,C_W
\right]
\le \delta_v,
\qquad
\delta_v=\exp(-\Omega(v)).
\]

We now justify selecting the $M_{N,k}$ largest scores, where $N=|V_1|$.
Let $V_{\mathrm{thr}}(Q)=\{i\in W:\widehat p_{Q,i}> 1/2\}$.
By Markov's inequality, as used in Lemma~\ref{lem:rounding}, with conditional probability $1-o(1)$,
\[
|V_{\mathrm{thr}}(Q)|\le |C_W|+O(|W|\delta_v^{1/2})
\le k+o(N/(\log N)^3),
\]
and
\[
|V_{\mathrm{thr}}(Q)\cap C_W|\ge (1-o(1))|C_W|.
\]
Normalization multiplies all scores by the same positive factor and preserves their order.
If $M_{N,k}=N$, the selected set contains $V_{\mathrm{thr}}(Q)$.
Otherwise, the preceding bound gives $|V_{\mathrm{thr}}(Q)|<M_{N,k}$, so the same containment holds.
Thus the selected set contains all but $o(k)$ vertices of $C_W$.
Finally,
\[
\frac{k^2}{M_{N,k}}
\ge
\Omega\!\left(\min\left\{k,\frac{k^2(\log N)^3}{N}\right\}\right)
=\Omega((\log n)^3),
\]
as needed for the union bounds in the cleanup between $V_1$ and $V_2$.
\end{proof}

\subsection{Proof of Lemma~\ref{lem:batched-seeded-tree-dp}}
\label{app:proof-batched-seeded-tree-dp}

\begin{proof}
For each rooted subtree $U$ of $T$, color set $A\subseteq[v]$, seed tuple $Q$, and root image $x\in V_1$, let $D_U^A(Q,x)$ sum the signed contributions of colorful embeddings of $U$ with color set $A$ and root image $x$, including the mask $w_Q(z)=\mathbf 1\{z\notin Q\}\prod_{q\in Q}\frac{1+X_{qz}}2$ at each image vertex $z$.
For a leaf $U$, $D_U^{\{\zeta(x)\}}(Q,x)=w_Q(x)$ and all other color states are zero.
All $R|V_1|$ mask values can be computed in $O_s(Rn)=n^{s/2+1+o(1)}$ time.

Form $U$ by attaching a child subtree $W$ to the root of $U_0$.
Writing $x$ and $y$ for the images of their respective roots gives
\[
D_U^A(Q,x)
=
\sum_{\substack{A_0,B\subseteq[v]\\A_0\cap B=\emptyset\\A_0\cup B=A}}
D_{U_0}^{A_0}(Q,x)
\left(\sum_{y\in V_1}X_{xy}D_W^B(Q,y)\right).
\]
The mask at $x$ occurs in $D_{U_0}^{A_0}$, and the disjointness of $A_0$ and $B$ ensures injectivity between the two pieces.
For fixed $W$ and $B$, the table $D_W^B$ is an $R\times |V_1|$ matrix.
The product $D_W^B X_{V_1,V_1}^{\top}$ computes the parenthesized sums for all $Q$ and $x$.
After padding, its shape is $(n^{s/2+o(1)}\times n)\cdot(n\times n)$, so it costs $n^{\omega(1,1,s/2)+o(1)}$ by transpose symmetry.
Multiplying entrywise by $D_{U_0}^{A_0}$ and summing over disjoint color sets costs $2^{O(v)}Rn$, no more than writing the products.
There are $2^{O(v)}$ rooted subpieces and choices of color sets, giving total runtime $2^{O(v)}n^{\omega(1,1,s/2)+o(1)}$ for each rooted tree and coloring.
The final table gives the score matrix $\bigl(T^{\mathrm{col}}_{Q,T,r}(G,\zeta;i)\bigr)_{Q,i}$.
\end{proof}

\subsection{Proof of Lemma~\ref{lem:batched-cleanup}}
\label{app:proof-batched-cleanup}

\begin{proof}
For each seed tuple $Q$, form $V_3(Q)\subseteq V_1$ by selecting the $M_{N,k}$ largest scores.
These scores use only edges inside $V_1$ and edges from the seeds to $V_1$.
Thus, conditional on $C$, the sets $V_3(Q)$ are independent of the edges between $V_1$ and $V_2$ outside the clique.
The edges between $C\cap V_1$ and $C\cap V_2$ are deterministic.

Let $B$ be the $R\times |V_1|$ indicator matrix of the sets $V_3(Q)$.
The product $B A_{V_1,V_2}$ counts the neighbors of every $v\in V_2$ in each $V_3(Q)$.
After padding, its shape is $(n^{s/2+o(1)}\times n)\cdot(n\times n)$, so it costs $n^{\omega(1,1,s/2)+o(1)}$.
Define
\[
C_2(Q)
\colonequals
\left\{
v\in V_2:
|E(\{v\},V_3(Q))|
\ge
\frac{|V_3(Q)|}{2}+0.2499(1-n^{-2/5})k
\right\}.
\]
For the seed $Q^\star$, the set $V_3(Q^\star)$ contains a $1-o(1)$ fraction of $C\cap V_1$ and has size $M_{N,k}$.
Hence
\[
\frac{|V_3(Q^\star)\cap (C\cap V_1)|^2}{|V_3(Q^\star)|}
\ge
\Omega\!\left(\frac{k^2}{M_{N,k}}\right)
=
\Omega((\log n)^3).
\]
Hoeffding's inequality and a union bound, as in Lemma~\ref{lem:implementable-threshold}, give $C_2(Q^\star)=C\cap V_2$ with probability $1-o(1)$.

Let $B_2$ be the $R\times |V_2|$ indicator matrix of the sets $C_2(Q)$.
The product $A_{V_1,V_2}B_2^\top$ counts the neighbors of every $u\in V_1$ in each $C_2(Q)$, with the same runtime bound.
Thresholding at $|C_2(Q)|$ gives
\[
C_1(Q)=\{u\in V_1:\ |E(\{u\},C_2(Q))|=|C_2(Q)|\}.
\]
Since $C_2(Q^\star)=C\cap V_2$, Lemma~\ref{lem:recover-main-from-holdout} gives $C_1(Q^\star)=C\cap V_1$ with probability $1-o(1)$.
Thus this seed outputs $C$.

Let $D$ be the $R\times n$ indicator matrix of the candidates $C_1(Q)\sqcup C_2(Q)$.
The product $DA$ counts each vertex's neighbors in each candidate.
We accept a row exactly when its candidate has size $k$ and every selected coordinate of $DA$ equals $k-1$.
This product has the same shape and runtime bound as the preceding two products, and forming $D$ and scanning the output costs $O(Rn)$.
Thus checking all candidates takes $n^{\omega(1,1,s/2)+o(1)}$ time for every $k\ge\lambda\sqrt n$, proving the runtime and correctness claims.
\end{proof}
\section{Omitted Proofs from Section~\ref{sec:derandomization}}
\label{app:sec6-proofs}

\subsection{Proof of Lemma~\ref{lem:balanced-hash-detection}}
\label{app:proof-balanced-hash-detection}

\begin{proof}
Write
\[
\widetilde F_{\sH}(G)
=
\sum_{H\in\sH}\sum_{\phi\in\Inj(V(H),[n])}
c_{\phi,H}\chi_\phi(G).
\]
If $S=\phi(V(H))$, then $c_{\phi,H}=N_{\mathcal T}(S)/(\beta_v|\mathcal T|)$, so $1/\delta\le c_{\phi,H}\le\delta$.
This gives the expectation bounds.
Under $\QQ$, distinct edge characters are orthogonal.
After grouping identical characters, all coefficients are nonnegative and change by at most a factor $\delta$, giving the null variance bound.

For the planted variance, write $c_{H,\phi}=c_{\phi,H}$ and $\chi_{H,\phi}=\chi_\phi$.
The mean satisfies
\[
\E_{\PP}\widetilde F_{\mathcal H}
=
\sum_{H,\phi}c_{H,\phi}\alpha_v
\ge
\delta^{-1}\E_{\PP}F_{\mathcal H}.
\]
The variance is
\[
\Var_{\PP}(\widetilde F_{\mathcal H})
=
\sum_{H,H',\phi,\psi}
c_{H,\phi}c_{H',\psi}
\bigl(\alpha_{\sigma(H,H';\phi,\psi)}-\alpha_v^2\bigr).
\]
In the corresponding expansion for $F_{\mathcal H}$, let $P$ sum the positive terms and $N$ sum the absolute values of the negative terms.
Since $\sigma(H,H';\phi,\psi)\le 2v$ and $\alpha_t$ decreases in $t$,
\[
\frac{N}{(\E_{\PP}F_{\mathcal H})^2}
\le
1-\frac{\alpha_{2v}}{\alpha_v^2}.
\]
Since $v^2=o(k)$, we have $2v\le k$ for sufficiently large $n$, and
\[
\frac{\alpha_{2v}}{\alpha_v^2}
=
\prod_{j=0}^{v-1}(1-x_j),
\qquad
x_j=\frac{v(n-k)}{(k-j)(n-v-j)}.
\]
Here $0\le x_j\le 2v/k$, so
\[
1-\frac{\alpha_{2v}}{\alpha_v^2}
\le
\sum_{j=0}^{v-1}x_j
\le
\frac{2v^2}{k}
=o(1).
\]
Since $\Var_{\PP}(F_{\mathcal H})=P-N$, the assumed normalized variance bound gives $P/(\E_{\PP}F_{\mathcal H})^2=o(1)$.
Dropping the negative terms and using $c_{H,\phi}c_{H',\psi}\le\delta^2$ gives
\[
\frac{\Var_{\PP}(\widetilde F_{\mathcal H})}
{(\E_{\PP}\widetilde F_{\mathcal H})^2}
\le
\delta^4\frac{P}{(\E_{\PP}F_{\mathcal H})^2}
=o(1).
\]
Here $\delta^4\le16$ by the hypothesis $\delta\le2$.
The same comparison applies to a fixed clique $C\subseteq[n]$ of size $k$.
For an embedded pattern $a=(H,\phi)$, write $\chi_a=\chi_{H,\phi}$ and let $R_C(a)$ be its set of edges not contained in $C$.
Then
\[
\E_{\PP_{n,C}}[\chi_a]
=
\mathbf 1\{R_C(a)=\emptyset\},
\qquad
\E_{\PP_{n,C}}[\chi_a\chi_b]
=
\mathbf 1\{R_C(a)=R_C(b)\}.
\]
Thus every covariance equals $\mathbf 1\{R_C(a)=R_C(b)\ne\emptyset\}\ge0$.
Using $\delta^{-1}\le c_{H,\phi}\le\delta$ gives
\begin{align*}
\delta^{-1}\E_{\PP_{n,C}}[F_{\mathcal H}]
&\le
\E_{\PP_{n,C}}[\widetilde F_{\mathcal H}]
\le
\delta\E_{\PP_{n,C}}[F_{\mathcal H}],\\
\Var_{\PP_{n,C}}(\widetilde F_{\mathcal H})
&\le
\delta^2\Var_{\PP_{n,C}}(F_{\mathcal H}).
\end{align*}
By invariance under relabeling, $F_{\mathcal H}$ has the same distribution for every fixed $C$ as for uniformly random $C$.
Its normalized variance is therefore $o(1)$ uniformly in $C$, and these inequalities prove the claim for $\widetilde F_{\mathcal H}$.
\end{proof}

\subsection{Proof of Lemma~\ref{lem:balanced-hash-rooted}}
\label{app:proof-balanced-hash-rooted}

\begin{proof}
Put $\gamma_v=1+\eta_v$ and fix $C\subseteq[n]$ of size $k$.
For a rooted pattern $a=(H,r,\phi)$, let $R_C(a)$ be its set of edges not contained in $C$.
As in the preceding proof,
\[
\Cov_{\PP_{n,C}}(\chi_a,\chi_b)
=
\mathbf 1\{R_C(a)=R_C(b)\ne\emptyset\}\ge0.
\]
Since every coefficient lies in $[\gamma_v^{-1},\gamma_v]$, positivity gives
\[
\Var_{\PP_{n,C}}(\widetilde F_{\mathcal H,i})
\le
\gamma_v^2\Var_{\PP_{n,C}}(F_{\mathcal H,i}).
\]
By relabeling, the mean and variance of $F_{\mathcal H,i}$ depend on $(C,i)$ only through whether $i\in C$.
Thus Lemmas~\ref{lem:rooted-moments} and~\ref{lem:p-var} give, uniformly in $C$ and $i$,
\[
\frac{\Var_{\PP_{n,C}}(F_{\mathcal H,i})}{\mu^2}\le\delta_v,
\qquad
\E_{\PP_{n,C}}F_{\mathcal H,i}
=
\begin{cases}
\mu,&i\in C,\\
0,&i\notin C.
\end{cases}
\]
If $i\in C$, only embeddings contained in $C$ have nonzero mean, so $\E_{\PP_{n,C}}[\widetilde F_{\mathcal H,i}]\in[\gamma_v^{-1}\mu,\gamma_v\mu]$.
It follows that
\[
\E_{\PP_{n,C}}\left[
\left(\frac{\widetilde F_{\mathcal H,i}}{\mu}-1\right)^2
\right]
\le \gamma_v^2\delta_v+O(\eta_v^2),
\qquad i\in C.
\]
If $i\notin C$, every rooted character has an edge outside $C$, so its mean is zero, and
\[
\E_{\PP_{n,C}}\left[
\left(\frac{\widetilde F_{\mathcal H,i}}{\mu}\right)^2
\right]
\le \gamma_v^2\delta_v,
\qquad i\notin C.
\]
Since $\eta_v=o(\delta_v^{1/2})$ and $\gamma_v=1+o(1)$, both bounds are $(1+o(1))\delta_v$, uniformly in $(C,i)$.
\end{proof}

\subsection{Proof of Lemma~\ref{lem:universal-holdout-wrapper}}
\label{app:proof-universal-holdout-wrapper}

\begin{proof}
Let $\mathfrak P$ be the splitter family from Lemma~\ref{lem:deterministic-holdout-splitter}, of size $|\mathfrak P|=O(n\log n/k)\le n^{1/2+o(1)}$.
For every fixed clique $C\subseteq[n]$ of size $k\ge(\lambda+\varepsilon)\sqrt n$, some $B\in\mathfrak P$ gives $V_2=B$ and $V_1=[n]\setminus B$ satisfying
\[
|C\cap V_2|\gg\log n,
\qquad
|C\cap V_1|=k-o(k),
\qquad
|V_1|\ge n/3.
\]
For every $B\in\mathfrak P$, put $N=|V_1|$, run the assumed deterministic score algorithm on $G[V_1]$, set
\[
M_{N,k}=\min\left\{N,\left\lceil k+\frac{N}{(\log N)^3}\right\rceil\right\},
\]
and let $V_3(B)$ consist of the vertices with the $M_{N,k}$ largest raw scores $S_i$, breaking ties by a fixed deterministic rule.
For a set $B$ with the splitter properties, $|C\cap V_1|=k-o(k)$ and $N\le n$, so $|C\cap V_1|\ge(\lambda+\varepsilon/2)\sqrt N$ for sufficiently large $n$.
Thus the scoring algorithm applies to $G[V_1]$.
For analysis, divide by the unknown common positive factor $\mu_{V_1,C\cap V_1}$ and set $V_{\mathrm{thr}}(B)=\{i\in V_1:\widehat p_i>1/2\}$.
Summing the mean square bounds over $C\cap V_1$ and its complement and applying Markov's inequality gives the following with probability $1-o(1)$: $V_{\mathrm{thr}}(B)$ misses at most $O(k\delta_N^{1/2})=o(k)$ clique vertices and contains at most $O(N\delta_N^{1/2})=o(N/(\log N)^3)$ vertices outside $C$.
If $M_{N,k}=N$, then $V_3(B)=V_1$.
Otherwise, $|V_{\mathrm{thr}}(B)|<M_{N,k}$, and normalization preserves the order of the scores, so $V_3(B)$ contains $V_{\mathrm{thr}}(B)$.
In either case, $V_3(B)$ contains all but $o(k)$ vertices of $C\cap V_1$ and has size $M_{N,k}\le k+N/(\log N)^3+1$.

For each row $B$, define
\[
\widehat C_2(B)
=
\left\{
y\in B:
|E(\{y\},V_3(B))|
\ge
\frac{|V_3(B)|}{2}+0.2499\,k
\right\}.
\]
The scores use only edges inside $V_1$, which are independent of the edges between $V_1$ and $V_2$ conditional on $C$.
Also, $|C\cap V_2|\gg\log n$, $|V_3(B)\cap C|=(1-o(1))k$, and
\[
\frac{k^2}{|V_3(B)|}
\ge
\Omega\!\left(\min\left\{k,\frac{k^2(\log N)^3}{N}\right\}\right)
=\Omega((\log n)^3).
\]
Hoeffding's inequality and a union bound, as in Lemma~\ref{lem:implementable-threshold}, give $\widehat C_2(B)=C\cap V_2$ with probability $1-o(1)$ for this $B$.
Then recover
\[
\widehat C_1(B)
=
\{x\in V_1:\ x\text{ is adjacent to every vertex of }\widehat C_2(B)\}.
\]
As in Lemma~\ref{lem:final-adjacency-cleanup}, $\widehat C_1(B)=C\cap V_1$ with probability $1-o(1)$.
Check each candidate $\widehat C_1(B)\cup\widehat C_2(B)$ and output the first that is a clique of size $k$.
The first moment argument in Section~\ref{sec:mm-rectangular} shows, uniformly for $k\ge(\lambda+\varepsilon)\sqrt n$, that with probability $1-o(1)$ no $k$-clique exists under the null law and $C$ is the unique $k$-clique under the planted law.
Thus any accepted candidate equals $C$, and the candidate for this $B$ is accepted with probability $1-o(1)$.

Computing scores for all $B\in\mathfrak P$ costs at most $n^{1/2+o(1)}T(n)$.
Direct cleanup and checking the candidates cost $O(|\mathfrak P|n^2)=O(n^{5/2+o(1)})$, which is dominated by this bound in the relevant applications.
\end{proof}

\subsection{Proof of Theorem~\ref{thm:universal-deterministic-exact-recovery}}
\label{app:proof-universal-deterministic-upper-bounds}

\begin{proof}
We first prove the bound using seeded trees.
Fix an integer $s\ge0$, put $v=\lfloor\sqrt{\log n}\rfloor$, and define
\[
\kappa_s=
\begin{cases}
0, & s=0,\\
(s+1)/2, & s\ge1,
\end{cases}
\qquad
\bar\kappa_s=\kappa_s+\frac12.
\]
Fix $\lambda>(2^s e)^{-1/2}$ and a clique $C\subseteq[n]$ of size $k\ge\lambda\sqrt n$.
Construct the holdout splitter family $\mathfrak P$ from Lemma~\ref{lem:deterministic-holdout-splitter}.
For each $B\in\mathfrak P$, set $V_2=B$ and $V_1=[n]\setminus B$.
For the set $B^\star$ given by Lemma~\ref{lem:deterministic-holdout-splitter}, we have
\[
|C\cap V_2|\ge h\gg\log n,\qquad
|C\cap V_1|=k-o(k),\qquad
|V_1|\le n.
\]

For each $B\in\mathfrak P$, use the Tur\'an construction inside $V_1$ to form a family of seeds of size $s$ that works for every clique location.
For $s=0$, use the empty seed.
For $s=1$, use all singletons in $V_1$.
For $s\ge2$, apply Lemma~\ref{lem:explicit-universal-seed-family}: partition $V_1$ into $B_{\mathrm{num}}=\lfloor k/(4s)\rfloor$ consecutive blocks of size $\Theta(n/k)$ and include every $s$-subset within a block.
For $B^\star$, we have $|C\cap V_1|=k-o(k)\ge k/2$ for sufficiently large $n$, so each block contains at least $2s-o(1)$ clique vertices on average.
Some block therefore contains at least $s$ clique vertices.
In all cases, the family contains a seed $Q^\star\subseteq C\cap V_1$.

For fixed $s$, the number of seeds for each $B$ is
\[
B_{\mathrm{num}}\binom{O_s(n/k)}s
=O_s(n^s/k^{s-1})
\le n^{(s+1)/2+o(1)}
\]
for $s\ge2$, and at most $n$ for $s=1$.
Thus there are at most $n^{\kappa_s+o(1)}$ seeds for each $B$.
Since $|\mathfrak P|\le n^{1/2+o(1)}$, there are at most $n^{\bar\kappa_s+o(1)}=n^{\kappa_s+1/2+o(1)}$ pairs $(B,Q)$.
We pad the matrices below to this number of rows.

For each pair $(B,Q)$, compute the rooted scores using only the induced subgraph on $V_1=[n]\setminus B$, with vertex weights
\[
w_{B,Q}(x)
=
\mathbf 1\{x\in V_1\setminus Q\}
\prod_{q\in Q}\frac{1+X_{qx}}2.
\]
For each rooted tree piece $U$, color set $A$, pair $(B,Q)$, and root image $x$, let $D_U^A(B,Q,x)$ sum the signed contributions of colorful embeddings of $U$, weighted by $w_{B,Q}$ at each image vertex.
For fixed $U,A$, this is an $n^{\bar\kappa_s+o(1)}\times n$ matrix.
Attaching a child subtree $W$ to the root of $U_0$ gives the recurrence
\[
D_U^A(B,Q,x)
=
\sum_{\substack{A_0,A_1\subseteq[v]\\A_0\cap A_1=\emptyset\\A_0\cup A_1=A}}
D_{U_0}^{A_0}(B,Q,x)
\left(\sum_y X_{xy}D_W^{A_1}(B,Q,y)\right).
\]
The mask at $x$ occurs in $D_{U_0}^{A_0}$.
For fixed $W,A_1$, a matrix product of shape $(n^{\bar\kappa_s}\times n)\cdot(n\times n)$ computes the parenthesized sums for all pairs and all $x$.
The entrywise products and sums over color sets cost $2^{O(v)}$ times the output size and do not change the exponent.
For $v=\lfloor\sqrt{\log n}\rfloor$, the tree patterns, roots, color sets, and colorings from balanced hash families contribute a factor $n^{o(1)}$.
Thus computing all scores takes $n^{\omega(1,1,\bar\kappa_s)+o(1)}=n^{\omega(1,1,\kappa_s+1/2)+o(1)}$ time.
Lemma~\ref{lem:seeded-rooted-derandomization} derandomizes color coding for each pair.
For each pair, put $N_B=|[n]\setminus B|$ and form $V_3(B,Q)$ by selecting the vertices with the
\[
M_{B,k}=\min\left\{N_B,\left\lceil k+\frac{N_B}{(\log N_B)^3}\right\rceil\right\}
\]
largest raw scores, breaking ties by a fixed deterministic rule.

For the pair $(B^\star,Q^\star)$, put $N_\star=|V_1|$ and $K=|(C\cap V_1)\setminus Q^\star|$.
Since $Q^\star\subseteq C\cap V_1$, with probability $1-o(1)$ over the edges between $Q^\star$ and the other vertices of $V_1$,
\[
K=(1-o(1))k,
\qquad
|W_{B^\star,Q^\star}|
=(1+o(1))\bigl(2^{-s}N_\star+(1-2^{-s})K\bigr).
\]
Since $N_\star\ge n/3$, we have $|W_{B^\star,Q^\star}|=\Theta(n)=n^{\Omega(1)}$.
The function $x\mapsto x/\sqrt{2^{-s}N_\star+(1-2^{-s})x}$ is increasing.
Together with $N_\star\le n$ and $K\ge(1-o(1))\lambda\sqrt n$, this gives effective signal at least $(1-o(1))\lambda2^{s/2}$, which exceeds $\exp(-1/2)$ by a constant margin.
Choose a fixed $\eta>0$ smaller than this margin.
Lemma~\ref{lem:seeded-rooted-derandomization} and the argument for selecting the largest scores in Lemma~\ref{lem:good-seeded-row} give a set $V_3(B^\star,Q^\star)\subseteq V_1$ containing all but $o(k)$ vertices of $C\cap V_1$ and satisfying
\begin{align*}
|V_3(B^\star,Q^\star)|&=M_{B^\star,k}
\le k+\frac{N_\star}{(\log N_\star)^3}+1,
\\
\frac{k^2}{M_{B^\star,k}}&=\Omega((\log n)^3).
\end{align*}

The scores use only edges inside $V_1$, so $V_3(B,Q)$ is independent of the edges between $V_2$ and $V_1$ conditional on $C$.
The set $V_2$ contains at least $h\gg\log n$ clique vertices.
For each pair $(B,Q)$, define
\[
C_2(B,Q)
=
\left\{
y\in B:
|E(\{y\},V_3(B,Q))|
\ge
\frac{|V_3(B,Q)|}{2}+0.2499k
\right\}.
\]
Although Lemma~\ref{lem:implementable-threshold} is stated for the random split in Section~\ref{sec:recovery}, its proof only uses the independence of the bipartite edges, $|V_2|\le n$, $|V_3\cap C_1|=(1-o(1))k$, and $k^2/|V_3|\gg\log n$.
These properties hold here, so the same proof of Lemma~\ref{lem:implementable-threshold} gives $C_2(B^\star,Q^\star)=C\cap V_2$ with probability $1-o(1)$.
Lemma~\ref{lem:recover-main-from-holdout} then recovers $C\cap V_1$ as the vertices adjacent to all of $C_2(B^\star,Q^\star)$.

Three products with the adjacency matrix clean up all candidates and check whether they are $k$-cliques.
Let $B_{\mathrm{cand}}$ be the indicator matrix of the sets $V_3(B,Q)$, with rows indexed by $(B,Q)$ and columns by vertices.
The product $B_{\mathrm{cand}}A$ counts each vertex's neighbors in each $V_3(B,Q)$.
For each pair, apply the threshold above to the coordinates in $V_2=B$ to form $C_2(B,Q)$.
Let $B_{\mathrm{hold}}$ be the indicator matrix of these sets.
The product $B_{\mathrm{hold}}A$ counts each vertex's neighbors in each $C_2(B,Q)$.
Thresholding at $|C_2(B,Q)|$ within $V_1$ gives
\[
C_1(B,Q)
=
\{x\in [n]\setminus B:\ |E(\{x\},C_2(B,Q))|=|C_2(B,Q)|\}.
\]
Let $B_{\mathrm{fin}}$ be the indicator matrix of the final candidates $C_1(B,Q)\sqcup C_2(B,Q)$.
The product $B_{\mathrm{fin}}A$ counts each vertex's neighbors in each final candidate.
Accept a row exactly when its candidate has size $k$ and every selected coordinate of $B_{\mathrm{fin}}A$ equals $k-1$.
All three products have at most $n^{\bar\kappa_s+o(1)}$ rows and cost $n^{\omega(1,1,\bar\kappa_s)+o(1)}$ time.
Forming $B_{\mathrm{fin}}$ and scanning the output costs $O(n^{\bar\kappa_s+1+o(1)})$, which is dominated by this bound since the output size gives $\omega(1,1,\bar\kappa_s)\ge\bar\kappa_s+1$.
Thus the total runtime is
\[
n^{\omega(1,1,\bar\kappa_s)+o(1)}
=
n^{\omega(1,1,\kappa_s+1/2)+o(1)}.
\]
The pair $(B^\star,Q^\star)$ produces $C$ with probability $1-o(1)$.
The first moment argument in Section~\ref{sec:mm-rectangular} shows that, with probability $1-o(1)$ over the background \Erdos--\Renyi\ edges, $C$ is the unique $k$-clique, uniformly for every fixed $C$ with $k\ge\lambda\sqrt n$.
Thus every accepted candidate equals $C$ on this event.
Accepting if any candidate is a $k$-clique also gives strong detection: under $\QQ_n$ no such clique exists with high probability, while under the planted law the pair $(B^\star,Q^\star)$ produces one with probability $1-o(1)$.
This proves the bound for $\tau_{\det}$ using seeded trees and rectangular matrix multiplication.

For the basic bounds using patterns of bounded treewidth and the bound using square matrix multiplication for treewidth $2$, exact recovery for every clique location still requires a deterministic holdout split.
Replace random colorings by balanced hash families and apply Lemma~\ref{lem:universal-holdout-wrapper} to the raw rooted scores.
Lemma~\ref{lem:balanced-hash-rooted} gives the required mean square bounds inside and outside each fixed induced clique, with $\delta_N=\exp(-\Omega(\sqrt{\log N}))$.
Selecting the vertices with the largest scores does not require knowing the common positive normalizing factor, since it preserves their order.
This gives, respectively,
\begin{align*}
\tau_{\det}(c(t)^{-1/2})
&\le
t+\frac32,\\
\tau_{\det}(1/\sqrt{c(2)})
&\le
\omega+\frac12.
\end{align*}

For detection at the threshold for treewidth $2$, Lemma~\ref{lem:balanced-hash-detection} replaces the random colorings by a balanced hash family, without a holdout split.
Combined with the dynamic program from Theorem~\ref{thm:tw2-mm}, it gives a deterministic algorithm for strong detection in time $n^{\omega+o(1)}$, for every clique location and every fixed $\lambda>1/\sqrt{c(2)}$.

The deterministic message-passing algorithm of Deshpande and Montanari~\cite{deshpande2015finding} gives the sharper recovery bound $\tau_{\det}(\exp(-1/2))\le2$.

Using the ordinary color-coding dynamic program for trees costs $n^{2+o(1)}$ time per seed.
For at most $n^{\bar\kappa_s+o(1)}$ pairs $(B,Q)$, the total cost is $n^{2+\bar\kappa_s+o(1)}=n^{5/2+\kappa_s+o(1)}$.
The same batched or direct cleanup costs no more, giving recovery exponent $5/2+\kappa_s$.
\end{proof}
\end{document}